\documentclass[a4paper, 10pt]{iopart}

\expandafter\let\csname equation*\endcsname\relax
\expandafter\let\csname endequation*\endcsname\relax

\usepackage{amssymb}
\usepackage{amsmath}
\usepackage{amsthm}
\usepackage{latexsym}
\usepackage{mathrsfs}
\usepackage{eucal}
\usepackage{hyperref}

\allowdisplaybreaks[2]

\theoremstyle{plain}
\newtheorem{theorem}{Theorem}
\newtheorem{proposition}[theorem]{Proposition}
\newtheorem{lemma}[theorem]{Lemma}
\newtheorem{corollary}[theorem]{Corollary}

\newtheorem{definition}[theorem]{Definition}
\newtheorem{remark}[theorem]{Remark}

\def\sDiv{\mathscr{D}}
\def\sCurl{\mathscr{C}}
\def\sCurlDagger{\mathscr{C}^\dagger}
\def\sTwist{\mathscr{T}}
\def\ObstrZero{\mathcal{O}_{2,2}^{(0)}}
\def\ObstrOne{\mathcal{O}_{2,2}^{(1)}}

\numberwithin{equation}{section}

\begin{document}

\title{Second order symmetry operators}

\author{Lars Andersson$^1$, Thomas B\"ackdahl$^2$, and Pieter Blue$^2$}
\eads{\mailto{lars.andersson@aei.mpg.de}, \mailto{t.backdahl@ed.ac.uk}, \mailto{p.blue@ed.ac.uk}}
\address{$^1$ Albert Einstein Institute, Am M\"uhlenberg 1, D-14476 Potsdam, Germany\\
$^2$  The University of Edinburgh, James Clerk Maxwell Building, Mayfield Road, Edinburgh,  EH9 3JZ, UK}

\begin{abstract}
Using systematic calculations in spinor language, we obtain simple
descriptions of the second order symmetry operators for the conformal wave equation, 
the Dirac-Weyl equation and the Maxwell equation on a curved four dimensional Lorentzian manifold. 
The conditions for existence of symmetry operators for the different equations are seen to be related. 
Computer algebra tools have been developed and used to systematically reduce the equations to a form which allows geometrical interpretation.
\end{abstract}

\medskip
\noindent

\section{Introduction}
The discovery by Carter \cite{carter:1968PhRv..174.1559C} 
of a fourth constant of the motion for the geodesic equations in the Kerr
black hole spacetime, allowing the geodesic equations to be integrated, 
together with the subsequent discovery by
Teukolsky, Chandrasekhar and others of 
the separability of the spin-$s$ equations for all half-integer spins up to
$s=2$ (which corresponds to the case of linearized Einstein equations) 
in the Kerr geometry, 
provides an essential tool for the analysis of fields in the Kerr geometry. 
The geometric fact behind the existence of Carter's constant is, as shown by Walker and Penrose \cite{PenroseWalker:1970}, the existence of a Killing tensor.
A Killing tensor is a symmetric tensor 
$K_{ab} = K_{(ab)}$, satisfying the equation $\nabla_{(a} K_{bc)} = 0$.
This condition implies that the quantity $K =  K_{ab} \dot \gamma^a \dot \gamma^b$ is constant along affinely parametrized geodesics. In particular, viewed as a function on phase space, $K$ Poisson commutes with the Hamiltonian generating the geodesic flow, $H = \dot \gamma^a \dot \gamma_a$. 

Carter further showed that in a Ricci flat spacetime with a Killing tensor $K_{ab}$, the operator ${\mathcal K} = \nabla_a K^{ab} \nabla_b$, which may be viewed as the ``quantization'' of $K$, commutes with the d'Alembertian $\mathcal H = \nabla^a \nabla_a$, which in turn is the ``quantization'' of $H$, 
cf. \cite{carter:1977PhRvD..16.3395C}. In particular, the operator $\mathcal K$ is a symmetry operator
for the wave equation $\mathcal H \phi = 0$, in the sense that it maps solutions to solutions. The properties of separability, and existence of symmetry operators, for partial differential equations are closely related  \cite{miller:MR0460751}. In fact, specializing to the Kerr geometry,  the symmetry operator found by Carter may be viewed as the spin-0 case of the symmetry operators for the higher spin fields as manifested in the Teukolsky system, see eg.  \cite{kalnins:miller:williams:1989JMP....30.2360K,kalnins:miller:williams:1989JMP....30.2925K}. 

In this paper we give necessary and sufficient conditions for the existence of second order symmetry operators, for massless test fields of spin 0, 1/2, 1, on a globally hyperbolic
Lorentzian spacetime of dimension 4. (As explained in Section~\ref{sec:indepspinors}, the global hyperbolicity condition can be relaxed.)
In each case, the conditions are the existence of a conformal Killing tensor or Killing spinor, and certain auxiliary conditions relating the Weyl curvature and the Killing tensor or spinor. 
We are particularly interested in symmetry operators for the spin-1 or Maxwell equation. In this case, we give a single auxiliary condition, which is substantially more transparent than the collection previously given in~\cite{KalMcLWil92a}. For the massless spin-1/2 or Dirac-Weyl equation,
our result on second order symmetry operators represents a simplification of the conditions given by McLenaghan, Smith and Walker \cite{mclenaghan:smith:walker:2000RSPSA.456.2629G} for the existence of symmetry operators of order two.
The conditions we find for spins 1/2 and 1 are closely related to the condition found recently for the spin-0 case for the conformal wave equation by 
Michel, Radoux and {\v S}ilhan
\cite{michel:radoux:silhan:2013arXiv1308.1046M}, cf. Theorem \ref{thm:intro:spin0} below. 
 
A major motivation for the work in this paper is provided by the application by two of the authors \cite{andersson:blue:2009arXiv0908.2265A} of the Carter symmetry operator for the wave equation in the Kerr spacetime, to prove an integrated energy estimate and
boundedness for solutions of the wave equation. The method used is a
generalization of the vector fields method \cite{klainerman:MR784477} 
to allow not only Killing vector
symmetries but symmetry operators of higher order. In order to apply such
methods to fields with non-zero spin, such as the Maxwell field, it is
desirable to have a clear understanding of the conditions for the existence
of symmetry operators and their structure. This serves as one of the main
motivations for the results presented in this paper, which give simple
necessary and sufficient conditions for the existence of symmetry operators
for the Maxwell equations in a 4-dimensional Lorentzian spacetime. 

The energies constructed from higher order symmetry operators correspond to conserved currents which are not generated by contracting the stress energy tensor with a conformal Killing vector. 
Such conserved currents are known to
exist eg. for the Maxwell equation, as well as fields with higher spin 
on Minkowski space, see
\cite{anco:pohjanpelto:2003RSPSA.459.1215A} and references therein. 
In a subsequent paper
\cite{andersson:backdahl:blue:currents} we shall present a detailed study of conserved currents up to second order for the Maxwell field. 

We will assume that all objects are smooth, we work in signature $(+,-,-,-)$, 
and we use the 2-spinor formalism, following the conventions and notation of \cite{PenRin84, PenRin86}.
For a translation to the Dirac 4-spinor notation, we refer to \cite[Page 221]{PenRin84}.
Recall that $\Lambda/24$ is the scalar curvature, $\Phi_{ABA'B'}$ the Ricci spinor, and $\Psi_{ABCD}$ the Weyl spinor.
Even though several results are independent of the existence of a spin structure, we will for simplicity assume that the spacetime is spin.
The 2-spinor formalism allows one to efficiently decompose spinor expressions into irreducible parts. 
All irreducible parts of a spinor are totally
symmetric spinors formed by taking traces of the spinor and symmetrizing all
free indices. Making use of these facts,  any spinor expression can be decomposed in terms of symmetric
spinors and spin metrics. This procedure is described in detail in Section 3.3 in \cite{PenRin84} and in particular by Proposition 3.3.54. 

This decomposition has been implemented in the package \emph{SymManipulator} \cite{Bae11a} by the second author. \emph{SymManipulator} is part of the \emph{xAct} tensor algebra package \cite{xAct} for \emph{Mathematica}. 
The package \emph{SymManipulator} includes many canonicalization and simplification steps to make the resulting expressions compact enough and the calculations rapid enough so that fairly large problems can be handled. A Mathematica 9 notebook file containing the main calculations for this paper is available as supplementary data at \url{http://hdl.handle.net/10283/541}.

We shall in this paper consider only massless spin-$s$ test fields. 
For the spin-$0$ case the field equation is the conformal wave equation
\begin{equation}\label{eq:confwave1}
(\nabla^a \nabla_a + 4 \Lambda) \phi = 0,
\end{equation}
for a scalar field $\phi$, 
while for non-zero spin the field is a symmetric spinor 
$\phi_{A \cdots F}$ of valence $(2s,0)$ satisfying the equation
\begin{equation}\label{eq:massless}
\nabla^A{}_{A'} \phi_{A \cdots F} = 0 .
\end{equation} 
In this paper we shall restrict our considerations to spins $0,1/2,1$. For $s \geq 3/2$, equation \eqref{eq:massless} implies algebraic consistency conditions, which strongly restrict the space of solutions in the presence of non-vanishing Weyl curvature. Note however
that there are consistent equations for fields of higher spin, see \cite[\S 5.8]{PenRin84} for discussion. 

Recall that a Killing spinor of valence $(k,l)$ is a symmetric spinor
$L_{A_1\cdots A_k}{}^{A'_1 \cdots A'_l}$,
\begin{equation}\label{eq:KSgeneral}
\nabla_{(A_1}{}^{(A'_1} L_{A_2\cdots A_{k+1})}{}^{A'_2 \cdots A'_{l+1})} = 0.
\end{equation}
A valence $(1,1)$ Killing spinor is simply a conformal Killing vector, while a valence $(2,0)$ Killing spinor is equivalent to a conformal Killing-Yano 2-form. On the other hand, a Killing spinor of valence $(2,2)$ is simply a traceless symmetric conformal Killing tensor.  
It is important to note that \eqref{eq:confwave1}, \eqref{eq:massless} and \eqref{eq:KSgeneral} are conformally invariant if $\phi$ and $\phi_{A \cdots F}$ are given conformal weight $-1$, and $L^{A_1\cdots A_kA'_1 \cdots A'_l}$ is given conformal weight $0$. See \cite[sections 5.7 and 6.7]{PenRin84} for details.

\newcommand{\Kcal}{\mathcal K}
\newcommand{\Fcal}{\mathcal F} 
\newcommand{\Hcal}{\mathcal H}
Recall that a symmetry operator for a system $\Hcal\varphi=0$, is a linear partial differential operator $\Kcal$ such that $\Hcal\Kcal\varphi=0$ for all $\varphi$ such that $\Hcal\varphi=0$. 
We say that two operators $\Kcal_1$ and $\Kcal_2$ are equivalent if $\Kcal_1-\Kcal_2=\Fcal\Hcal$ for some differential operator $\Fcal$. We are interested only in \emph{non-trivial} symmetry operators, i.e. operators which are not equivalent to the trivial operator $0$. For simplicity, we will only consider equivalence classes of symmetry operators.

To state our main results, we need two auxiliary conditions. 
\begin{definition} \label{def:auxcond}
Let $L_{AB}{}^{A'B'}$ be a Killing spinor of valence $(2,2)$.
\begin{enumerate} \setcounter{enumi}{-1}
\renewcommand{\theenumi}{A\arabic{enumi}}
\renewcommand{\labelenumi}{\theenumi)} 
\item \label{point:A0}  
$L_{AB}{}^{A'B'}$ satisfies 
auxiliary condition \ref{point:A0}
if there is a function $Q$ such that
\begin{align}
\nabla_{AA'}Q={}&\tfrac{1}{3} \Psi_{ABCD} \nabla^{(B|B'|}L^{CD)}{}_{A'B'}
 + \tfrac{1}{3} \bar\Psi_{A'B'C'D'} \nabla^{B(B'}L_{AB}{}^{C'D')} \nonumber\\&
 + L^{BC}{}_{A'}{}^{B'} \nabla_{(A}{}^{C'}\Phi_{BC)B'C'}
 + L_{A}{}^{BB'C'} \nabla^{C}{}_{(A'}\Phi_{|BC|B'C')} . 
\label{eq:A0}
\end{align}
\item \label{point:A1} 
$L_{AB}{}^{A'B'}$ satisfies 
auxiliary condition \ref{point:A1} if there is a vector field $P_A{}^{A'}$ such that 
\begin{align}
\nabla_{(A}{}^{(A'}P_{B)}{}^{B')}= {}&L{}^{CDA'B'} \Psi_{ABCD}
 - L{}_{AB}{}^{C'D'} \bar\Psi^{A'B'}{}_{C'D'} . 
\label{eq:A1}
\end{align}
\end{enumerate}
\end{definition}
\begin{remark}
Under conformal transformations such that $L^{ABA'B'}$, $P^{AA'}$ and $Q$ are given conformal weight $0$, the equations \eqref{eq:A0} and \eqref{eq:A1} are conformally invariant.
\end{remark}

We start by recalling the result of Michel et al. for the spin-$0$ conformal wave equation, which we state here in the case of a Lorentzian spacetime of dimension 4. 

\begin{theorem}[{\cite[Theorem~4.8]{michel:radoux:silhan:2013arXiv1308.1046M}}] \label{thm:intro:spin0}
Consider the conformal wave equation  
\begin{equation}\label{eq:confwave}
(\nabla^a \nabla_a + 4 \Lambda ) \phi = 0
\end{equation}
in a 4-dimensional Lorentzian spacetime. 
There is a non-trivial second order symmetry operator for \eqref{eq:confwave}
if and only if there is a non-zero Killing spinor of valence $(2,2)$ satisfying condition \ref{point:A0} of Definition \ref{def:auxcond}. 
\end{theorem}
Previous work on the conformal wave equation was done by  
 \cite{kamran:mclenaghan:1985LMaPh...9...65K}, see also Kress
  \cite{kress:thesis}, see also \cite{kalnins:miller:1983JMP....24.1047K}. Symmetry operators of general order for the Laplace-Beltrami operator in the conformally flat case have been analyzed by Eastwood \cite{eastwood:MR2180410}. 

Next we consider fields with spins 1/2 and 1. The massless spin-1/2 equations are  
\begin{subequations}\label{eq:spin1/2eqs}
\begin{align} 
\nabla^A{}_{A'} \phi_A &= 0, \label{eq:spin1/2left} \\
\intertext{and its complex conjugate form} 
\nabla_A{}^{A'} \chi_{A'} &= 0, \label{eq:spin1/2right}
\end{align}
\end{subequations}
which we shall refer to as the left and right Dirac-Weyl equations \footnote{The use of the terms left and right is explained by noting that spinors of valence $(k,0)$ represent left-handed particles, while spinors of valence $(0,k)$ represent right-handed particles, cf. \cite[\S 5.7]{PenRin84}. The Dirac equation is the equation for massive, charged spin-1/2 fields, and couples the left- and right-handed parts of the field, see \cite[\S 4.4]{PenRin84}. We shall not consider the symmetry operators for the Dirac equation here.}.
Analogously with the terminology used by  Kalnins et al. 
\cite{KalMcLWil92a} for the spin-1 case, we call a symmetry operator $\phi_A \mapsto \lambda_A$, which takes a solution of the left equation to a solution of the left equation a symmetry operator of the \emph{first kind}, while an operator $\phi_A \mapsto \chi_{A'}$  which takes a solution of the left equation to a solution of the right equation a symmetry operator of the \emph{second kind}.  

If one considers symmetry operators in the Dirac 4-spinor notation, a 4-spinor would correspond to a pair of 2-spinors $(\phi_A, \varphi_{A'})$. Therefore a symmetry  operator $(\phi_A, \varphi_{A'})\mapsto (\lambda_A, \chi_{A'})$ for a 4-spinor is formed by a combination of symmetry operators of first $\phi_A\mapsto \lambda_{A}$, and second $\phi_A\mapsto \chi_{A'}$ kind, together with complex conjugate versions of first $\varphi_{A'}\mapsto \chi_{A'}$, and second $\varphi_{A'}\mapsto \lambda_A$ kind symmetry operators.

\begin{theorem} \label{thm:intro:spin1/2} 
Consider the Dirac-Weyl equations \eqref{eq:spin1/2eqs} in a Lorentzian spacetime of dimension~4. 
\begin{enumerate} 
\item There is a non-trivial second order symmetry operator of the first kind for the Dirac-Weyl equation 
if and only if there is a non-zero Killing spinor of valence $(2,2)$ satisfying auxiliary conditions \ref{point:A0} and \ref{point:A1} of definition \ref{def:auxcond}. 
\item There is a non-trivial second order symmetry operator of the second kind for the Dirac-Weyl equation if and only if there is a non-zero Killing spinor $L_{ABC}{}^{A'}$ of valence $(3,1)$, such that the auxiliary condition 
\begin{align}
0={}&\tfrac{3}{4} \Psi_{ABCD} \nabla^{FA'}L^{CD}{}_{FA'}
 + \tfrac{5}{6} \Psi_{B}{}^{CDF} \nabla_{(A}{}^{A'}L_{CDF)A'}\nonumber\\&
 + \tfrac{5}{6} \Psi_{A}{}^{CDF} \nabla_{(B}{}^{A'}L_{CDF)A'}
 -  \tfrac{3}{5} L_{B}{}^{CDA'} \nabla_{(A}{}^{B'}\Phi_{CD)A'B'}\nonumber\\&
 -  \tfrac{3}{5} L_{A}{}^{CDA'} \nabla_{(B}{}^{B'}\Phi_{CD)A'B'}
 + \tfrac{4}{3} L^{CDFA'} \nabla_{(A|A'|}\Psi_{BCDF)}.
\label{eq:A1/2*} 
\end{align}
is satisfied. 
\end{enumerate} 
\end{theorem} 
\begin{remark}
\begin{enumerate}
\item{Under conformal transformations such that $\widehat L^{ABCA'}=L^{ABCA'}$, the equation \eqref{eq:A1/2*} is conformally invariant. }
\item{We remark that the auxiliary condition \ref{point:A0}, appears both in Theorem~\ref{thm:intro:spin1/2}, and for the conformal wave equation in Theorem~\ref{thm:intro:spin0}.}
\end{enumerate}
\end{remark}

In previous work, 
Benn and Kress \cite{benn:kress:2004CQGra..21..427B} showed that a first order symmetry operator of the second kind for the Dirac equation exists exactly when there is a valence $(2,0)$ Killing spinor. See also Carter and McLenaghan \cite{carter:mclenaghan:1979PhRvD..19.1093C}, and Durand, Lina, and Vinet 
\cite{durand:lina:vinet:1988PhRvD..38.3837D} for earlier work. The conditions for the existence of a second order symmetry operator for the Dirac-Weyl equations in a general spacetime were considered in \cite{mclenaghan:smith:walker:2000RSPSA.456.2629G},
see also \cite{fels:kamran:1990RSPSA.428..229F}.
The conditions derived here represent a simplification of the conditions found in \cite{mclenaghan:smith:walker:2000RSPSA.456.2629G}. Further, we mention that 
symmetry operators of general order for the Dirac operator on Minkowski space have been analyzed by Michel \cite{michel:2014arXiv1208.4052}. 

For the spin-1 case, we similarly have the left and right Maxwell equations 
\begin{subequations}\label{eq:spin1eqs}
\begin{align}
\nabla^B{}_{A'} \phi_{AB} &= 0, \label{eq:spin1left} \\ 
\nabla_A{}^{B'} \chi_{A'B'} &= 0 \label{eq:spin1right}
\end{align}
\end{subequations}
The left-handed and right-handed spinors $\phi_{AB}$, $\chi_{A'B'}$ represent an anti-self-dual and a self-dual 2-form, respectively. Each equation in \eqref{eq:spin1eqs} is thus equivalent to a real Maxwell equation, cf. \cite[\S 3.4]{PenRin84}. Analogously to the spin-1/2 case, we consider second order symmetry operators of the first and second kind.
\begin{theorem} \label{thm:intro:spin1}
Consider the Maxwell equations \eqref{eq:spin1eqs} in a Lorentzian spacetime of dimension 4. 
\begin{enumerate} 
\item \label{point:first:maxwell} There is a non-trivial second order symmetry operator of the first kind for the Maxwell equation if and only if there is a non-zero Killing spinor of valence $(2,2)$ such that the auxiliary condition \ref{point:A1} of definition \ref{def:auxcond} is satisfied. 
\item \label{point:second:maxwell}
There is a non-trivial second order symmetry operator of the second kind for the Maxwell equation if and only if there is a non-zero Killing spinor $L_{ABCD}$ of valence $(4,0)$.
\end{enumerate} 
\end{theorem} 
Note that no auxiliary condition is needed in point \emph{(\ref{point:second:maxwell})} of Theorem \ref{thm:intro:spin1}. 
The conditions for the existence of second order symmetry operators for the Maxwell equations have been given in previous work by Kalnins, McLenaghan and Williams 
\cite{KalMcLWil92a}, see also  \cite{kalnins:mclenaghan:williams:1992grra.conf..129K}, following earlier work by 
Kalnins, Miller and Williams \cite{kalnins:miller:williams:1989JMP....30.2360K},  see also \cite{kress:thesis}. In \cite{KalMcLWil92a}, the conditions for a second order symmetry operator of the second kind were analyzed completely, and agree with the condition given in point \emph{(\ref{point:second:maxwell})} of Theorem~\ref{thm:intro:spin1}. However, the conditions for a second order symmetry operator of the first kind stated there consist of a set of five equations, of a not particularly transparent nature. The result given here in point \emph{(\ref{point:first:maxwell})} of Theorem \ref{thm:intro:spin1} provides a substantial simplification and clarification of this previous result. 

The necessary and sufficient conditions given in theorems \ref{thm:intro:spin0}, \ref{thm:intro:spin1/2}, \ref{thm:intro:spin1} involve the existence of a Killing spinor and auxiliary conditions. The following result gives examples of Killing spinors for which the auxiliary conditions \ref{point:A0}, \ref{point:A1} and equation \eqref{eq:A1/2*} are satisfied.
\begin{proposition}\label{prop:factorize}
Let $\xi^{AA'}$ and $\zeta^{AA'}$ be (not necessarily distinct) conformal Killing vectors and let $\kappa_{AB}$ be a Killing spinor of valence $(2,0)$. 
\begin{enumerate} 
\item The symmetric spinor $\xi_{(A}{}^{(A'} \zeta_{B)}{}^{B')}$ is a Killing spinor of valence $(2,2)$, which admits solutions to the auxiliary conditions \ref{point:A0} and \ref{point:A1}. 
\item The symmetric spinor $\kappa_{AB}\bar{\kappa}_{A'B'}$ is also a Killing spinor of valence $(2,2)$, which admits solutions to the auxiliary conditions \ref{point:A0} and \ref{point:A1}. 
\item The spinor $\kappa_{(AB}\xi_{C)}{}^{C'}$ is a Killing spinor of valence $(3,1)$, which satisfies auxiliary equation \eqref{eq:A1/2*}. 
\item 
The spinor $\kappa_{(AB}\kappa_{CD)}$ is a Killing spinor of valence $(4,0)$. \label{point:factoredvalence4KS}
\end{enumerate}
\end{proposition}
The point \emph{(\ref{point:factoredvalence4KS})} is immediately clear. The other parts will be proven in Section~\ref{sec:factorizations}.

We now consider the following condition
\begin{align} 
0={}&\Psi_{(ABC}{}^{F}\Phi_{D)FA'B'}. \label{eq:intro:AlignmentPsiPhi}
\end{align}
relating the Ricci curvature $\Phi_{ABA'B'}$ and the Weyl curvature $\Psi_{ABCD}$. A spacetime where \eqref{eq:intro:AlignmentPsiPhi} holds will be said to satisfy the aligned matter condition.
In particular this holds in Vacuum and in the Kerr-Newman class of spacetimes.
Under the aligned matter condition we can show that the converse of Proposition~\ref{prop:factorize} part \emph{(\ref{point:factoredvalence4KS})} is true. The following theorem will be proved in Section~\ref{sec:factorValence4}.
\begin{theorem}\label{thm:Valence4Factorization}
If the aligned matter condition \eqref{eq:intro:AlignmentPsiPhi} is satisfied, $\Psi_{ABCD}\neq 0$ and $L_{ABCD}$ is a valence $(4,0)$ Killing spinor, then there is a valence $(2,0)$ Killing spinor $\kappa_{AB}$ such that
\begin{align}
L_{ABCD}={}&\kappa_{(AB}\kappa_{CD)}.
\end{align}
\end{theorem}

\begin{remark}
If $\Psi_{ABCD}=0$, the valence $(4,0)$ Killing spinor will still factor but in terms of valence $(1,0)$ Killing spinors, which then can be combined into valence $(2,0)$ Killing spinors. However, the two factors might be distinct. 
\end{remark}

A calculation shows that if \eqref{eq:intro:AlignmentPsiPhi} holds, $\kappa_{AB}$ is a valence $(2,0)$ Killing spinor, then $\xi^{AA'} = \nabla^{BA'} \kappa^{A}{}_{B}$ is a Killing vector field. Taking this fact into account, we have the following corollary to the results stated above. It tells that generically one can generate a wide variety of symmetry operators from just a single valence $(2,0)$ Killing spinor.
\begin{corollary} \label{cor:sufficient} Consider the massless test fields of spins 0, 1/2 and 1 in a Lorentzian spacetime of dimension 4. Assume that there is Killing spinor $\kappa_{AB}$ (not identically zero) of valence $(2,0)$. Then there are non-trivial second order symmetry operators for the massless spin-$s$ field equations for spins 0 and 1, as well as a non-trivial second order symmetry operator of the first kind for the massless spin-$1/2$ field.

If, in addition, the aligned matter condition \eqref{eq:intro:AlignmentPsiPhi} 
holds, and $\xi_{AA'} = \nabla^{B}{}_{A'} \kappa_{AB}$ is not identically zero,  
then there is also a non-trivial second order symmetry operators of the second kind for the massless spin-$1/2$ field.  
\end{corollary} 
We end this introduction by giving a simple form for symmetry operators for the Maxwell equation, generated from a Killing spinor of valence $(2,0)$.  

\begin{theorem}\label{Thm:SymopMaxwellSimple} 
Let $\kappa_{AB}$ be a Killing spinor of valence $(2,0)$ and let 
\begin{align} 
\Theta_{AB}\equiv{}&-2 \kappa_{(A}{}^{C}\phi_{B)C}.
\end{align} 
Define the potentials
\begin{subequations} \label{eq:intro:ABdef}
\begin{align} 
A_{AA'}={}&\bar{\kappa}_{A'}{}^{B'} \nabla_{BB'}\Theta_{A}{}^{B}
 -  \tfrac{1}{3} \Theta_{A}{}^{B} \nabla_{BB'}\bar{\kappa}_{A'}{}^{B'},\\
B_{AA'}={}&\kappa_{A}{}^{B} \nabla_{CA'}\Theta_{B}{}^{C}
 + \tfrac{1}{3} \Theta_{A}{}^{B} \nabla_{CA'}\kappa_{B}{}^{C} . 
\end{align}
\end{subequations} 
Assume that $\phi_{AB}$ is a solution to the Maxwell equation in a Lorentzian spacetime of dimension~4. 
Let $A_{AA'}, B_{AA'}$ be given by \eqref{eq:intro:ABdef}. Then
\begin{subequations}
\begin{align} 
\chi_{AB}={}&\nabla_{(B}{}^{A'}A_{A)A'},\\
\omega_{A'B'}={}&\nabla^{B}{}_{(A'}B_{|B|B')}
\end{align} 
\end{subequations}
are solutions to the left and right Maxwell equations, respectively. 
\end{theorem} 
The proof can be found in sections \ref{sec:symopfirstmaxwellfact} and \ref{sec:symopsecondmaxwellfact}.
The general form of the symmetry operators for spins 0, 1/2 and 1 is discussed in detail below. 
\begin{remark}
The symmetry operators of the Maxwell equation can in general be written in potential form. See Theorem~\ref{Thm::SymOpFirstKind} and Theorem~\ref{Thm::SymOpSecondKind}.
\end{remark}

The method used in this paper can also be used to show that the symmetry operators $R$-commute with the Dirac and Maxwell equations. Recall that an operator $S$ is said to $R$-commute with a linear PDE $L\phi=0$ if there is an operator $R$ such that $LS=RL$. Even providing a formula for the relevant $R$ operators would require additional notation, so we have omitted this result from this paper.

\subsection*{Overview of this paper} 
In Section~\ref{sec:prel} we define the fundamental operators $\sDiv, \sCurl, \sCurlDagger, \sTwist$ obtained by projecting the covariant derivative of a symmetric spinor on its irreducible parts. These operators are analogues of the Stein-Weiss operators discussed in Riemannian geometry and play a central role in our analysis. We give the commutation properties of these operators, derive the integrability conditions for Killing spinors, and end the section by discussing some aspects of the methods used in the analysis.  Section~\ref{sec:spin0} gives the analysis of symmetry operators for the conformal wave equation. The results here are given for completeness, and agree with those in \cite{michel:radoux:silhan:2013arXiv1308.1046M} for the case of a Lorentzian spacetime of dimesion 4. The symmetry operators for the Dirac-Weyl equation are discussed in Section~\ref{sec:spin1/2} and our results for the Maxwell case are given in Section~\ref{sec:spin1}. Special conditions under which the auxiliary conditions can be solved is discussed in Section~\ref{sec:factorizations}.
Finally, Section~\ref{sec:symopfactored} contains simplified expressions for the symmetry operators for some of the cases discussed in Section~\ref{sec:factorizations}.

\section{Preliminaries} \label{sec:prel} 

\subsection{Fundamental operators} 
Let $S_{k,l}$ denote the vector bundle of symmetric spinors with $k$ unprimed indices and $l$ primed indices. We will call these spinors symmetric valence $(k,l)$ spinors. Furthermore, let $\mathcal{S}_{k,l}$ denote the space of smooth ($C^\infty$) sections of $S_{k,l}$.
\begin{definition}
For any $\varphi_{A_1\dots A_k}{}^{A_{1}'\dots A_{l}'}\in \mathcal{S}_{k,l}$, we define the operators
$\sDiv_{k,l}:\mathcal{S}_{k,l}\rightarrow \mathcal{S}_{k-1,l-1}$,
$\sCurl_{k,l}:\mathcal{S}_{k,l}\rightarrow \mathcal{S}_{k+1,l-1}$,
$\sCurlDagger_{k,l}:\mathcal{S}_{k,l}\rightarrow \mathcal{S}_{k-1,l+1}$ and
$\sTwist_{k,l}:\mathcal{S}_{k,l}\rightarrow \mathcal{S}_{k+1,l+1}$ as
\begin{subequations}
\begin{align}
(\sDiv_{k,l}\varphi)_{A_1\dots A_{k-1}}{}^{A_1'\dots A_{l-1}'}\equiv{}&
\nabla^{BB'}\varphi_{A_1\dots A_{k-1}B}{}^{A_1'\dots A_{l-1}'}{}_{B'},\\
(\sCurl_{k,l}\varphi)_{A_1\dots A_{k+1}}{}^{A_1'\dots A_{l-1}'}\equiv{}&
\nabla_{(A_1}{}^{B'}\varphi_{A_2\dots A_{k+1})}{}^{A_1'\dots A_{l-1}'}{}_{B'},\\
(\sCurlDagger_{k,l}\varphi)_{A_1\dots A_{k-1}}{}^{A_1'\dots A_{l+1}'}\equiv{}&
\nabla^{B(A_1'}\varphi_{A_1\dots A_{k-1}B}{}^{A_2'\dots A_{l+1}')},\\
(\sTwist_{k,l}\varphi)_{A_1\dots A_{k+1}}{}^{A_1'\dots A_{l+1}'}\equiv{}&
\nabla_{(A_1}{}^{(A_1'}\varphi_{A_2\dots A_{k+1})}{}^{A_2'\dots A_{l+1}')}.
\end{align}
\end{subequations}
\end{definition}
\begin{remark}
\begin{enumerate}
\item{
These operators are all conformally covariant, but the conformal weight differs between the operators. See \cite[Section~6.7]{PenRin86} for details.}
\item{The left Dirac-Weyl and Maxwell equations can be written as $(\sCurlDagger_{1,0}\phi)_{A'}=0$ and $(\sCurlDagger_{2,0}\phi)_{AA'}=0$ respectively. Similarly the right equations can be written in terms of the $\sCurl$ operator.
}
\end{enumerate}
\end{remark}
The operator $\sDiv_{k,l}$ only makes sense when $k\geq 1$ and $l\geq 1$. Likewise $\sCurl_{k,l}$ is
defined only if $l \geq 1$ and	$\sCurlDagger_{k,l}$ only if $k\geq 1$. To make a clean presentation, we will use formulae where invalid operators appear for some choices of $k$ and $l$. However, the operators will always be multiplied with a factor that vanishes for these invalid choices of $k$ and $l$.
From the definition it is clear that the complex conjugates of 
$(\sDiv_{k,l}\varphi)$,
$(\sCurl_{k,l}\varphi)$,
$(\sCurlDagger_{k,l}\varphi)$ and
$(\sTwist_{k,l}\varphi)$ are
$(\sDiv_{l,k}\bar\varphi)$,
$(\sCurlDagger_{l,k}\bar\varphi)$,
$(\sCurl_{l,k}\bar\varphi)$ and
$(\sTwist_{l,k}\bar\varphi)$ respectively, with the appropriate indices.

The main motivation for the introduction of these operators is the irreducible decomposition of the covariant derivative of a symmetric spinor field.
\begin{lemma}
For any $\varphi_{A_1\dots A_k}{}^{A_{1}'\dots A_{l}'}\in \mathcal{S}_{k,l}$, we have the irreducible decomposition
\begin{align}
\nabla_{A_1}{}^{A_1'}\varphi{}_{A_2\dots A_{k+1}}{}^{A_2'\dots A_{l+1}'}={}&
(\sTwist_{k,l}\varphi){}_{A_1\dots A_{k+1}}{}^{A_1'\dots A_{l+1}'}\nonumber\\
&-\tfrac{l}{l+1}\bar\epsilon^{A_1'(A_2'}(\sCurl_{k,l}\varphi){}_{A_1\dots A_{k+1}}{}^{A_3'\dots A_{l+1}')}\nonumber\\
&-\tfrac{k}{k+1}\epsilon_{A_1(A_2}(\sCurlDagger_{k,l}\varphi){}_{A_3\dots A_{k+1})}{}^{A_1'\dots A_{l+1}'}\nonumber\\
&+\tfrac{kl}{(k+1)(l+1)}\epsilon_{A_1(A_2}\bar\epsilon^{A_1'(A_2'}(\sDiv_{k,l}\varphi){}_{A_3\dots A_{k+1})}{}^{A_3'\dots A_{l+1}')}.\label{eq:IrrDecGeneralDer}
\end{align}
\end{lemma}
\begin{proof}
It follows from in \cite[Proposition 3.3.54]{PenRin84} that the irreducible decomposition must have this form. The coefficients are then found by contracting indices and partially expanding the symmetries.
\end{proof}
With this notation, the Bianchi system takes the form
\begin{subequations}
\begin{align}
(\sDiv_{2,2} \Phi)_{AA'}={}&-3 (\sTwist_{0,0} \Lambda)_{AA'},\\
(\sCurlDagger_{4,0} \Psi)_{ABCA'}={}&(\sCurl_{2,2} \Phi)_{ABCA'}.
\end{align}
\end{subequations}
In the rest of the paper we will use these equations every time the left hand sides appear in the calculations.

With the definitions above, a Killing spinor of valence $(k,l)$ is an element 
$L_{A\cdots F}{}^{ A'\dots F'} \in \ker \sTwist_{k,l}$, a conformal Killing vector is a Killing spinor of valence $(1,1)$, and a  
trace-less conformal Killing tensor is a Killing spinor of valence $(2,2)$. 
We further introduce the following operators, acting on a valence $(2,2)$ Killing spinor. 
\begin{definition}
For $L_{AB}{}^{ A'B'} \in \ker \sTwist_{2,2}$, define
\begin{subequations}
\begin{align} 
(\ObstrZero L)_A{}^{A'} \equiv {}&\tfrac{1}{3} \Psi_{ABCD} (\sCurl_{2,2} L)^{BCDA'}
 +  L^{BCA'B'} (\sCurl_{2,2} \Phi)_{ABCB'}\nonumber\\
& + \tfrac{1}{3} \bar\Psi^{A'}{}_{B'C'D'} (\sCurlDagger_{2,2} L)_{A}{}^{B'C'D'}
 + L_{A}{}^{BB'C'} (\sCurlDagger_{2,2} \Phi)_{B}{}^{A'}{}_{B'C'}. \\
(\ObstrOne L)_{AB}{}^{A'B'} \equiv {}&L{}^{CDA'B'} \Psi_{ABCD}
 - L{}_{AB}{}^{C'D'} \bar\Psi^{A'B'}{}_{C'D'} 
\end{align}
\end{subequations}
\end{definition}
The operators $\ObstrZero$ and $\ObstrOne$ are the right hand sides of \eqref{eq:A0} and \eqref{eq:A1} in conditions \ref{point:A0} and \ref{point:A1} respectively. They will play an important role in the conditions for the existence of symmetry operators. 

Given a conformal Killing vector $\xi^{AA'}$, we follow \cite[Equations (2) and (15)]{MR2056970}, see also \cite{anco:pohjanpelto:2003RSPSA.459.1215A}, and define a conformally weighted Lie derivative acting on a symmetric valance $(2s,0)$ spinor field as follows 
\begin{definition}
For $\xi^{AA'} \in \ker \sTwist_{1,1}$, and $\varphi_{A_1\dots A_{2s}}\in \mathcal{S}_{2s,0}$, we define
\begin{align}
\hat{\mathcal{L}}_{\xi}\varphi_{A_1\dots A_{2s}}\equiv{}&\xi^{BB'} \nabla_{BB'}\varphi_{A_1\dots A_{2s}}+s \varphi_{B(A_2\dots A_{2s}} \nabla_{A_1)B'}\xi^{BB'}
 + \tfrac{1-s}{4} \varphi_{A_1\dots A_{2s}} \nabla^{CC'}\xi_{CC'}.
\end{align}
\end{definition}
This operator turns out to be important when we describe first order symmetry operators. See Section~\ref{sec:symopsecondDiracfactored} for further discussion.

\subsection{Commutator relations}
Let $\varphi_{A_1\dots A_k}{}^{A_{1}'\dots A_{l}'}\in \mathcal{S}_{k,l}$ and define the standard commutators
\begin{equation}
\square_{AB}\equiv \nabla_{(A|A'|}\nabla_{B)}{}^{A'} \qquad \text{ and }\qquad \square_{A'B'}\equiv \nabla_{A(A'}\nabla^{A}{}_{B')}.
\end{equation}
Acting on spinors, these commutators can always be written in terms of curvature spinors as described  in \cite[Section 4.9]{PenRin84}.

\begin{lemma}\label{lemma:commutators}
The operators $\sDiv$, $\sCurl$, $\sCurlDagger$ and $\sTwist$ satisfies the following commutator relations
\begin{subequations}
\begin{align}
(\sDiv_{k+1,l-1} &\sCurl_{k,l} \varphi){}_{A_1\dots A_{k}}{}^{A_1'\dots A_{l-2}'}\nonumber\\*
 ={}&\tfrac{k}{k+1}(\sCurl_{k-1,l-1} \sDiv_{k,l} \varphi){}_{A_1\dots A_{k}}{}^{A_1'\dots A_{l-2}'}
-\square_{B'C'}\varphi{}_{A_1\dots A_{k}}{}^{A_1'\dots A_{l-2}'B'C'},
\quad k\geq 0, l\geq 2, \label{eq:DivCurl}\\
(\sDiv_{k-1,l+1} &\sCurlDagger_{k,l} \varphi){}_{A_1\dots A_{k-2}}{}^{A_1'\dots A_{l}'}\nonumber\\*
 ={}&\tfrac{l}{l+1}(\sCurlDagger_{k-1,l-1} \sDiv_{k,l} \varphi){}_{A_1\dots A_{k-2}}{}^{A_1'\dots A_{l}'}
-\square_{BC}\varphi{}_{A_1\dots A_{k-2}}{}^{BCA_1'\dots A_{l}'},
\quad k\geq 2, l\geq 0,\label{eq:DivCurlDagger}\\
(\sCurl_{k+1,l+1}&\sTwist_{k,l} \varphi){}_{A_1\dots A_{k+2}}{}^{A_1'\dots A_{l}'}\nonumber\\*
={}&\tfrac{l}{l+1}(\sTwist_{k+1,l-1} \sCurl_{k,l} \varphi){}_{A_1\dots A_{k+2}}{}^{A_1'\dots A_{l}'}
-\square_{(A_1A_2}\varphi{}_{A_3\dots A_{k+2})}{}^{A_1'\dots A_{l}'},
\quad k\geq 0, l\geq 0,\label{eq:CurlTwist}\\
(\sCurlDagger_{k+1,l+1}&\sTwist_{k,l}\varphi){}_{A_1\dots A_{k}}{}^{A_1'\dots A_{l+2}'}\nonumber\\*
={}&
\tfrac{k}{k+1}(\sTwist_{k-1,l+1} \sCurlDagger_{k,l} \varphi){}_{A_1\dots A_{k}}{}^{A_1'\dots A_{l+2}'}
-\square^{(A_1'A_2'}\varphi{}_{A_1\dots A_{k}}{}^{A_3'\dots A_{l+2}')},
\quad k\geq 0, l\geq 0,\label{eq:CurlDaggerTwist}\\
(\sDiv_{k+1,l+1}&\sTwist_{k,l}\varphi)_{A_1\dots A_k}{}^{A_1'\dots A_l'}\nonumber\\*
={}&
-(\tfrac{1}{k+1}+\tfrac{1}{l+1})(\sCurl_{k-1,l+1}\sCurlDagger_{k,l}\varphi)_{A_1\dots A_k}{}^{A_1'\dots A_l'}
+\tfrac{l(l+2)}{(l+1)^2}(\sTwist_{k-1,l-1}\sDiv_{k,l}\varphi)_{A_1\dots A_k}{}^{A_1'\dots A_l'}\nonumber\\*
&-\tfrac{l+2}{l+1}\square^B{}_{(A_1}\varphi_{A_2\dots A_k)B}{}^{A_1'\dots A_l'}
-\tfrac{l}{l+1}\square^{B'(A_1'}\varphi_{A_1\dots A_k}{}^{A_2'\dots A_l')}{}_{B'},
\quad k\geq 1, l\geq 0,\label{eq:DivTwistCurlCurlDagger}\\
(\sDiv_{k+1,l+1}&\sTwist_{k,l}\varphi)_{A_1\dots A_k}{}^{A_1'\dots A_l'}\nonumber\\*
={}&
-(\tfrac{1}{k+1}+\tfrac{1}{l+1})(\sCurlDagger_{k+1,l-1}\sCurl_{k,l}\varphi)_{A_1\dots A_k}{}^{A_1'\dots A_l'}
+\tfrac{k(k+2)}{(k+1)^2}(\sTwist_{k-1,l-1}\sDiv_{k,l}\varphi)_{A_1\dots A_k}{}^{A_1'\dots A_l'}\nonumber\\*
&-\tfrac{k}{k+1}\square^B{}_{(A_1}\varphi_{A_2\dots A_k)B}{}^{A_1'\dots A_l'}
-\tfrac{k+2}{k+1}\square^{B'(A_1'}\varphi_{A_1\dots A_k}{}^{A_2'\dots A_l')}{}_{B'},
\quad k\geq 0, l\geq 1,\label{eq:DivTwistCurlDaggerCurl}\\
(\sCurl_{k-1,l+1}&\sCurlDagger_{k,l}\varphi)_{A_1\dots A_k}{}^{A_1'\dots A_l'}\nonumber\\*
={}&
(\sCurlDagger_{k+1,l-1}\sCurl_{k,l}\varphi)_{A_1\dots A_k}{}^{A_1'\dots A_l'}
+(\tfrac{1}{k+1}-\tfrac{1}{l+1})(\sTwist_{k-1,l-1}\sDiv_{k,l}\varphi)_{A_1\dots A_k}{}^{A_1'\dots A_l'}\nonumber\\*
&-\square_{(A_1}{}^{B}\varphi_{A_2\dots A_k)B}{}^{A_1'\dots A_l'}
+\square^{B'(A_1'}\varphi_{A_1\dots A_k}{}^{A_2'\dots A_l')}{}_{B'},
\quad k\geq 1, l\geq 1.\label{eq:CurlCurlDagger}
\end{align}
\end{subequations}
\end{lemma}

\begin{proof}
We first observe that \eqref{eq:DivCurl} and \eqref{eq:DivCurlDagger} are related by complex conjugation. Likewise \eqref{eq:CurlTwist} and \eqref{eq:CurlDaggerTwist} as well as \eqref{eq:DivTwistCurlCurlDagger} and \eqref{eq:DivTwistCurlDaggerCurl} are also related by complex conjugation. Furthermore, \eqref{eq:CurlCurlDagger} is given by the difference between \eqref{eq:DivTwistCurlCurlDagger} and \eqref{eq:DivTwistCurlDaggerCurl}.  It is therefore enough to prove \eqref{eq:DivCurl},
\eqref{eq:CurlDaggerTwist} and \eqref{eq:DivTwistCurlCurlDagger}. We consider each in turn.

\begin{itemize}
\item{We first prove \eqref{eq:DivCurl}.
We partially expand the symmetry, identify the commutator in one term, and commute derivatives in the other.
\begin{align*}
(\sDiv_{k+1,l-1}& \sCurl_{k,l} \varphi){}_{A_1\dots A_{k}}{}^{A_1'\dots A_{l-2}'}\\
 ={}&\nabla^{BB'}\nabla_{(A_1}{}^{C'}\varphi_{A_2\dots A_k B)}{}^{A_1'\dots A_{l-2}'}{}_{B'C'}\\
 ={}&\tfrac{1}{k+1}\nabla^{B(B'}\nabla_{B}{}^{C')}\varphi_{A_1\dots A_k}{}^{A_1'\dots A_{l-2}'}{}_{B'C'}
+\tfrac{k}{k+1}\nabla^{B(B'}\nabla_{(A_1}{}^{C')}\varphi_{A_2\dots A_k) B}{}^{A_1'\dots A_{l-2}'}{}_{B'C'}\\
 ={}&-\tfrac{1}{k+1}\square^{B'C'}\varphi_{A_1\dots A_k}{}^{A_1'\dots A_{l-2}'}{}_{B'C'}
+\tfrac{k}{k+1}\epsilon^B{}_{(A_1}\square^{B'C'}\varphi_{A_2\dots A_k) B}{}^{A_1'\dots A_{l-2}'}{}_{B'C'}\\
&+\tfrac{k}{k+1}\nabla_{(A_1}{}^{C'}\nabla^{BB'}\varphi_{A_2\dots A_k) B}{}^{A_1'\dots A_{l-2}'}{}_{B'C'}\\
 ={}&
\tfrac{k}{k+1}(\sCurl_{k-1,l-1} \sDiv_{k,l} \varphi){}_{A_1\dots A_{k}}{}^{A_1'\dots A_{l-2}'}
-\square_{B'C'}\varphi{}_{A_1\dots A_{k}}{}^{A_1'\dots A_{l-2}'B'C'}.
\end{align*}
}
\item{
To prove \eqref{eq:CurlDaggerTwist}, we first partially expand the symmetrization over the unprimed indices in the irreducible decomposition \eqref{eq:IrrDecGeneralDer} and symmetrizing over the primed indices. This gives
\begin{align}
\nabla_{A_1}{}^{(A_2'}\varphi_{A_2\dots A_k B}{}^{A_3'\dots A_{l+2}')}={}&
(\sTwist_{k,l}\varphi){}_{A_1\dots A_{k}B}{}^{A_2'\dots A_{l+2}'}
-\tfrac{1}{k+1}\epsilon_{A_1B}(\sCurlDagger_{k,l}\varphi){}_{A_2\dots A_{k}}{}^{A_2'\dots A_{l+2}'}\nonumber\\
&-\tfrac{k-1}{k+1}\epsilon_{A_1(A_2}(\sCurlDagger_{k,l}\varphi){}_{A_3\dots A_{k})B}{}^{A_2'\dots A_{l+2}'}.\label{eq:Irrdecvarphihelp}
\end{align}
Using the definitions of $\sTwist$ and $\sCurlDagger$, commuting derivatives and using \eqref{eq:Irrdecvarphihelp}, we have
\begin{align}
(\sTwist_{k-1,l+1} \sCurlDagger_{k,l} \varphi){}_{A_1\dots A_{k}}{}^{A_1'\dots A_{l+2}'}
={}&
\nabla_{(A_1}{}^{(A_1'}\nabla^{|B|A_2'}\varphi_{A_2\dots A_{k})B}{}^{A_3'\dots A_{l+2}')} \nonumber\\
={}&
\square^{(A_1'A_2'}\varphi{}_{A_1\dots A_{k}}{}^{A_3'\dots A_{l}')}
+\nabla^{B(A_1'}\nabla_{(A_1}{}^{A_2'}\varphi_{A_2\dots A_{k})B}{}^{A_3'\dots A_{l+2}')}\nonumber\\
={}&
\square^{(A_1'A_2'}\varphi{}_{A_1\dots A_{k}}{}^{A_3'\dots A_{l}')}
+\nabla^{B(A_1'}(\sTwist_{k,l}\varphi){}_{A_1\dots A_{k}B}{}^{A_2'\dots A_{l+2}')}\nonumber\\
&-\tfrac{1}{k+1}\epsilon_{(A_1|B|}\nabla^{B(A_1'}(\sCurlDagger_{k,l}\varphi){}_{A_2\dots A_{k})}{}^{A_2'\dots A_{l+2}')}\nonumber\\
={}&\square^{(A_1'A_2'}\varphi{}_{A_1\dots A_{k}}{}^{A_3'\dots A_{l}')}
+(\sCurlDagger_{k+1,l+1}\sTwist_{k,l}\varphi){}_{A_1\dots A_{k}}{}^{A_1'\dots A_{l+2}'}\nonumber\\
&+\tfrac{1}{k+1}(\sTwist_{k-1,l+1}\sCurlDagger_{k,l}\varphi){}_{A_1\dots A_{k}}{}^{A_1'\dots A_{l+2}'}.
\end{align}
Isolating the $\sCurlDagger\sTwist$-term gives \eqref{eq:CurlDaggerTwist}.
}
\item{
Finally to prove \eqref{eq:DivTwistCurlCurlDagger}, we assume $k\geq 1$ and observe 
\begin{align}
(\sDiv_{k+1,l+1}&\sTwist_{k,l}\varphi)_{A_1\dots A_k}{}^{A_1'\dots A_l'}\nonumber\\
={}&-\nabla^B{}_{B'}\nabla_{(B}{}^{(B'}\varphi_{A_1\dots A_k)}{}^{A_1'\dots A_l')}\nonumber\\
={}&
-\tfrac{1}{k+1}\nabla^B{}_{B'}\nabla_{B}{}^{(B'}\varphi_{A_1\dots A_k}{}^{A_1'\dots A_l')}
-\tfrac{k}{k+1}\nabla^B{}_{B'}\nabla_{(A_1}{}^{(B'}\varphi_{A_2\dots A_k)B}{}^{A_1'\dots A_l')}\nonumber\\
={}&
\tfrac{1}{k+1}(\sDiv_{k+1,l+1}\sTwist_{k,l}\varphi)_{A_1\dots A_k}{}^{A_1'\dots A_l'}
-\tfrac{k}{(k+1)^2}(\sCurl_{k-1,l+1}\sCurlDagger_{k,l}\varphi)_{A_1\dots A_k}{}^{A_1'\dots A_l'}\nonumber\\
&-\tfrac{k}{k+1}\nabla^B{}_{B'}\nabla_{(A_1}{}^{(B'}\varphi_{A_2\dots A_k)B}{}^{A_1'\dots A_l')},
\end{align}
where we in the last step used the irreducible decomposition \eqref{eq:IrrDecGeneralDer} on the first term.
We can solve for the $\sDiv\sTwist$-term from which it follows that
\begin{align}
(\sDiv_{k+1,l+1}&\sTwist_{k,l}\varphi)_{A_1\dots A_k}{}^{A_1'\dots A_l'}\nonumber\\
={}&
-\tfrac{1}{k+1}(\sCurl_{k-1,l+1}\sCurlDagger_{k,l}\varphi)_{A_1\dots A_k}{}^{A_1'\dots A_l'}
-\nabla^B{}_{B'}\nabla_{(A_1}{}^{(B'}\varphi_{A_2\dots A_k)B}{}^{A_1'\dots A_l')}\nonumber\\
={}&
-\tfrac{1}{k+1}(\sCurl_{k-1,l+1}\sCurlDagger_{k,l}\varphi)_{A_1\dots A_k}{}^{A_1'\dots A_l'}
-\tfrac{1}{l+1}\nabla^B{}_{B'}\nabla_{(A_1}{}^{B'}\varphi_{A_2\dots A_k)B}{}^{A_1'\dots A_l'}\nonumber\\
&-\tfrac{l}{l+1}\nabla^B{}_{B'}\nabla_{(A_1}{}^{(A_1'}\varphi_{A_2\dots A_k)B}{}^{A_2'\dots A_l')B'}
\nonumber\\
={}&
-\tfrac{1}{k+1}(\sCurl_{k-1,l+1}\sCurlDagger_{k,l}\varphi)_{A_1\dots A_k}{}^{A_1'\dots A_l'}
-\tfrac{1}{l+1}\nabla_{(A_1}{}^{B'}\nabla^{B}{}_{|B'|}\varphi_{A_2\dots A_k)B}{}^{A_1'\dots A_l'}\nonumber\\
&
-\tfrac{2}{l+1}\square^B{}_{(A_1}\varphi_{A_2\dots A_k)B}{}^{A_1'\dots A_l'}
-\tfrac{l}{l+1}\nabla_{(A_1}{}^{(A_1'}\nabla^{|B|}{}_{|B'|}\varphi_{A_2\dots A_k)B}{}^{A_2'\dots A_l')B'}\nonumber\\
&-\tfrac{l}{l+1}\square^B{}_{(A_1}\varphi_{A_2\dots A_k)B}{}^{A_1'\dots A_l'}
-\tfrac{l}{l+1}\square^{B'(A_1'}\varphi_{A_1\dots A_k}{}^{A_2'\dots A_l')}{}_{B'}
\nonumber\\
={}&
-(\tfrac{1}{k+1}+\tfrac{1}{l+1})(\sCurl_{k-1,l+1}\sCurlDagger_{k,l}\varphi)_{A_1\dots A_k}{}^{A_1'\dots A_l'}
+\tfrac{l(l+2)}{(l+1)^2}(\sTwist_{k-1,l-1}\sDiv_{k,l}\varphi)_{A_1\dots A_k}{}^{A_1'\dots A_l'}\nonumber\\
&-\tfrac{l+2}{l+1}\square^B{}_{(A_1}\varphi_{A_2\dots A_k)B}{}^{A_1'\dots A_l'}
-\tfrac{l}{l+1}\square^{B'(A_1'}\varphi_{A_1\dots A_k}{}^{A_2'\dots A_l')}{}_{B'}.
\end{align}
}
\end{itemize}
\end{proof}
\begin{remark}
The operators $\sDiv$, $\sCurl$, $\sCurlDagger$ and $\sTwist$ together with the irreducible decomposition \eqref{eq:IrrDecGeneralDer} and the relations in Lemma~\ref{lemma:commutators} have all been implemented in the \emph{SymManipulator} package version 0.9.0 \cite{Bae11a}.
\end{remark}

\subsection{Integrability conditions for Killing spinors}\label{sec:integrabilitycond}
Here we demonstrate a procedure for obtaining an integrability condition for a Killing spinor of arbitrary valence.
Let $\kappa_{A_1\dots A_k}{}^{A'_1\dots A'_l}\in \ker \sTwist_{k,l}$.
By applying the $\sCurl$ operator $l+1$ times to the Killing spinor equation, and repeatedly commute derivatives with \eqref{eq:CurlTwist} we get
\begin{align}
0={}& (\underbrace{\sCurl_{k+l+1,1}\sCurl_{k+l,2} \cdots \sCurl_{k+2,l}\sCurl_{k+1,l+1}}_{l+1}\sTwist_{k,l}\kappa)_{A_{1}\dots A_{k+l+2}}\nonumber\\
={}&\frac{l}{l+1}(\sCurl_{k+l+1,1}\sCurl_{k+l,2} \cdots \sCurl_{k+2,l}\sTwist_{k+1,l-1}\sCurl_{k,l}\kappa)_{A_{1}\dots A_{k+l+2}}+\text{curvature terms}\nonumber\\
={}&\frac{1}{l+1}(\sCurl_{k+l+1,1}\sTwist_{k+l,0}\sCurl_{k+l-1,1} \cdots \sCurl_{k+1,l-1}\sCurl_{k,l}\kappa)_{A_{1}\dots A_{k+l+2}}+\text{curvature terms}\nonumber\\
={}&\text{curvature terms}.
\end{align}
Here, the curvature terms have $l-m$ derivatives of $\kappa$ and $m$ derivatives of the curvature spinors, where $0\leq m\leq l$. The main idea behind this is the observation that the commutator \eqref{eq:CurlTwist} acting on a spinor field without primed indices only gives curvature terms.
In the same way we can use \eqref{eq:CurlDaggerTwist} to get
\begin{align}
0={}& (\underbrace{\sCurlDagger_{1,k+l+1}\sCurlDagger_{2,k+l} \cdots \sCurlDagger_{k,l+2}\sCurlDagger_{k+1,l+1}}_{k+1}\sTwist_{k,l}\kappa)_{A'_{1}\dots A'_{k+l+2}}\nonumber\\
={}&\text{curvature terms}.
\end{align}

\subsection{Splitting equations into independent parts}\label{sec:indepspinors}
In our derivation of necessary conditions for the existence of symmetry operators, it is crucial that, at each fixed point in spacetime, we can freely choose the values of the Dirac-Weyl and the Maxwell field and of the symmetric components of any given order of their derivatives. The remaining components of the derivatives to a given order, which involve at least one pair of antisymmetrized indices, can be solved for using the field equations or curvature conditions. See sections \ref{sec:ReductionDirac} and \ref{sec:ReductionMaxwell} for detailed expressions. In the literature, the condition that the symmetric components can be freely and independently specified but that no other parts can be is referred to as the exactness of the set of fields \cite[Section~5.10]{PenRin84}. The symmetric components of the derivatives are exactly those that can be expressed in terms of the operator $\sTwist$. One can show that, in a globally hyperbolic spacetime, the Dirac-Weyl and Maxwell fields each form exact sets. However, it is not necessary for the spacetime to be globally hyperbolic for this condition to hold. If the spacetime is such that the fields fail to form an exact set, then our methods still give sufficient conditions for the existence of symmetry operators, but they may no longer be necessary.

The freedom to choose the symmetric components is used in this paper to show that equations of the type $L^{ABA'}(\sTwist_{1,0}\phi)_{ABA'}+M^{A}\phi_{A}=0$ with $(\sCurl_{1,0}^\dagger\phi)_{A'}=0$ forces $L^{(AB)A'}=0$ and $M^A=0$ because $(\sTwist_{1,0}\phi)_{ABA'}$ and $\phi_A$ can be freely and independently specified at a single point. Similar arguments involving derivatives of up to third order are also used. 

In several places we will have equations of the form
\begin{align}\label{eq:ProductExample}
0=S^{ABC}{}_{A'}(\sTwist_{1,0}\phi)_{AB}{}^{A'}T_{C},
\end{align}
where $T_{A}$ and $(\sTwist_{1,0}\phi)_{ABA'}$ are free and independent. In particular 
all linear combinations of the form $(\sTwist_{1,0}\phi)_{AB}{}^{A'}T_{C}$ will then span the space of spinors $W_{ABC}{}^{A'}=W_{(AB)C}{}^{A'}$. As the equation \eqref{eq:ProductExample} is linear we therefore get 
\begin{align}
0=S^{ABC}{}_{A'}W_{ABC}{}^{A'},
\end{align}
for all $W_{ABC}{}^{A'}=W_{(AB)C}{}^{A'}$.  We can then make an irreducible decomposition 
\begin{align}
W_{ABC}{}^{A'}={}&W_{(ABC)}{}^{A'}
 -  \tfrac{2}{3} W_{(A}{}^{D}{}_{|D|}{}^{A'}\epsilon_{B)C},
\end{align}
which gives
\begin{align}
0={}&(- \tfrac{1}{3} S_{B}{}^{C}{}_{CA'}
 -  \tfrac{1}{3} S^{C}{}_{BCA'}) W^{BA}{}_{A}{}^{A'}
 -  S_{ABCA'} W^{(ABC)A'}.
\end{align}
As $W_{ABC}{}^{A'}$ is free, its irreducible components $W_{(ABC)}{}^{A'}$ and $W_{A}{}^{D}{}_{D}{}^{A'}$ are free and independent. We can therefore conclude that 
\begin{subequations}
\begin{align}
0={}&S_{B}{}^{C}{}_{CA'}+ S^{C}{}_{BCA'},\\
0={}&  S_{(ABC)A'}.
\end{align}
\end{subequations}
Observe that we only get the symmetric part in the last equation due to the symmetry of $W_{(ABC)}{}^{A'}$.

Instead of introducing a new spinor $W_{ABC}{}^{A'}$ we will in the rest of the paper work directly with the irreducible decomposition of $(\sTwist_{1,0}\phi)_{AB}{}^{A'}T_{C}$ and get
\begin{align}
0={}& (- \tfrac{1}{3} S_{A}{}^{C}{}_{CA'}
 -  \tfrac{1}{3} S^{C}{}_{ACA'}) T_{B} (\sTwist_{1,0} \phi)^{ABA'}
- S_{ABCA'} T^{(A}(\sTwist_{1,0} \phi)^{BC)A'}.
\end{align}
The formal computations will be the same, and by the argument above, the symmetrized coefficients for the irreducible parts $T_{B} (\sTwist_{1,0} \phi)^{ABA'}$ and $T^{(A}(\sTwist_{1,0} \phi)^{BC)A'}$ will individually have to vanish.

\section{The conformal wave equation} \label{sec:spin0}
For completeness we give here a detailed description of the symmetry operators for the conformal wave equation. 
\begin{theorem}[\cite{michel:radoux:silhan:2013arXiv1308.1046M}]
The equation
\begin{align}
(\square+4 \Lambda) \phi={}&0,
\end{align}
has a symmetry operator $\phi \rightarrow \chi$ , with order less or equal to two, if and only if there are spinors $L_{AB}{}^{A'B'}=L_{(AB)}{}^{(A'B')}$, $P_{AA'}$ and $Q$ such that
\begin{subequations}
\begin{align}
(\sTwist_{2,2} L)_{ABC}{}^{A'B'C'}={}&0,\\
(\sTwist_{1,1} P)_{AB}{}^{A'B'}={}&0,\\
(\sTwist_{0,0} Q)_{A}{}^{A'}={}&\tfrac{2}{5}(\ObstrZero L)_A{}^{A'}.\label{eq:auxcondwave}
\end{align}
\end{subequations}
The symmetry operator then takes the form
\begin{align}
\chi={}&- \tfrac{3}{5} L^{ABA'B'} \Phi_{ABA'B'} \phi
 + Q \phi
 + \tfrac{1}{4} \phi (\sDiv_{1,1} P)
 + \tfrac{1}{15} \phi (\sDiv_{1,1} \sDiv_{2,2} L)
 + P^{AA'} (\sTwist_{0,0} \phi)_{AA'}\nonumber\\
& + \tfrac{2}{3} (\sDiv_{2,2} L)^{AA'} (\sTwist_{0,0} \phi)_{AA'}
 + L^{ABA'B'} (\sTwist_{1,1} \sTwist_{0,0} \phi)_{ABA'B'}.\label{eq:wavesymop1}
\end{align}
\end{theorem}
The existence of $Q$ satisfying \eqref{eq:auxcondwave} is exactly the auxiliary condition \ref{point:A0}.
The proof can also be carried out using the same technique as in the rest of the paper.

\section{The Dirac-Weyl equation} \label{sec:spin1/2} 
The following theorems imply Theorem~\ref{thm:intro:spin1/2}. 
\begin{theorem}\label{Thm::SymOpFirstKindDirac}
There exists a symmetry operator of the first kind for the Dirac-Weyl equation $\phi_{A}\rightarrow \chi_{A}$, with order less or equal to two, if and only if there are spinor fields
$L_{AB}{}^{A'B'}=L_{(AB)}{}^{(A'B')}$, $P_{AA'}$ and $Q$ such that
\begin{subequations}
\begin{align}
(\sTwist_{2,2} L)_{ABC}{}^{A'B'C'}={}&0,\\
(\sTwist_{1,1} P)_{AB}{}^{A'B'}={}&- \tfrac{1}{3} (\ObstrOne L)_{AB}{}^{A'B'},\label{eq:ObstrDiracSecond1}\\
(\sTwist_{0,0} Q)_{A}{}^{A'}={}&\tfrac{3}{10} (\ObstrZero L)_{A}{}^{A'}.\label{eq:ObstrDiracSecond2}
\end{align}
\end{subequations}
The symmetry operator then takes the form
\begin{align}
\chi_{A}={}&- \tfrac{8}{15} L^{BCA'B'} \Phi_{BCA'B'} \phi_{A}
 + Q \phi_{A}
 + \tfrac{1}{2} \phi^{B} (\sCurl_{1,1} P)_{AB}
 + \tfrac{2}{9} \phi^{B} (\sCurl_{1,1} \sDiv_{2,2} L)_{AB}
 + \tfrac{3}{8} \phi_{A} (\sDiv_{1,1} P)\nonumber\\
& + \tfrac{2}{15} \phi_{A} (\sDiv_{1,1} \sDiv_{2,2} L)
 + P^{BA'} (\sTwist_{1,0} \phi)_{ABA'}
 + \tfrac{8}{9} (\sDiv_{2,2} L)^{BA'} (\sTwist_{1,0} \phi)_{ABA'}\nonumber\\
& + \tfrac{2}{3} (\sCurl_{2,2} L)_{ABCA'} (\sTwist_{1,0} \phi)^{BCA'}
 + L^{BCA'B'} (\sTwist_{2,1} \sTwist_{1,0} \phi)_{ABCA'B'}.\label{eq:diracsymop1}
\end{align}
\end{theorem}
\begin{remark}
\begin{enumerate}
\item{
Observe that \eqref{eq:ObstrDiracSecond1} is the auxiliary condition \ref{point:A1} for existence of a symmetry operator of the first kind for Maxwell equation, and \eqref{eq:ObstrDiracSecond2} is the auxiliary condition \ref{point:A0} for existence of a symmetry operator for the conformal wave equation. }
\item{
With $L_{ABA'B'}=0$ the first order operator takes the form
\begin{align}
\chi_{A}={}&
\hat{\mathcal{L}}_{P}\phi_{A}+ Q \phi_{A}.
\end{align}
}
\end{enumerate}
\end{remark}

\begin{theorem}\label{Thm::SymOpSecondKindDirac}
There exists a symmetry operator of the second kind for the Dirac-Weyl equation $\phi_{A}\rightarrow \omega_{A'}$, with order less or equal to two, if and only if there are spinor fields
$L_{ABC}{}^{A'}=L_{(ABC)}{}^{A'}$ and $P_{AB}=P_{(AB)}$ such that
\begin{subequations}
\begin{align}
(\sTwist_{3,1} L)_{ABCD}{}^{A'B'}={}&0,\label{eq:TwistL31}\\
(\sTwist_{2,0} P)_{ABC}{}^{A'}={}&0,\label{eq:TwistPDirac2}\\
0={}&- \tfrac{9}{8} \Psi_{ABCD} (\sDiv_{3,1} L)^{CD}
 + \tfrac{9}{5} L_{(A}{}^{CDA'}(\sCurl_{2,2} \Phi)_{B)CDA'}\nonumber\\
&
 -  \tfrac{5}{2} \Psi_{(A}{}^{CDF}(\sCurl_{3,1} L)_{B)CDF}
 - 2 L^{CDFA'} (\sTwist_{4,0} \Psi)_{ABCDFA'}. \label{eq:DiracSecondKindObstr}
\end{align}
\end{subequations}
The operator takes the form
\begin{align}
\omega_{A'}={}&- \tfrac{1}{2} L_{BCDB'} \Phi^{CD}{}_{A'}{}^{B'} \phi^{B}
 + \tfrac{2}{3} \phi^{B} (\sCurlDagger_{2,0} P)_{BA'}
 + \tfrac{1}{4} \phi^{B} (\sCurlDagger_{2,0} \sDiv_{3,1} L)_{BA'}
 + P^{BC} (\sTwist_{1,0} \phi)_{BCA'}\nonumber\\
& + \tfrac{3}{4} (\sDiv_{3,1} L)^{BC} (\sTwist_{1,0} \phi)_{BCA'}
 + \tfrac{3}{4} (\sCurlDagger_{3,1} L)_{BCA'B'} (\sTwist_{1,0} \phi)^{BCB'}
 + L^{BCDB'} (\sTwist_{2,1} \sTwist_{1,0} \phi)_{BCDA'B'}.
\end{align}
\end{theorem}

\begin{remark}
The scheme for deriving integrability conditions in Section~\ref{sec:integrabilitycond} can be used to show that
\begin{align}
0={}&- \tfrac{2}{5} L_{(ABC}{}^{A'}(\sCurl_{2,2} \Phi)_{DFH)A'}
 + 3 L_{(AB}{}^{LA'}(\sTwist_{4,0} \Psi)_{CDFH)LA'}
 + 5 \Psi_{(ABC}{}^{L}(\sCurl_{3,1} L)_{DFH)L}\nonumber\\
& + \tfrac{3}{4} \Psi_{(ABCD}(\sDiv_{3,1} L)_{FH)},
\end{align}
follows from \eqref{eq:TwistL31}. 
Despite the superficial similarity of this equation to the condition \eqref{eq:DiracSecondKindObstr}, we conjecture that \eqref{eq:DiracSecondKindObstr} does not follow from \eqref{eq:TwistL31}.
\end{remark}

\subsection{Reduction of derivatives of the field}\label{sec:ReductionDirac}
In our notation, the Dirac-Weyl equation $\nabla^{A}{}_{A'}\phi_{A}=0$, takes the form $(\sCurlDagger_{1,0} \phi)_{A'}=0$. We see that the only remaining irreducible part of $\nabla_{A}{}^{A'}\phi_{B}$ is $(\sTwist_{1,0} \phi)_{AB}{}^{A'}$. By commuting derivatives we see that all higher order derivatives of $\phi_{A}$ can be reduced to totally symmetrized derivatives and lower order terms consisting of curvature times lower order symmetrized derivatives.

Together with the Dirac-Weyl equation, the commutators \eqref{eq:DivTwistCurlCurlDagger}, \eqref{eq:CurlTwist}, \eqref{eq:CurlDaggerTwist} give
\begin{subequations}
\begin{align}
(\sDiv_{2,1} \sTwist_{1,0} \phi)_{A}={}&-6 \Lambda \phi_{A},\\
(\sCurl_{2,1} \sTwist_{1,0} \phi)_{ABC}={}&- \Psi_{ABCD} \phi^{D},\\
(\sCurlDagger_{2,1} \sTwist_{1,0} \phi)_{AA'B'}={}&- \Phi_{ABA'B'} \phi^{B}.
\end{align}
\end{subequations}
The higher order derivatives can be computed by using the commutators \eqref{eq:DivTwistCurlCurlDagger}, \eqref{eq:CurlTwist}, \eqref{eq:CurlDaggerTwist} together with the equations above and the Bianchi system to get
\begin{subequations}
\begin{align}
(\sDiv_{3,2} \sTwist_{2,1} \sTwist_{1,0} \phi)_{AB}{}^{A'}={}&\tfrac{5}{6} \phi^{C} (\sCurl_{2,2} \Phi)_{ABC}{}^{A'}
 + \tfrac{10}{3} \Phi_{(A}{}^{CA'B'}(\sTwist_{1,0} \phi)_{B)CB'}
 -  \tfrac{16}{3} \phi_{(A}(\sTwist_{0,0} \Lambda)_{B)}{}^{A'}\nonumber\\
& - 12 \Lambda (\sTwist_{1,0} \phi)_{AB}{}^{A'}
 + \tfrac{3}{2} \Psi_{ABCD} (\sTwist_{1,0} \phi)^{CDA'},\\
(\sCurl_{3,2} \sTwist_{2,1} \sTwist_{1,0} \phi)_{ABCD}{}^{A'}={}&\Phi_{(AB}{}^{A'B'}(\sTwist_{1,0} \phi)_{CD)B'}
 + \tfrac{5}{2} \Psi_{(ABC}{}^{F}(\sTwist_{1,0} \phi)_{D)F}{}^{A'}\nonumber\\
& -  \tfrac{1}{10} \phi_{(A}(\sCurl_{2,2} \Phi)_{BCD)}{}^{A'}
 -  \tfrac{1}{2} \phi^{F} (\sTwist_{4,0} \Psi)_{ABCDF}{}^{A'},\\
(\sCurlDagger_{3,2} \sTwist_{2,1} \sTwist_{1,0} \phi)_{AB}{}^{A'B'C'}={}&\tfrac{8}{3} \Phi^{C}{}_{(A}{}^{(A'B'}(\sTwist_{1,0} \phi)_{B)C}{}^{C')}
 -  \tfrac{2}{9} \phi_{(A}(\sCurlDagger_{2,2} \Phi)_{B)}{}^{A'B'C'}\nonumber\\
& -  \tfrac{2}{3} \phi^{C} (\sTwist_{2,2} \Phi)_{ABC}{}^{A'B'C'}
 -  \bar\Psi^{A'B'C'}{}_{D'} (\sTwist_{1,0} \phi)_{AB}{}^{D'}.
\end{align}
\end{subequations}
Using irreducible decompositions and the equations above, one can in a systematic way reduce any third order derivative of $\phi_A$ in terms of $\phi_A$, $(\sTwist_{1,0} \phi)_{AB}{}^{A'}$, $(\sTwist_{2,1}\sTwist_{1,0} \phi)_{ABC}{}^{A'B'}$ and $(\sTwist_{3,2}\sTwist_{2,1}\sTwist_{1,0} \phi)_{ABCD}{}^{A'B'C'}$.

\subsection{First kind of symmetry operator for the Dirac-Weyl equation}
\begin{proof}[Proof of Theorem~\ref{Thm::SymOpFirstKindDirac}]
The general second order differential operator, mapping a Dirac-Weyl field $\phi_{A}$ to $\mathcal{S}_{1,0}$ is equivalent to $\phi_{A}\rightarrow \chi_{A}$, where
\begin{align}
\chi_{A}={}&N_{A}{}^{B} \phi_{B}
 + M_{A}{}^{BCA'} (\sTwist_{1,0} \phi)_{BCA'}
 + L_{A}{}^{BCDA'B'} (\sTwist_{2,1} \sTwist_{1,0} \phi)_{BCDA'B'},\label{eq:chiDiracDef}\\
\intertext{and}
L_{ABCD}{}^{A'B'}={}&L_{A(BCD)}{}^{(A'B')},\qquad
M_{ABC}{}^{A'}={}M_{A(BC)}{}^{A'}.\label{eq:SymLMDirac1}
\end{align}
Here, we have used the reduction of the derivatives to the $\sTwist$ operator as discussed in Section~\ref{sec:ReductionDirac}. 
The symmetries \eqref{eq:SymLMDirac1} comes from the symmetries of $(\sTwist_{1,0} \phi)_{AB}{}^{A'}$ and $(\sTwist_{2,1}\sTwist_{1,0} \phi)_{ABC}{}^{A'B'}$.
To be able to make a systematic treatment of the dependence of different components of the coefficients,
we will use the irreducible decompositions
\begin{subequations}
\begin{align}
L_{ABCD}{}^{A'B'}={}&\underset{4,2}{L}{}_{ABCD}{}^{A'B'}
 + \tfrac{3}{4} \underset{2,2}{L}{}_{(BC}{}^{A'B'}\epsilon_{D)A},\\
M_{ABC}{}^{A'}={}&\underset{3,1}{M}{}_{ABC}{}^{A'}
 + \tfrac{2}{3} \underset{1,1}{M}{}_{(B}{}^{A'}\epsilon_{C)A},\\
N_{AB}={}&\underset{2,0}{N}{}_{AB}
 -  \tfrac{1}{2} \underset{0,0}{N}{} \epsilon_{AB}.
\end{align}
\end{subequations}
where
\begin{align*}
\underset{2,2}{L}{}_{AB}{}^{A'B'}\equiv{}&L^{C}{}_{ABC}{}^{A'B'},&
\underset{1,1}{M}{}_{A}{}^{A'}\equiv{}&M^{B}{}_{AB}{}^{A'},&
\underset{0,0}{N}{}\equiv{}&N^{A}{}_{A},\\
\underset{4,2}{L}{}_{ABCD}{}^{A'B'}\equiv{}&L_{(ABCD)}{}^{A'B'},&
\underset{3,1}{M}{}_{ABC}{}^{A'}\equiv{}&M_{(ABC)}{}^{A'},&
\underset{2,0}{N}{}_{AB}\equiv{}&N_{(AB)}.
\end{align*}
We use the convention that a spinor with underscripts $\underset{k,l}{T}{}$ is a totally symmetric valence $(k,l)$ spinor.
Using these spinors, we can rewrite \eqref{eq:chiDiracDef} as
\begin{align}
\chi_{A}={}&- \tfrac{1}{2} \underset{0,0}{N}{} \phi_{A}
 -  \underset{2,0}{N}{}_{AB} \phi^{B}
 -  \tfrac{2}{3} \underset{1,1}{M}{}^{BA'} (\sTwist_{1,0} \phi)_{ABA'}
 -  \underset{3,1}{M}{}_{ABCA'} (\sTwist_{1,0} \phi)^{BCA'}\nonumber\\
& -  \tfrac{3}{4} \underset{2,2}{L}{}^{BCA'B'} (\sTwist_{2,1} \sTwist_{1,0} \phi)_{ABCA'B'}
 -  \underset{4,2}{L}{}_{ABCDA'B'} (\sTwist_{2,1} \sTwist_{1,0} \phi)^{BCDA'B'}.
\end{align}

The condition for the operator $\phi_{A}\rightarrow \chi_{A}$ to be a symmetry operator is
\begin{equation}
(\sCurlDagger_{1,0} \chi)_{A'}=0.\label{eq:chiDirac}
\end{equation}
The definition of the $\sCurlDagger$ operator, the Leibniz rule for the covariant derivative, and the irreducible decomposition \eqref{eq:IrrDecGeneralDer} allows us to write this equation in terms of the fundamental operators acting on the coefficients and the field. Furthermore, using the results from the previous subsection, we see that this equation can be reduced to a linear combination of the spinors  $(\sTwist_{3,2}\sTwist_{2,1}\sTwist_{1,0} \phi)_{ABCD}{}^{A'B'C'}$, $(\sTwist_{2,1}\sTwist_{1,0} \phi)_{ABC}{}^{A'B'}$, $(\sTwist_{1,0} \phi)_{AB}{}^{A'}$ and $\phi_A$. For a general Dirac-Weyl field and an arbitrary point on the manifold, there are no relations between these spinors. Hence, they are independent, and therefore their coefficients have to vanish individually. After the reduction of the derivatives of the field to the $\sTwist$ operator, we can therefore study the different order derivatives in \eqref{eq:chiDirac} separately. We begin with the highest order, and work our way down to order zero.

\subsubsection{Third order part}
The third order derivative term of \eqref{eq:chiDirac} is
\begin{align}\label{eq:thirdorderpartdirac}
0={}&- \underset{4,2}{L}{}^{ABCDB'C'} (\sTwist_{3,2} \sTwist_{2,1} \sTwist_{1,0} \phi)_{ABCDA'B'C'}.
\end{align}
We will now use the argument from Section~\ref{sec:indepspinors} to derive equations for the coefficients in a systematic way. To get rid of the free index in equation \eqref{eq:thirdorderpartdirac} we multiply with an arbitrary spinor field $T^{A'}$ to get
\begin{align}
0={}&- \underset{4,2}{L}{}^{ABCDB'C'}T^{A'} (\sTwist_{3,2} \sTwist_{2,1} \sTwist_{1,0} \phi)_{ABCDA'B'C'}.
\end{align}
From the argument in Section~\ref{sec:indepspinors} and the observation that $T^{A'} (\sTwist_{3,2} \sTwist_{2,1} \sTwist_{1,0} \phi)_{ABCDA'B'C'}$ is irreducible we conclude that
\begin{align}
\underset{4,2}{L}{}_{ABCD}{}^{A'B'}={}&0.
\end{align}

\subsubsection{Second order part}
The second order derivative terms of \eqref{eq:chiDirac} can now be reduced to
\begin{align}
0={}&- \underset{3,1}{M}{}^{ABCB'} (\sTwist_{2,1} \sTwist_{1,0} \phi)_{ABCA'B'}
 + \tfrac{1}{2} (\sCurl_{2,2} \underset{2,2}{L}{})^{ABCB'} (\sTwist_{2,1} \sTwist_{1,0} \phi)_{ABCA'B'}\nonumber\\
& + \tfrac{3}{4} (\sTwist_{2,2} \underset{2,2}{L}{})_{ABCA'B'C'} (\sTwist_{2,1} \sTwist_{1,0} \phi)^{ABCB'C'}.
\end{align}
Here we again multiply with an arbitrary spinor field $T^{A'}$, but here $(\sTwist_{2,1} \sTwist_{1,0} \phi)^{ABCB'C'}T^{A'}$ is not irreducible. Therefore, we decompose it into irreducible parts and get
\begin{align}
0={}&\tfrac{3}{4} T^{(A'}(\sTwist_{2,1} \sTwist_{1,0} \phi)^{|ABC|B'C')} (\sTwist_{2,2} \underset{2,2}{L}{})_{ABCA'B'C'}
\nonumber\\
& + \bigl(\tfrac{1}{2} (\sCurl_{2,2} \underset{2,2}{L}{})^{ABCB'}- \underset{3,1}{M}{}^{ABCB'} \bigr) 
T^{A'}(\sTwist_{2,1} \sTwist_{1,0} \phi)_{ABCA'B'}.
\end{align}
The argument in Section~\ref{sec:indepspinors} tells that the coefficients of the different irreducible parts have to vanish individually which gives
\begin{subequations}
\begin{align}
(\sTwist_{2,2} \underset{2,2}{L}{})_{ABC}{}^{A'B'C'}={}&0,\label{eq:TwistL22Dirac1}\\
\underset{3,1}{M}{}_{ABC}{}^{A'}={}&\tfrac{1}{2} (\sCurl_{2,2} \underset{2,2}{L}{})_{ABC}{}^{A'}.
\end{align}
\end{subequations}

\subsubsection{First order part}
The first order derivative terms of \eqref{eq:chiDirac} are
\begin{align}
0={}&- \underset{2,0}{N}{}^{AB} (\sTwist_{1,0} \phi)_{ABA'}
 + \tfrac{1}{3} (\sCurl_{1,1} \underset{1,1}{M}{})^{AB} (\sTwist_{1,0} \phi)_{ABA'}
 -  \tfrac{1}{2} (\sDiv_{3,1} \underset{3,1}{M}{})^{AB} (\sTwist_{1,0} \phi)_{ABA'}\nonumber\\
& -  \tfrac{2}{3} \underset{2,2}{L}{}_{A}{}^{CB'C'} \Phi_{BCB'C'} (\sTwist_{1,0} \phi)^{AB}{}_{A'}
 - 6 \Lambda \underset{2,2}{L}{}_{ABA'B'} (\sTwist_{1,0} \phi)^{ABB'}\nonumber\\
& + \tfrac{4}{3} \underset{2,2}{L}{}_{A}{}^{C}{}_{B'}{}^{C'} \Phi_{BCA'C'} (\sTwist_{1,0} \phi)^{ABB'}
 + \tfrac{5}{3} \underset{2,2}{L}{}_{A}{}^{C}{}_{A'}{}^{C'} \Phi_{BCB'C'} (\sTwist_{1,0} \phi)^{ABB'}\nonumber\\
& + \tfrac{3}{4} \underset{2,2}{L}{}^{CD}{}_{A'B'} \Psi_{ABCD} (\sTwist_{1,0} \phi)^{ABB'}
 + \tfrac{3}{4} \underset{2,2}{L}{}_{AB}{}^{C'D'} \bar\Psi_{A'B'C'D'} (\sTwist_{1,0} \phi)^{ABB'}\nonumber\\
& -  (\sCurlDagger_{3,1} \underset{3,1}{M}{})_{ABA'B'} (\sTwist_{1,0} \phi)^{ABB'}
 + \tfrac{2}{3} (\sTwist_{1,1} \underset{1,1}{M}{})_{ABA'B'} (\sTwist_{1,0} \phi)^{ABB'}.
\end{align}
Here we again multiply with an arbitrary spinor field $T^{A'}$ and decompose $(\sTwist_{1,0} \phi)^{ABB'}T^{A'}$ into irreducible parts. Due to the argument in Section~\ref{sec:indepspinors} the coefficients of the different irreducible parts have to vanish individually which gives
\begin{subequations}
\begin{align}
0={}&- \underset{2,0}{N}{}_{AB}
 + \tfrac{1}{3} (\sCurl_{1,1} \underset{1,1}{M}{})_{AB}
 -  \tfrac{1}{4} (\sDiv_{3,1} \sCurl_{2,2} \underset{2,2}{L}{})_{AB}
 -  \tfrac{1}{2} \underset{2,2}{L}{}_{(A}{}^{CA'B'}\Phi_{B)CA'B'},\\
0={}&-6 \Lambda \underset{2,2}{L}{}_{AB}{}^{A'B'}
 + \tfrac{3}{4} \underset{2,2}{L}{}^{CDA'B'} \Psi_{ABCD}
 + \tfrac{3}{4} \underset{2,2}{L}{}_{AB}{}^{C'D'} \bar\Psi^{A'B'}{}_{C'D'}
 -  \tfrac{1}{2} (\sCurlDagger_{3,1} \sCurl_{2,2} \underset{2,2}{L}{})_{AB}{}^{A'B'}\nonumber\\
& + 3 \underset{2,2}{L}{}_{(A}{}^{C(A'|C'|}\Phi_{B)C}{}^{B')}{}_{C'}
 + \tfrac{2}{3} (\sTwist_{1,1} \underset{1,1}{M}{})_{AB}{}^{A'B'}.
\end{align}
\end{subequations}
Using the commutators \eqref{eq:DivCurl} and \eqref{eq:DivTwistCurlDaggerCurl} together with \eqref{eq:TwistL22Dirac1}, this reduces to
\begin{subequations}
\begin{align}
\underset{2,0}{N}{}_{AB}={}&\tfrac{1}{3} (\sCurl_{1,1} \underset{1,1}{M}{})_{AB}
 -  \tfrac{1}{6} (\sCurl_{1,1} \sDiv_{2,2} \underset{2,2}{L}{})_{AB},\\
(\sTwist_{1,1} \underset{1,1}{M}{})_{AB}{}^{A'B'}={}&- \tfrac{3}{8} \underset{2,2}{L}{}^{CDA'B'} \Psi_{ABCD}
 + \tfrac{3}{8} \underset{2,2}{L}{}_{AB}{}^{C'D'} \bar\Psi^{A'B'}{}_{C'D'}
 + (\sTwist_{1,1} \sDiv_{2,2} \underset{2,2}{L}{})_{AB}{}^{A'B'}.\label{eq:TwistM11Dirac1a}
\end{align}
\end{subequations}
Isolating the $\sTwist$ terms in \eqref{eq:TwistM11Dirac1a} leads us to make the ansatz
\begin{align}
\underset{1,1}{M}{}_{A}{}^{A'}={}&- \tfrac{3}{2} P_{A}{}^{A'}
 + (\sDiv_{2,2} \underset{2,2}{L}{})_{A}{}^{A'},\label{eq:M11AnsatzDirac1}
\end{align}
where $P_{A}{}^{A'}$ is undetermined. 
With this ansatz, the first order equations reduce to
\begin{subequations}
\begin{align}
(\sTwist_{1,1} P)_{AB}{}^{A'B'}={}&\tfrac{1}{4} \underset{2,2}{L}{}^{CDA'B'} \Psi_{ABCD}
 -  \tfrac{1}{4} \underset{2,2}{L}{}_{AB}{}^{C'D'} \bar\Psi^{A'B'}{}_{C'D'}\nonumber\\
={}&\tfrac{1}{4} (\ObstrOne \underset{2,2}{L}{})_{AB}{}^{A'B'},\label{eq:TwistP11Dirac1b}\\
\underset{2,0}{N}{}_{AB}={}&- \tfrac{1}{2} (\sCurl_{1,1} P)_{AB}
 + \tfrac{1}{6} (\sCurl_{1,1} \sDiv_{2,2} \underset{2,2}{L}{})_{AB}\label{eq:N20Dirac1b}.
\end{align}
\end{subequations}

\subsubsection{Zeroth order part}\label{sec:zerothorderDirac1}
Using the equations above, the zeroth order derivative terms of \eqref{eq:chiDirac} are
\begin{align}
0={}&\phi^{A} (-2 \Lambda \underset{1,1}{M}{}_{AA'}
 + \tfrac{2}{3} \underset{1,1}{M}{}^{BB'} \Phi_{ABA'B'}
 + \tfrac{1}{4} \Psi_{ABCD} (\sCurl_{2,2} \underset{2,2}{L}{})^{BCD}{}_{A'}
 -  \tfrac{5}{12} \underset{2,2}{L}{}^{BC}{}_{A'}{}^{B'} (\sCurl_{2,2} \Phi)_{ABCB'}\nonumber\\
& -  (\sCurlDagger_{2,0} \underset{2,0}{N}{})_{AA'}
 -  \tfrac{1}{6} \underset{2,2}{L}{}_{A}{}^{BB'C'} (\sCurlDagger_{2,2} \Phi)_{BA'B'C'}
 -  \tfrac{8}{3} \underset{2,2}{L}{}_{ABA'B'} (\sTwist_{0,0} \Lambda)^{BB'}
 + \tfrac{1}{2} (\sTwist_{0,0} \underset{0,0}{N}{})_{AA'}\nonumber\\
& + \tfrac{1}{2} \underset{2,2}{L}{}^{BCB'C'} (\sTwist_{2,2} \Phi)_{ABCA'B'C'}).
\end{align}
Here, there is no reason to multiply with an arbitrary $T^{A'}$ and do an irreducible decomposition of $T^{A'}\phi^A$ because $T^{A'}\phi^A$ is already irreducible. Still the argument in Section~\ref{sec:indepspinors} gives that the coefficient of $\phi^A$ will have to vanish.
With the substitutions \eqref{eq:N20Dirac1b} and \eqref{eq:M11AnsatzDirac1}, the vanishing of this coefficient is equivalent to
\begin{align}
(\sTwist_{0,0} \underset{0,0}{N}{})_{A}{}^{A'}={}&-6 \Lambda P_{A}{}^{A'}
 + 2 \Phi_{AB}{}^{A'}{}_{B'} P^{BB'}
 -  \tfrac{1}{2} \Psi_{ABCD} (\sCurl_{2,2} \underset{2,2}{L}{})^{BCDA'}\nonumber\\
& + \tfrac{5}{6} \underset{2,2}{L}{}^{BCA'B'} (\sCurl_{2,2} \Phi)_{ABCB'}
 + \tfrac{1}{3} \underset{2,2}{L}{}_{A}{}^{BB'C'} (\sCurlDagger_{2,2} \Phi)_{B}{}^{A'}{}_{B'C'}
 -  (\sCurlDagger_{2,0} \sCurl_{1,1} P)_{A}{}^{A'}\nonumber\\
& + \tfrac{1}{3} (\sCurlDagger_{2,0} \sCurl_{1,1} \sDiv_{2,2} \underset{2,2}{L}{})_{A}{}^{A'}
 + 4 \Lambda (\sDiv_{2,2} \underset{2,2}{L}{})_{A}{}^{A'}
 -  \tfrac{4}{3} \Phi_{AB}{}^{A'}{}_{B'} (\sDiv_{2,2} \underset{2,2}{L}{})^{BB'}\nonumber\\
& + \tfrac{16}{3} \underset{2,2}{L}{}_{AB}{}^{A'}{}_{B'} (\sTwist_{0,0} \Lambda)^{BB'}
 -  \underset{2,2}{L}{}^{BCB'C'} (\sTwist_{2,2} \Phi)_{ABC}{}^{A'}{}_{B'C'}.\label{eq:TwistN00Dirac1a}
\end{align}
To simplify the $\sCurlDagger\sCurl\sDiv$ term, we first commute the innermost operators with \eqref{eq:DivCurl}. Then the outermost operators are commuted with \eqref{eq:DivCurlDagger}. After that, we are left with the operator $\sDiv\sCurlDagger\sCurl$, which can be turned into $\sDiv\sTwist\sDiv$ by using \eqref{eq:DivTwistCurlDaggerCurl} and \eqref{eq:TwistL22Dirac1}. Finally, the $\sDiv\sTwist\sDiv$ operator can be turned into $\sCurlDagger\sCurl\sDiv$ and $\sTwist\sDiv\sDiv$, again by using \eqref{eq:DivTwistCurlDaggerCurl}, but this time on the outermost operators. In detail
\begin{align*}
(\sCurlDagger_{2,0} \sCurl_{1,1} \sDiv_{2,2} \underset{2,2}{L}{})_{AA'}={}&- \tfrac{3}{2} \nabla_{BA'}\square_{B'C'}\underset{2,2}{L}{}_{A}{}^{BB'C'}
 + \tfrac{3}{2} (\sCurlDagger_{2,0} \sDiv_{3,1} \sCurl_{2,2} \underset{2,2}{L}{})_{AA'}\\*
={}&3 \square_{BC}(\sCurl_{2,2} \underset{2,2}{L}{})_{A}{}^{BC}{}_{A'}
 -  \tfrac{3}{2} \nabla_{BA'}\square_{B'C'}\underset{2,2}{L}{}_{A}{}^{BB'C'}
 + 3 (\sDiv_{2,2} \sCurlDagger_{3,1} \sCurl_{2,2} \underset{2,2}{L}{})_{AA'}\\
={}&3 \square_{BC}(\sCurl_{2,2} \underset{2,2}{L}{})_{A}{}^{BC}{}_{A'}
 -  \tfrac{3}{2} \nabla_{BA'}\square_{B'C'}\underset{2,2}{L}{}_{A}{}^{BB'C'}\nonumber\\*
& - 6 \nabla^{BC'}\square_{(A'}{}^{B'}\underset{2,2}{L}{}_{|AB|C')B'}
 - 3 \nabla^{CB'}\square_{(A}{}^{B}\underset{2,2}{L}{}_{C)BA'B'}\nonumber\\*
& + 4 (\sDiv_{2,2} \sTwist_{1,1} \sDiv_{2,2} \underset{2,2}{L}{})_{AA'}\\
={}&2 \square_{AB}(\sDiv_{2,2} \underset{2,2}{L}{})^{B}{}_{A'}
 + 6 \square_{A'B'}(\sDiv_{2,2} \underset{2,2}{L}{})_{A}{}^{B'}
 + 3 \square_{BC}(\sCurl_{2,2} \underset{2,2}{L}{})_{A}{}^{BC}{}_{A'}\nonumber\\*
& -  \tfrac{3}{2} \nabla_{BA'}\square_{B'C'}\underset{2,2}{L}{}_{A}{}^{BB'C'}
 - 6 \nabla^{BC'}\square_{(A'}{}^{B'}\underset{2,2}{L}{}_{|AB|C')B'}\nonumber\\*
& - 3 \nabla^{CB'}\square_{(A}{}^{B}\underset{2,2}{L}{}_{C)BA'B'}
 - 4 (\sCurlDagger_{2,0} \sCurl_{1,1} \sDiv_{2,2} \underset{2,2}{L}{})_{AA'}
 + 3 (\sTwist_{0,0} \sDiv_{1,1} \sDiv_{2,2} \underset{2,2}{L}{})_{AA'}.
\end{align*}
Isolating the $\sCurlDagger\sCurl\sDiv$ terms, expanding the commutators and using \eqref{eq:TwistL22Dirac1} yield
\begin{align}
(\sCurlDagger_{2,0} \sCurl_{1,1} \sDiv_{2,2} \underset{2,2}{L}{})_{AA'}={}&- \tfrac{8}{5} \Phi^{BC}{}_{A'}{}^{B'} (\sCurl_{2,2} \underset{2,2}{L}{})_{ABCB'}
 + \tfrac{6}{5} \Psi_{ABCD} (\sCurl_{2,2} \underset{2,2}{L}{})^{BCD}{}_{A'}\nonumber\\
& - 2 \underset{2,2}{L}{}^{BC}{}_{A'}{}^{B'} (\sCurl_{2,2} \Phi)_{ABCB'}
 + \tfrac{6}{5} \bar\Psi_{A'B'C'D'} (\sCurlDagger_{2,2} \underset{2,2}{L}{})_{A}{}^{B'C'D'}\nonumber\\
& -  \tfrac{8}{5} \Phi_{A}{}^{BB'C'} (\sCurlDagger_{2,2} \underset{2,2}{L}{})_{BA'B'C'}
 - 2 \underset{2,2}{L}{}_{A}{}^{BB'C'} (\sCurlDagger_{2,2} \Phi)_{BA'B'C'}\nonumber\\
& - 12 \Lambda (\sDiv_{2,2} \underset{2,2}{L}{})_{AA'}
 + \tfrac{44}{15} \Phi_{ABA'B'} (\sDiv_{2,2} \underset{2,2}{L}{})^{BB'}
 -  \tfrac{64}{5} \underset{2,2}{L}{}_{ABA'B'} (\sTwist_{0,0} \Lambda)^{BB'}\nonumber\\
& + \tfrac{3}{5} \underset{2,2}{L}{}^{BCB'C'} (\sTwist_{2,2} \Phi)_{ABCA'B'C'}
 + \tfrac{3}{5} (\sTwist_{0,0} \sDiv_{1,1} \sDiv_{2,2} \underset{2,2}{L}{})_{AA'}.\label{eq:CurlDaggerCurlDivLL22Dirac}
\end{align}
Using this in \eqref{eq:TwistN00Dirac1a}, and using \eqref{eq:DivTwistCurlCurlDagger} combined with \eqref{eq:TwistP11Dirac1b} gives
\begin{align}
(\sTwist_{0,0} \underset{0,0}{N}{})_{A}{}^{A'}={}&- \tfrac{8}{15} \Phi^{BCA'B'} (\sCurl_{2,2} \underset{2,2}{L}{})_{ABCB'}
 + \tfrac{3}{20} \Psi_{ABCD} (\sCurl_{2,2} \underset{2,2}{L}{})^{BCDA'}\nonumber\\
& -  \tfrac{1}{12} \underset{2,2}{L}{}^{BCA'B'} (\sCurl_{2,2} \Phi)_{ABCB'}
 + \tfrac{3}{20} \bar\Psi^{A'}{}_{B'C'D'} (\sCurlDagger_{2,2} \underset{2,2}{L}{})_{A}{}^{B'C'D'}\nonumber\\
& -  \tfrac{8}{15} \Phi_{A}{}^{BB'C'} (\sCurlDagger_{2,2} \underset{2,2}{L}{})_{B}{}^{A'}{}_{B'C'}
 -  \tfrac{1}{12} \underset{2,2}{L}{}_{A}{}^{BB'C'} (\sCurlDagger_{2,2} \Phi)_{B}{}^{A'}{}_{B'C'}\nonumber\\
& -  \tfrac{16}{45} \Phi_{AB}{}^{A'}{}_{B'} (\sDiv_{2,2} \underset{2,2}{L}{})^{BB'}
 + \tfrac{16}{15} \underset{2,2}{L}{}_{AB}{}^{A'}{}_{B'} (\sTwist_{0,0} \Lambda)^{BB'}\nonumber\\
& -  \tfrac{4}{5} \underset{2,2}{L}{}^{BCB'C'} (\sTwist_{2,2} \Phi)_{ABC}{}^{A'}{}_{B'C'}
 -  \tfrac{3}{4} (\sTwist_{0,0} \sDiv_{1,1} P)_{A}{}^{A'}
 + \tfrac{1}{5} (\sTwist_{0,0} \sDiv_{1,1} \sDiv_{2,2} \underset{2,2}{L}{})_{A}{}^{A'}.\label{eq:TwistN00Dirac1c}
\end{align}
To simplify the remaining terms, we define
\begin{equation}
\Upsilon\equiv\underset{2,2}{L}{}_{ABA'B'} \Phi^{ABA'B'}.\label{eq:UpsilonDef}
\end{equation}
Using \eqref{eq:TwistL22Dirac1} the gradient of $\Upsilon$ reduces to
\begin{align}
(\sTwist_{0,0} \Upsilon)_{AA'}={}&- \tfrac{4}{3} \underset{2,2}{L}{}_{A}{}^{B}{}_{A'}{}^{B'} (\sTwist_{0,0}\Lambda)_{BB'}
 + \tfrac{2}{3} \Phi^{BC}{}_{A'}{}^{B'} (\sCurl_{2,2} \underset{2,2}{L}{})_{ABCB'}\nonumber\\
& + \tfrac{2}{3} \underset{2,2}{L}{}^{BC}{}_{A'}{}^{B'} (\sCurl_{2,2} \Phi)_{ABCB'}
 + \tfrac{2}{3} \Phi_{A}{}^{BB'C'} (\sCurlDagger_{2,2} \underset{2,2}{L}{})_{BA'B'C'}\nonumber\\
& + \tfrac{2}{3} \underset{2,2}{L}{}_{A}{}^{BB'C'} (\sCurlDagger_{2,2} \Phi)_{BA'B'C'}
 + \tfrac{4}{9} \Phi_{ABA'B'} (\sDiv_{2,2} \underset{2,2}{L}{})^{BB'}\nonumber\\
& + \underset{2,2}{L}{}^{BCB'C'} (\sTwist_{2,2} \Phi)_{ABCA'B'C'}.\label{eq:GradUpsilon}
\end{align}
This can be used to eliminate most of the terms in \eqref{eq:TwistN00Dirac1c}. Together with the definition of the operator $\ObstrZero$, we find that \eqref{eq:TwistN00Dirac1c} reduces to
\begin{align}
(\sTwist_{0,0} \underset{0,0}{N}{})_{A}{}^{A'}={}&\tfrac{9}{20} (\ObstrZero \underset{2,2}{L}{})_{A}{}^{A'}
 -  \tfrac{4}{5} (\sTwist_{0,0} \Upsilon)_{A}{}^{A'}
 -  \tfrac{3}{4} (\sTwist_{0,0} \sDiv_{1,1} P)_{A}{}^{A'}
 + \tfrac{1}{5} (\sTwist_{0,0} \sDiv_{1,1} \sDiv_{2,2} \underset{2,2}{L}{})_{A}{}^{A'}.
\end{align}
It is now clear that the ansatz
\begin{align}
\underset{0,0}{N}{}={}&-2 Q
 -  \tfrac{4}{5} \Upsilon
 -  \tfrac{3}{4} (\sDiv_{1,1} P)
 + \tfrac{1}{5} (\sDiv_{1,1} \sDiv_{2,2} \underset{2,2}{L}{}),
\end{align}
with $Q$ undetermined gives
\begin{align}
(\sTwist_{0,0} Q)_{A}{}^{A'}={}&- \tfrac{9}{40} (\ObstrZero \underset{2,2}{L}{})_{A}{}^{A'}.\label{eq:TwistQ00Dirac1}
\end{align}

We can now conclude that the only restrictive equations are \eqref{eq:TwistL22Dirac1}, \eqref{eq:TwistP11Dirac1b} and \eqref{eq:TwistQ00Dirac1}. The other equations give expressions for the remaining coefficients in terms of $\underset{2,2}{L}{}_{AB}{}^{A'B'}$, $P_{AA'}$, and $Q$. For convenience we make the replacement $\underset{2,2}{L}{}_{AB}{}^{A'B'} \rightarrow - \tfrac{4}{3} L_{AB}{}^{A'B'}$.
\end{proof}

\subsection{Second kind of symmetry operator for the Dirac-Weyl equation}
\begin{proof}[Proof of Theorem~\ref{Thm::SymOpSecondKindDirac}]
The general second order differential operator, mapping a Dirac-Weyl field $\phi_{A}$ to $\mathcal{S}_{0,1}$ is equivalent to $\phi_{A}\rightarrow \omega_{A'}$, where
\begin{align}
\omega_{A'}={}&N^{B}{}_{A'} \phi_{B}
 + M_{A'}{}^{BCB'} (\sTwist_{1,0} \phi)_{BCB'}
 + L_{A'}{}^{BCDB'C'} (\sTwist_{2,1} \sTwist_{1,0} \phi)_{BCDB'C'},
\end{align}
where
\begin{align}
L_{A'ABCB'C'}={}&L_{A'(ABC)(B'C')},&
M_{A'ABB'}={}&M_{A'(AB)B'}.\label{eq:SymLMDirac2}
\end{align}
Here, we have used the reduction of the derivatives to the $\sTwist$ operator as discussed above. 
The symmetries \eqref{eq:SymLMDirac2} comes from the symmetries of $(\sTwist_{1,0} \phi)_{AB}{}^{A'}$ and $(\sTwist_{2,1}\sTwist_{1,0} \phi)_{ABC}{}^{A'B'}$.
As we did above, we will decompose the coefficients into irreducible parts to more clearly see which parts are independent.
The irreducible decompositions of $L^{A'}{}_{ABC}{}^{B'C'}$ and $M^{A'}{}_{AB}{}^{B'}$ are
\begin{align}
L^{A'}{}_{ABC}{}^{B'C'}={}&\underset{3,3}{L}{}_{ABC}{}^{A'B'C'}
 + \tfrac{2}{3} \underset{3,1}{L}{}_{ABC}{}^{(B'}\bar\epsilon^{C')A'},\\
M^{A'}{}_{AB}{}^{B'}={}&\underset{2,2}{M}{}_{AB}{}^{A'B'}
 -  \tfrac{1}{2} \underset{2,0}{M}{}_{AB} \bar\epsilon^{A'B'},
\end{align}
where
\begin{align*}
\underset{3,1}{L}{}_{ABC}{}^{A'}\equiv{}&L^{B'}{}_{ABC}{}^{A'}{}_{B'},&
\underset{2,0}{M}{}_{AB}\equiv{}&M^{A'}{}_{ABA'},\\
\underset{3,3}{L}{}_{ABC}{}^{A'B'C'}\equiv{}&L^{(A'}{}_{ABC}{}^{B'C')},&
\underset{2,2}{M}{}_{AB}{}^{A'B'}\equiv{}&M^{(A'}{}_{AB}{}^{B')}.
\end{align*}
With these irreducible decompositions, we get
\begin{align}
\omega_{A'}={}&N^{B}{}_{A'} \phi_{B}
 -  \tfrac{1}{2} \underset{2,0}{M}{}^{BC} (\sTwist_{1,0} \phi)_{BCA'}
 -  \underset{2,2}{M}{}_{BCA'B'} (\sTwist_{1,0} \phi)^{BCB'}
 -  \tfrac{2}{3} \underset{3,1}{L}{}^{BCDB'} (\sTwist_{2,1} \sTwist_{1,0} \phi)_{BCDA'B'}\nonumber\\
& -  \underset{3,3}{L}{}_{BCDA'B'C'} (\sTwist_{2,1} \sTwist_{1,0} \phi)^{BCDB'C'}.
\end{align}

The condition for the operator $\phi_{A}\rightarrow \omega_{A'}$ to be a symmetry operator is
\begin{equation}
(\sCurl_{0,1} \omega)_{A}=0.\label{eq:omegaDirac}
\end{equation}
Using the results from Section~\ref{sec:ReductionDirac}, we see that this equation can be reduced to a linear combination of the spinors  $\phi_A$, $(\sTwist_{1,0} \phi)_{AB}{}^{A'}$, $(\sTwist_{2,1}\sTwist_{1,0} \phi)_{ABC}{}^{A'B'}$ and $(\sTwist_{3,2}\sTwist_{2,1}\sTwist_{1,0} \phi)_{ABCD}{}^{A'B'C'}$. 
As above, we can treat these as independent, and therefore their coefficients have to vanish individually. After the reduction of the derivatives of the field to the $\sTwist$ operator, we can therefore study the different order derivatives in \eqref{eq:omegaDirac} separately. We begin with the highest order, and work our way down to order zero. 

\subsubsection{Third order part}
The third order part of \eqref{eq:omegaDirac} is
\begin{align}
0={}&- \underset{3,3}{L}{}^{BCDA'B'C'} (\sTwist_{3,2} \sTwist_{2,1} \sTwist_{1,0} \phi)_{ABCDA'B'C'}.
\end{align}
Using the argument from Section~\ref{sec:indepspinors}, we see that this implies
\begin{align}
\underset{3,3}{L}{}_{ABC}{}^{A'B'C'}={}&0.
\end{align}

\subsubsection{Second order part}
The second order part of \eqref{eq:omegaDirac} now takes the form
\begin{align}
0={}&- \underset{2,2}{M}{}^{BCA'B'} (\sTwist_{2,1} \sTwist_{1,0} \phi)_{ABCA'B'}
 + \tfrac{1}{2} (\sCurlDagger_{3,1} \underset{3,1}{L}{})^{BCA'B'} (\sTwist_{2,1} \sTwist_{1,0} \phi)_{ABCA'B'}\nonumber\\
& + \tfrac{2}{3} (\sTwist_{3,1} \underset{3,1}{L}{})_{ABCDA'B'} (\sTwist_{2,1} \sTwist_{1,0} \phi)^{BCDA'B'}.
\end{align}
Here we multiply with an arbitrary spinor field $T^{A}$ and decompose $(\sTwist_{2,1}\sTwist_{1,0} \phi)^{BCDC'D'}T^{A}$ into irreducible parts. Due to the argument in Section~\ref{sec:indepspinors} the coefficients of the different irreducible parts have to vanish individually which gives
\begin{subequations}
\begin{align}
(\sTwist_{3,1} \underset{3,1}{L}{})_{ABCD}{}^{A'B'}={}&0,\label{eq:TwistLL31}\\
\underset{2,2}{M}{}_{AB}{}^{A'B'}={}&\tfrac{1}{2} (\sCurlDagger_{3,1} \underset{3,1}{L}{})_{AB}{}^{A'B'}\label{eq:M22Eq1}.
\end{align}
\end{subequations}

\subsubsection{First order part}
The first order part of \eqref{eq:omegaDirac} can now be reduced to
\begin{align}
0={}&- N^{BA'} (\sTwist_{1,0} \phi)_{ABA'}
 + \tfrac{1}{3} (\sCurlDagger_{2,0} \underset{2,0}{M}{})^{BA'} (\sTwist_{1,0} \phi)_{ABA'}
 -  \tfrac{2}{3} (\sDiv_{2,2} \underset{2,2}{M}{})^{BA'} (\sTwist_{1,0} \phi)_{ABA'}\nonumber\\
& + \tfrac{1}{3} \underset{3,1}{L}{}_{BCDB'} \Phi^{CD}{}_{A'}{}^{B'} (\sTwist_{1,0} \phi)_{A}{}^{BA'}
 -  \tfrac{5}{12} \underset{3,1}{L}{}^{CDF}{}_{A'} \Psi_{BCDF} (\sTwist_{1,0} \phi)_{A}{}^{BA'}\nonumber\\
& - 6 \Lambda \underset{3,1}{L}{}_{ABCA'} (\sTwist_{1,0} \phi)^{BCA'}
 + \tfrac{1}{3} \underset{3,1}{L}{}_{BCDB'} \Phi_{A}{}^{D}{}_{A'}{}^{B'} (\sTwist_{1,0} \phi)^{BCA'}\nonumber\\
& + \tfrac{5}{3} \underset{3,1}{L}{}_{ACDB'} \Phi_{B}{}^{D}{}_{A'}{}^{B'} (\sTwist_{1,0} \phi)^{BCA'}
 + \tfrac{5}{4} \underset{3,1}{L}{}_{B}{}^{DF}{}_{A'} \Psi_{ACDF} (\sTwist_{1,0} \phi)^{BCA'}\nonumber\\
& + \tfrac{3}{4} \underset{3,1}{L}{}_{A}{}^{DF}{}_{A'} \Psi_{BCDF} (\sTwist_{1,0} \phi)^{BCA'}
 -  (\sCurl_{2,2} \underset{2,2}{M}{})_{ABCA'} (\sTwist_{1,0} \phi)^{BCA'}\nonumber\\
& + \tfrac{1}{2} (\sTwist_{2,0} \underset{2,0}{M}{})_{ABCA'} (\sTwist_{1,0} \phi)^{BCA'}.
\end{align}
Here we again multiply with an arbitrary spinor field $T^{A}$ and decompose $(\sTwist_{1,0} \phi)^{BCC'}T^{A}$ into irreducible parts. Due to the argument in Section~\ref{sec:indepspinors} the coefficients of the different irreducible parts have to vanish individually which gives
\begin{align}
0={}&- N_{A}{}^{A'}
 -  \tfrac{1}{3} \underset{3,1}{L}{}^{BCDA'} \Psi_{ABCD}
 + \tfrac{1}{3} (\sCurlDagger_{2,0} \underset{2,0}{M}{})_{A}{}^{A'}
 -  \tfrac{2}{3} (\sDiv_{2,2} \underset{2,2}{M}{})_{A}{}^{A'},\\
0={}&-6 \Lambda \underset{3,1}{L}{}_{ABC}{}^{A'}
 -  (\sCurl_{2,2} \underset{2,2}{M}{})_{ABC}{}^{A'}
 + 2 \underset{3,1}{L}{}_{(BC|DB'|}\Phi_{A)}{}^{DA'B'}
 + 2 \underset{3,1}{L}{}_{(A}{}^{DFA'}\Psi_{BC)DF}\nonumber\\
& + \tfrac{1}{2} (\sTwist_{2,0} \underset{2,0}{M}{})_{ABC}{}^{A'}.
\end{align}
By \eqref{eq:M22Eq1}, the commutator \eqref{eq:DivTwistCurlCurlDagger} and \eqref{eq:TwistLL31} these reduce to
\begin{subequations}
\begin{align}
N_{A}{}^{A'}={}&- \tfrac{1}{3} \underset{3,1}{L}{}_{ABCB'} \Phi^{BCA'B'}
 + \tfrac{1}{3} (\sCurlDagger_{2,0} \underset{2,0}{M}{})_{A}{}^{A'}
 -  \tfrac{1}{6} (\sCurlDagger_{2,0} \sDiv_{3,1} \underset{3,1}{L}{})_{A}{}^{A'},\\
(\sTwist_{2,0} \underset{2,0}{M}{})_{ABC}{}^{A'}={}&(\sTwist_{2,0} \sDiv_{3,1} \underset{3,1}{L}{})_{ABC}{}^{A'}.
\end{align}
\end{subequations}
If we make the ansatz 
\begin{align}
\underset{2,0}{M}{}_{AB}={}&-2 P_{AB}
 + (\sDiv_{3,1} \underset{3,1}{L}{})_{AB},
\end{align}
these equations reduce to
\begin{subequations}
\begin{align}
N_{A}{}^{A'}={}&- \tfrac{1}{3} \underset{3,1}{L}{}_{ABCB'} \Phi^{BCA'B'}
 -  \tfrac{2}{3} (\sCurlDagger_{2,0} P)_{A}{}^{A'}
 + \tfrac{1}{6} (\sCurlDagger_{2,0} \sDiv_{3,1} \underset{3,1}{L}{})_{A}{}^{A'},\label{eq:N11Dirac2a}\\
(\sTwist_{2,0} P)_{ABC}{}^{A'}={}&0.\label{eq:TwistPDirac2a}
\end{align}
\end{subequations}

\subsubsection{Zeroth order part}
The zeroth order part of \eqref{eq:omegaDirac} can now be reduced to
\begin{align}
0={}&-2 \Lambda \underset{2,0}{M}{}_{AB} \phi^{B}
 + \tfrac{1}{2} \underset{2,0}{M}{}^{CD} \Psi_{ABCD} \phi^{B}
 -  \phi^{B} (\sCurl_{1,1} N)_{AB}
 -  \tfrac{1}{20} \underset{3,1}{L}{}_{B}{}^{CDA'} \phi^{B} (\sCurl_{2,2} \Phi)_{ACDA'}\nonumber\\
& + \tfrac{1}{60} \underset{3,1}{L}{}^{BCDA'} \phi_{A} (\sCurl_{2,2} \Phi)_{BCDA'}
 -  \tfrac{5}{12} \underset{3,1}{L}{}_{A}{}^{CDA'} \phi^{B} (\sCurl_{2,2} \Phi)_{BCDA'}\nonumber\\
& -  \tfrac{1}{3} \Phi_{B}{}^{CA'B'} \phi^{B} (\sCurlDagger_{3,1} \underset{3,1}{L}{})_{ACA'B'}
 -  \tfrac{1}{2} \phi_{A} (\sDiv_{1,1} N)
 -  \tfrac{8}{3} \underset{3,1}{L}{}_{ABCA'} \phi^{B} (\sTwist_{0,0} \Lambda)^{CA'}\nonumber\\
& + \tfrac{1}{3} \underset{3,1}{L}{}^{CDFA'} \phi^{B} (\sTwist_{4,0} \Psi)_{ABCDFA'}.
\end{align}
Here we again multiply with an arbitrary spinor field $T^{A}$ and decompose $\phi^{B}T^{A}$ into irreducible parts. Due to the argument in Section~\ref{sec:indepspinors} the coefficients of the different irreducible parts have to vanish individually which gives
\begin{subequations}
\begin{align}
0={}&- \tfrac{1}{6} \underset{3,1}{L}{}^{ABCA'} (\sCurl_{2,2} \Phi)_{ABCA'}
 + \tfrac{1}{6} \Phi^{ABA'B'} (\sCurlDagger_{3,1} \underset{3,1}{L}{})_{ABA'B'}
 -  \tfrac{1}{2} (\sDiv_{1,1} N),\label{eq:ZerithOrderDirac2a}\\
0={}&-2 \Lambda \underset{2,0}{M}{}_{AB}
 + \tfrac{1}{2} \underset{2,0}{M}{}^{CD} \Psi_{ABCD}
 -  (\sCurl_{1,1} N)_{AB}
 -  \tfrac{7}{15} \underset{3,1}{L}{}_{(A}{}^{CDA'}(\sCurl_{2,2} \Phi)_{B)CDA'}\nonumber\\
& -  \tfrac{1}{3} \Phi_{(A}{}^{CA'B'}(\sCurlDagger_{3,1} \underset{3,1}{L}{})_{B)CA'B'}
 -  \tfrac{8}{3} \underset{3,1}{L}{}_{ABCA'} (\sTwist_{0,0} \Lambda)^{CA'}
 + \tfrac{1}{3} \underset{3,1}{L}{}^{CDFA'} (\sTwist_{4,0} \Psi)_{ABCDFA'}.\label{eq:ZerithOrderDirac2b}
\end{align}
\end{subequations}
The equation \eqref{eq:N11Dirac2a} together with the commutator \eqref{eq:DivCurlDagger} gives \eqref{eq:ZerithOrderDirac2a}. If we substitute \eqref{eq:N11Dirac2a} in \eqref{eq:ZerithOrderDirac2b}, we get a term with the third order operator $\sCurl\sCurlDagger\sDiv$. To handle this we use the same technique as in Section~\ref{sec:zerothorderDirac1}. We first commute the innermost operators with \eqref{eq:DivCurlDagger}. Then the outermost operators are commuted with \eqref{eq:DivCurl}. After that, we are left with the operator $\sDiv\sCurl\sCurlDagger$, which can be turned into $\sDiv\sTwist\sDiv$ by using \eqref{eq:DivTwistCurlCurlDagger} and \eqref{eq:TwistL22Dirac1}. Finally, the $\sDiv\sTwist\sDiv$ operator can be turned into $\sCurl\sCurlDagger\sDiv$ and $\sTwist\sDiv\sDiv$, again by using \eqref{eq:DivTwistCurlCurlDagger}, but this time on the outermost operators.
\begin{align*}
(\sCurl_{1,1} \sCurlDagger_{2,0} \sDiv_{3,1} \underset{3,1}{L}{})_{AB}={}&2 (\sCurl_{1,1} \sDiv_{2,2} \sCurlDagger_{3,1} \underset{3,1}{L}{})_{AB}
 + 2 \nabla_{(A}{}^{A'}\square^{CD}\underset{3,1}{L}{}_{B)CDA'}\\
={}&3 \square_{A'B'}(\sCurlDagger_{3,1} \underset{3,1}{L}{})_{AB}{}^{A'B'}
 + 3 (\sDiv_{3,1} \sCurl_{2,2} \sCurlDagger_{3,1} \underset{3,1}{L}{})_{AB}
 + 2 \nabla_{(A}{}^{A'}\square^{CD}\underset{3,1}{L}{}_{B)CDA'}\\
={}&3 \square_{A'B'}(\sCurlDagger_{3,1} \underset{3,1}{L}{})_{AB}{}^{A'B'}
 + 2 \nabla_{CB'}\square_{A'}{}^{B'}\underset{3,1}{L}{}_{AB}{}^{CA'}
 - 6 \nabla^{DA'}\square_{(A}{}^{C}\underset{3,1}{L}{}_{BD)CA'}\nonumber\\
& + 3 (\sDiv_{3,1} \sTwist_{2,0} \sDiv_{3,1} \underset{3,1}{L}{})_{AB}
 + 2 \nabla_{(A}{}^{A'}\square^{CD}\underset{3,1}{L}{}_{B)CDA'}\\
={}&3 \square_{A'B'}(\sCurlDagger_{3,1} \underset{3,1}{L}{})_{AB}{}^{A'B'}
 + 2 \nabla_{CB'}\square_{A'}{}^{B'}\underset{3,1}{L}{}_{AB}{}^{CA'}
 - 6 \nabla^{DA'}\square_{(A}{}^{C}\underset{3,1}{L}{}_{BD)CA'}\nonumber\\
& - 4 (\sCurl_{1,1} \sCurlDagger_{2,0} \sDiv_{3,1} \underset{3,1}{L}{})_{AB}
 - 6 \square_{(A}{}^{C}(\sDiv_{3,1} \underset{3,1}{L}{})_{B)C}
 + 2 \nabla_{(A}{}^{A'}\square^{CD}\underset{3,1}{L}{}_{B)CDA'}.
\end{align*}
Isolating the $\sCurl\sCurlDagger\sDiv$ terms and expanding the commutators and using \eqref{eq:TwistLL31} yield
\begin{align}
(\sCurl_{1,1} \sCurlDagger_{2,0} \sDiv_{3,1} \underset{3,1}{L}{})_{AB}={}&- \Psi_{B}{}^{CDF} (\sCurl_{3,1} \underset{3,1}{L}{})_{ACDF}
 -  \Psi_{A}{}^{CDF} (\sCurl_{3,1} \underset{3,1}{L}{})_{BCDF}\nonumber\\
& -  \tfrac{42}{25} \underset{3,1}{L}{}_{B}{}^{CDA'} (\sCurl_{2,2} \Phi)_{ACDA'}
 -  \tfrac{42}{25} \underset{3,1}{L}{}_{A}{}^{CDA'} (\sCurl_{2,2} \Phi)_{BCDA'}\nonumber\\
& -  \tfrac{3}{2} \Phi_{B}{}^{CA'B'} (\sCurlDagger_{3,1} \underset{3,1}{L}{})_{ACA'B'}
 -  \tfrac{3}{2} \Phi_{A}{}^{CA'B'} (\sCurlDagger_{3,1} \underset{3,1}{L}{})_{BCA'B'}\nonumber\\
& - 12 \Lambda (\sDiv_{3,1} \underset{3,1}{L}{})_{AB}
 + \tfrac{21}{10} \Psi_{ABCD} (\sDiv_{3,1} \underset{3,1}{L}{})^{CD}
 - 12 \underset{3,1}{L}{}_{ABCA'} (\sTwist_{0,0} \Lambda)^{CA'}\nonumber\\
& + \tfrac{2}{5} \underset{3,1}{L}{}^{CDFA'} (\sTwist_{4,0} \Psi)_{ABCDFA'}.
\end{align}
The equation \eqref{eq:N11Dirac2a} together with the equation above, the commutator \eqref{eq:DivTwistCurlCurlDagger} and \eqref{eq:TwistPDirac2a} gives
\begin{align}
(\sCurl_{1,1} N)_{AB}={}&4 \Lambda P_{AB}
 -  \Psi_{ABCD} P^{CD}
 - 2 \Lambda (\sDiv_{3,1} \underset{3,1}{L}{})_{AB}
 + \tfrac{7}{20} \Psi_{ABCD} (\sDiv_{3,1} \underset{3,1}{L}{})^{CD}\nonumber\\
& -  \tfrac{17}{75} \underset{3,1}{L}{}_{(A}{}^{CDA'}(\sCurl_{2,2} \Phi)_{B)CDA'}
 -  \tfrac{1}{3} \Phi_{(A}{}^{CA'B'}(\sCurlDagger_{3,1} \underset{3,1}{L}{})_{B)CA'B'}\nonumber\\
& -  \tfrac{1}{3} \Psi_{(A}{}^{CDF}(\sCurl_{3,1} \underset{3,1}{L}{})_{B)CDF}
 -  \tfrac{8}{3} \underset{3,1}{L}{}_{ABCA'} (\sTwist_{0,0} \Lambda)^{CA'}\nonumber\\
& + \tfrac{1}{15} \underset{3,1}{L}{}^{CDFA'} (\sTwist_{4,0} \Psi)_{ABCDFA'}.
\end{align}
Due to this, the equation \eqref{eq:ZerithOrderDirac2b} reduces to the auxiliary condition
\begin{align}
0={}&\tfrac{3}{4} \Psi_{ABCD} (\sDiv_{3,1} \underset{3,1}{L}{})^{CD}
 -  \tfrac{6}{5} \underset{3,1}{L}{}_{(A}{}^{CDA'}(\sCurl_{2,2} \Phi)_{B)CDA'}\nonumber\\
& + \tfrac{5}{3} \Psi_{(A}{}^{CDF}(\sCurl_{3,1} \underset{3,1}{L}{})_{B)CDF}
 + \tfrac{4}{3} \underset{3,1}{L}{}^{CDFA'} (\sTwist_{4,0} \Psi)_{ABCDFA'}.\label{eq:auxcondDirac2}
\end{align}
We can conclude that the only restrictive equations are \eqref{eq:TwistLL31}, \eqref{eq:TwistPDirac2a} and \eqref{eq:auxcondDirac2}. The other equations express the remaining coefficients in terms of $\underset{3,1}{L}{}_{ABCA'}$ and $P_{AB}$. For convenience we make the replacement $\underset{3,1}{L}{}_{ABC}{}^{A'} \rightarrow - \tfrac{3}{2} L_{ABC}{}^{A'}$.
\end{proof}

\section{The Maxwell equation} \label{sec:spin1} 

\begin{theorem}\label{Thm::SymOpFirstKind}
There exists a symmetry operator of the first kind $\phi_{AB}\rightarrow \chi_{AB}$, with order less or equal to two, if and only if there are spinor fields $L_{AB}{}^{A'B'}=L_{(AB)}{}^{(A'B')}$, $P_{AA'}$ and $Q$ such that
\begin{subequations}
\begin{align}
(\sTwist_{2,2} L)_{ABC}{}^{A'B'C'}={}&0,\label{eq:Lfirstkind}\\
(\sTwist_{1,1} P)_{AB}{}^{A'B'}={}&- \tfrac{2}{3} (\ObstrOne L)_{AB}{}^{A'B'},\label{eq:Qcondfirstkind}\\
(\sTwist_{0,0} Q)_{BA'}={}&0.
\end{align}
\end{subequations}
The symmetry operator then takes the form
\begin{align}
\chi_{AB}={}&Q \phi_{AB}+(\sCurl_{1,1} A)_{AB},\\
\intertext{where}
A_{AA'}={}&- P^{B}{}_{A'} \phi_{AB}
 + \tfrac{1}{3} \phi^{BC} (\sCurl_{2,2} L)_{ABCA'}
 -  \tfrac{4}{9} \phi_{AB} (\sDiv_{2,2} L)^{B}{}_{A'}
 -  L^{BC}{}_{A'}{}^{B'} (\sTwist_{2,0} \phi)_{ABCB'}.
\end{align}
We also note that
\begin{align}
(\sCurlDagger_{1,1} A)_{A'B'} ={}& 0.\label{eq:CurlDaggerA}
\end{align}
\end{theorem}
\begin{remark}
\begin{enumerate}
\item{
Observe that one can add a gradient of a scalar to the potential $A_{AA'}$ without changing the symmetry operator. Hence, adding $\nabla_{AA'}(\Lambda^{BC}\phi_{BC})$ to $A_{AA'}$ with an arbitrary field $\Lambda_{AB}$ is possible.}
\item{
With $L_{ABA'B'}=0$, the first order operator takes the form
\begin{align}
\chi_{AB}={}&
\hat{\mathcal{L}}_{P}\phi_{AB}+ Q \phi_{AB}.
\end{align}
}
\end{enumerate}
\end{remark}

\begin{theorem}\label{Thm::SymOpSecondKind}
There exists second order a symmetry operator of the second kind $\phi_{AB}\rightarrow \omega_{A'B'}$, with order less or equal to two, if and only if there is a spinor field $L_{ABCD}=L_{(ABCD)}$ such that
\begin{align}
(\sTwist_{4,0} L)_{ABCDF}{}^{A'}={}&0.\label{eq:LSecondKind}
\end{align}
The symmetry operator then takes the form
\begin{align}
\omega_{A'B'}={}&(\sCurlDagger_{1,1} B)_{A'B'},\\
\intertext{where}
B_{AA'}={}&\tfrac{3}{5} \phi^{BC} (\sCurlDagger_{4,0} L)_{ABCA'}
 + L_{ABCD} (\sTwist_{2,0} \phi)^{BCD}{}_{A'}.
\end{align}
We also note that
\begin{align}
(\sCurl_{1,1} B)_{AB}={}&0.\label{eq:CurlB}
\end{align}
\end{theorem}
\begin{remark}
\begin{enumerate}
\item{
Observe that also here we can add a gradient of a scalar to the potential $B_{AA'}$ without changing the symmetry operator. Hence, adding $\nabla_{AA'}(\Lambda^{BC}\phi_{BC})$ to $B_{AA'}$ with an arbitrary field $\Lambda_{AB}$ is possible.}
\item{
Due to the equations \eqref{eq:CurlDaggerA} and \eqref{eq:CurlB}, we can use $A_{AA'}+B_{AA'}$ as a potential for both $\chi_{AB}$ and $\omega_{A'B'}$.
}
\end{enumerate}
\end{remark}

\subsection{Reduction of derivatives of the field}\label{sec:ReductionMaxwell}
In our notation, the Maxwell equation $\nabla^{A}{}_{A'}\phi_{AB}=0$, takes the form $(\sCurlDagger_{2,0} \phi)_{A'}=0$. From this we see that the only irreducible part of $\nabla_{A}{}^{A'}\phi_{BC}$ is $(\sTwist_{2,0} \phi)_{ABC}{}^{A'}$. By commuting derivatives we see that all higher order derivatives of $\phi_{AB}$ can be reduced to totally symmetrized derivatives and lower order terms consisting of curvature times lower order symmetrized derivatives.

Together with the Maxwell equation, the commutators \eqref{eq:DivTwistCurlCurlDagger}, \eqref{eq:CurlTwist}, \eqref{eq:CurlDaggerTwist} gives
\begin{subequations}
\begin{align}
(\sDiv_{3,1} \sTwist_{2,0} \phi)_{AB}={}&
-8 \Lambda \phi_{AB}
 + 2 \Psi_{ABCD} \phi^{CD},\\
(\sCurl_{3,1} \sTwist_{2,0} \phi)_{ABCD}={}&2 \Psi_{(ABC}{}^{F}\phi_{D)F},\\
(\sCurlDagger_{3,1} \sTwist_{2,0} \phi)_{ABA'B'}={}&2 \Phi_{(A}{}^{C}{}_{|A'B'|}\phi_{B)C}
\end{align}
\end{subequations}
The higher order derivatives can be computed using the commutators \eqref{eq:DivTwistCurlCurlDagger}, \eqref{eq:CurlTwist}, \eqref{eq:CurlDaggerTwist} together with the equations above and the Bianchi system to get
\begin{subequations}
\begin{align}
(\sDiv_{4,2} \sTwist_{3,1} \sTwist_{2,0} \phi)_{ABCA'}={}&\tfrac{9}{2} \Phi_{(A}{}^{D}{}_{|A'|}{}^{B'}(\sTwist_{2,0} \phi)_{BC)DB'}
 + \tfrac{9}{2} \Psi_{(AB}{}^{DF}(\sTwist_{2,0} \phi)_{C)DFA'}\nonumber\\
& -  \tfrac{15}{2} \phi_{(AB}(\sTwist_{0,0}\Lambda)_{C)A'}
 + \tfrac{21}{10} \phi_{(A}{}^{D}(\sCurl_{2,2} \Phi)_{BC)DA'}\nonumber\\
& + \tfrac{3}{2} \phi^{DF} (\sTwist_{4,0} \Psi)_{ABCDFA'}
 - 15 \Lambda (\sTwist_{2,0} \phi)_{ABCA'},\\
(\sCurl_{4,2} \sTwist_{3,1} \sTwist_{2,0} \phi)_{ABCDFB'}={}&\Phi_{(AB|B'|}{}^{A'}(\sTwist_{2,0} \phi)_{CDF)A'}
 + 4 \Psi_{(ABC}{}^{H}(\sTwist_{2,0} \phi)_{DF)HB'}\nonumber\\
& -  \tfrac{1}{5} \phi_{(AB}(\sCurl_{2,2} \Phi)_{CDF)B'}
 -  \phi_{(A}{}^{H}(\sTwist_{4,0} \Psi)_{BCDF)HB'},\\
(\sCurlDagger_{4,2} \sTwist_{3,1} \sTwist_{2,0} \phi)_{ABC}{}^{A'B'C'}={}&\tfrac{9}{2} \Phi^{D}{}_{(A}{}^{(A'B'}(\sTwist_{2,0} \phi)_{BC)D}{}^{C')}
 -  \tfrac{1}{2} \phi_{(AB}(\sCurlDagger_{2,2} \Phi)_{C)}{}^{A'B'C'}\nonumber\\
& -  \tfrac{3}{2} \phi_{(A}{}^{D}(\sTwist_{2,2} \Phi)_{BC)D}{}^{A'B'C'}
 -  \bar\Psi^{A'B'C'}{}_{D'} (\sTwist_{2,0} \phi)_{ABC}{}^{D'}.
\end{align}
\end{subequations}
These can in a systematic way be used to reduce any derivative up to third order of $\phi_{AB}$ in terms of $\phi_{AB}$, $(\sTwist_{2,0} \phi)_{ABC}{}^{A'}$, $(\sTwist_{3,1} \sTwist_{2,0} \phi)_{ABCD}{}^{A'B'}$ and $(\sTwist_{4,2} \sTwist_{3,1} \sTwist_{2,0} \phi)_{ABCDF}{}^{A'B'C'}$.

\subsection{First kind of symmetry operator for the Maxwell equation}
\begin{proof}[Proof of Theorem~\ref{Thm::SymOpFirstKind}]
The general second order differential operator, mapping a Maxwell field $\phi_{AB}$ to $\mathcal{S}_{2,0}$ is equivalent to $\phi_{AB}\rightarrow \chi_{AB}$, where
\begin{align}
\chi_{AB}={}&N_{ABCD} \phi^{CD}
 + M_{ABCDFA'} (\sTwist_{2,0} \phi)^{CDFA'}
 + L_{ABCDFHA'B'} (\sTwist_{3,1} \sTwist_{2,0} \phi)^{CDFHA'B'},
\end{align}
and
\begin{align*}
L_{AB}{}^{CDFHA'B'}={}&L_{(AB)}{}^{(CDFH)(A'B')},&
M_{AB}{}^{CDFA'}={}&M_{(AB)}{}^{(CDF)A'},&
N_{AB}{}^{CD}={}&N_{(AB)}{}^{(CD)}.
\end{align*}
Here, we have used the reduction of the derivatives to the $\sTwist$ operator as discussed in Section~\ref{sec:ReductionMaxwell}. 
The symmetries comes from the symmetries of $ (\sTwist_{2,0} \phi)_{ABC}{}^{A'}$ and $(\sTwist_{3,1} \sTwist_{2,0} \phi)_{ABCD}{}^{A'B'}$.
To be able to make a systematic treatment of the dependence of different components of the coefficients,
we will use the irreducible decompositions
\begin{subequations}
\begin{align}
L_{AB}{}^{CDFHA'B'}={}&\underset{6,2}{L}{}_{AB}{}^{CDFHA'B'}
 -  \tfrac{4}{3} \epsilon_{(A}{}^{(C}\underset{4,2}{L}{}^{DFH)}{}_{B)}{}^{A'B'}
 -  \tfrac{3}{5} \epsilon^{(C}{}_{(A}\epsilon_{B)}{}^{D}\underset{2,2}{L}{}^{FH)A'B'},\\
M_{AB}{}^{CDFA'}={}&\underset{5,1}{M}{}_{AB}{}^{CDFA'}
 -  \tfrac{6}{5} \epsilon_{(A}{}^{(C}\underset{3,1}{M}{}^{DF)}{}_{B)}{}^{A'}
 -  \tfrac{1}{2} \epsilon^{(C}{}_{(A}\epsilon_{B)}{}^{D}\underset{1,1}{M}{}^{F)A'},\\
N_{AB}{}^{CD}={}&\underset{4,0}{N}{}_{AB}{}^{CD}
 -  \epsilon_{(A}{}^{(C}\underset{2,0}{N}{}^{D)}{}_{B)}
 -  \tfrac{1}{3} \underset{0,0}{N}{} \epsilon_{(A}{}^{(C}\epsilon^{D)}{}_{B)}.
\end{align}
\end{subequations}
where the different irreducible parts are
\begin{align*}
\underset{2,2}{L}{}_{AB}{}^{A'B'}\equiv{}&L^{CD}{}_{ABCD}{}^{A'B'},&
\underset{1,1}{M}{}_{A}{}^{A'}\equiv{}&M^{BC}{}_{ABC}{}^{A'},&
\underset{0,0}{N}{}\equiv{}&N^{AB}{}_{AB},\\
\underset{4,2}{L}{}_{ABCD}{}^{A'B'}\equiv{}&L_{(A}{}^{F}{}_{BCD)F}{}^{A'B'},&
\underset{3,1}{M}{}_{ABC}{}^{A'}\equiv{}&M_{(A}{}^{D}{}_{BC)D}{}^{A'},&
\underset{2,0}{N}{}_{AB}\equiv{}&N_{(A}{}^{C}{}_{B)C},\\
\underset{6,2}{L}{}_{ABCDFH}{}^{A'B'}\equiv{}&L_{(ABCDFH)}{}^{A'B'},&
\underset{5,1}{M}{}_{ABCDF}{}^{A'}\equiv{}&M_{(ABCDF)}{}^{A'},&
\underset{4,0}{N}{}_{ABCD}\equiv{}&N_{(ABCD)}.
\end{align*}

Now, we want the operator to be a symmetry operator, which means that
\begin{equation}
(\sCurlDagger_{2,0}\chi)_{AA'}=0.\label{eq:chiMaxwell}
\end{equation}
Using the results from the previous subsection, we see that this equation can be reduced to a linear combination of the spinors $(\sTwist_{2,4} \sTwist_{1,3} \sTwist_{0,2} \phi)^{A'B'C'}{}_{ABCDF}$, $(\sTwist_{1,3} \sTwist_{0,2} \phi)^{A'B'}{}_{ABCD}$, $(\sTwist_{0,2} \phi)^{A'}{}_{ABC}$ and $\phi_{AB}$. For a general Maxwell field and an arbitrary point on the manifold, there are no relations between these spinors. Hence, they are independent, and therefore their coefficients have to vanish individually. After the reduction of the derivatives of the Maxwell field to the $\sTwist$ operator, we can therefore study the different order derivatives of $\phi_{AB}$ in \eqref{eq:chiMaxwell} separately.

\subsubsection{Third order part}
The third order derivative terms of \eqref{eq:chiMaxwell} are
\begin{align}\label{eq:thirdorderpartmaxwell}
0={}&\tfrac{2}{3} \underset{4,2}{L}{}_{BCDFB'C'} (\sTwist_{4,2} \sTwist_{3,1} \sTwist_{2,0} \phi)_{A}{}^{BCDF}{}_{A'}{}^{B'C'}\nonumber\\
& + \underset{6,2}{L}{}_{ABCDFHB'C'} (\sTwist_{4,2} \sTwist_{3,1} \sTwist_{2,0} \phi)^{BCDFH}{}_{A'}{}^{B'C'}.
\end{align}
We can multiply this with an arbitrary vector field $T^{AA'}$ and split $(\sTwist_{4,2} \sTwist_{3,1} \sTwist_{2,0} \phi)^{ABCDFA'B'C'}T^{H}{}_{A'}$ into irreducible parts. Then we get
\begin{align}
0={}&\underset{6,2}{L}{}_{ABCDFHB'C'} T^{(A|A'|}(\sTwist_{4,2} \sTwist_{3,1} \sTwist_{2,0} \phi)^{BCDFH)}{}_{A'}{}^{B'C'}\nonumber\\
& + \tfrac{2}{3} \underset{4,2}{L}{}_{BCDFB'C'} T_{AA'} (\sTwist_{4,2} \sTwist_{3,1} \sTwist_{2,0} \phi)^{ABCDFA'B'C'}.
\end{align}
The argument in Section~\ref{sec:indepspinors} gives that the symmetrized coefficients of the irreducible parts $T^{(A|A'|}(\sTwist_{4,2} \sTwist_{3,1} \sTwist_{2,0} \phi)^{BCDFH)}{}_{A'}{}^{B'C'}$ and $T_{AA'} (\sTwist_{4,2} \sTwist_{3,1} \sTwist_{2,0} \phi)^{ABCDFA'B'C'}$ must vanish. This means that  \eqref{eq:thirdorderpartmaxwell} is equivalent to the system
\begin{subequations}\label{eq:ThirdOrderSystem}
\begin{align}
\underset{6,2}{L}{}_{ABCDFHB'C'}={}&0,\label{eq:L62}\\
\underset{4,2}{L}{}_{BCDFB'C'}={}&0.\label{eq:L42}
\end{align}
\end{subequations}
The only remaining irreducible component of $L_{AB}{}^{CDFHA'B'}$ is $\underset{2,2}{L}{}_{AB}{}^{A'B'}$.

\subsubsection{Second order part}
If we use everything above we find that the second order part of \eqref{eq:chiMaxwell} reduces to
\begin{align}
0={}&\tfrac{3}{5} \underset{3,1}{M}{}^{BCDB'} (\sTwist_{3,1} \sTwist_{2,0} \phi)_{ABCDA'B'}
 -  \tfrac{2}{5} (\sCurl_{2,2} \underset{2,2}{L}{})^{BCDB'} (\sTwist_{3,1} \sTwist_{2,0} \phi)_{ABCDA'B'}\nonumber\\
& + \tfrac{3}{5} (\sTwist_{2,2} \underset{2,2}{L}{})^{BCD}{}_{A'}{}^{B'C'} (\sTwist_{3,1} \sTwist_{2,0} \phi)_{ABCDB'C'}
 + \underset{5,1}{M}{}_{ABCDFB'} (\sTwist_{3,1} \sTwist_{2,0} \phi)^{BCDF}{}_{A'}{}^{B'}.
\end{align}
Again contracting with an arbitrary vector $T^{AA'}$ and splitting $(\sTwist_{3,1} \sTwist_{2,0} \phi)^{ABCDA'B'}T^{FC'}$ into irreducible parts we find
\begin{align}
0={}&\underset{5,1}{M}{}_{ABCDFB'} T^{(A|A'|}(\sTwist_{3,1} \sTwist_{2,0} \phi)^{BCDF)}{}_{A'}{}^{B'}\nonumber\\
& -  \tfrac{3}{5} T^{A(A'}(\sTwist_{3,1} \sTwist_{2,0} \phi)_{A}{}^{|BCD|B'C')} (\sTwist_{2,2} \underset{2,2}{L}{})_{BCDA'B'C'}\nonumber\\
& + T_{AA'} (\tfrac{3}{5} \underset{3,1}{M}{}_{BCDB'}
 -  \tfrac{2}{5} (\sCurl_{2,2} \underset{2,2}{L}{})_{BCDB'}) (\sTwist_{3,1} \sTwist_{2,0} \phi)^{ABCDA'B'}.
\end{align}
Again using the argument in Section~\ref{sec:indepspinors} we find
\begin{subequations}\label{eq:SecondOrderSystem}
\begin{align}
(\sTwist_{2,2} \underset{2,2}{L}{})_{BCD}{}^{A'B'C'}={}&0,\label{eq:TwistL22}\\
\underset{5,1}{M}{}_{ABCDFB'}={}&0,\label{eq:M51}\\
\underset{3,1}{M}{}_{BCDB'}={}&\tfrac{2}{3} (\sCurl_{2,2} \underset{2,2}{L}{})_{BCDB'}.\label{eq:M31ToL22}
\end{align}
\end{subequations}
\subsubsection{First order part}
Now, after contracting the first order part of \eqref{eq:chiMaxwell} with an arbitrary tensor $T^{A}{}_{A'}$, splitting $(\sTwist_{2,0} \phi)_{ABCA'}T_{DB'}$ into irreducible parts, and using the argument in Section~\ref{sec:indepspinors}, we find that the first order part of \eqref{eq:chiMaxwell} is equivalent to the system
\begin{subequations}
\begin{align}
\underset{2,0}{N}{}_{BC}={}&\tfrac{1}{2} (\sCurl_{1,1} \underset{1,1}{M}{})_{BC}
 -  \tfrac{1}{2} (\sDiv_{3,1} \sCurl_{2,2} \underset{2,2}{L}{})_{BC}
 -  \tfrac{3}{2} \underset{2,2}{L}{}_{(B}{}^{DB'C'}\Phi_{C)DB'C'},\label{eq:N20eq1}\\
\underset{4,0}{N}{}_{ABCD}={}&\tfrac{1}{5} (\sCurl_{3,1} \sCurl_{2,2} \underset{2,2}{L}{})_{ABCD}
 + \tfrac{3}{10} \underset{2,2}{L}{}_{(AB}{}^{B'C'}\Phi_{CD)B'C'},\label{eq:N40eq1}\\
(\sTwist_{1,1} \underset{1,1}{M}{})_{BC}{}^{A'B'}={}&12 \Lambda \underset{2,2}{L}{}_{BC}{}^{A'B'}
 -  \tfrac{9}{5} \underset{2,2}{L}{}^{DFA'B'} \Psi_{BCDF}
 -  \tfrac{6}{5} \underset{2,2}{L}{}_{BC}{}^{C'D'} \bar\Psi^{A'B'}{}_{C'D'}\nonumber\\
& + (\sCurlDagger_{3,1} \sCurl_{2,2} \underset{2,2}{L}{})_{BC}{}^{A'B'}
 - 6 \underset{2,2}{L}{}^{D}{}_{(B}{}^{C'(A'}\Phi_{C)D}{}^{B')}{}_{C'},\label{eq:TwistM11Eq1}\\
(\sTwist_{3,1} \sCurl_{2,2} \underset{2,2}{L}{})_{ABCD}{}^{A'B'}={}&3 \underset{2,2}{L}{}_{(AB}{}^{C'(A'}\Phi_{CD)}{}^{B')}{}_{C'}
 + 3 \underset{2,2}{L}{}_{(A}{}^{FA'B'}\Psi_{BCD)F}.\label{eq:TwistCurlFirstOrder}
\end{align}
\end{subequations}
The commutators \eqref{eq:DivCurl}, \eqref{eq:DivTwistCurlDaggerCurl} and \eqref{eq:CurlTwist} applied to $\underset{2,2}{L}{}{}_{AB}{}^{A'B'}$ yield
\begin{subequations}
\begin{align}
(\sDiv_{3,1} \sCurl_{2,2} \underset{2,2}{L}{})_{AB}={}&\tfrac{2}{3} (\sCurl_{1,1} \sDiv_{2,2} \underset{2,2}{L}{})_{AB}
 - 2 \underset{2,2}{L}{}_{(A}{}^{CA'B'}\Phi_{B)CA'B'},\label{eq:CommDivCurlL22}\\
(\sCurlDagger_{3,1} \sCurl_{2,2} \underset{2,2}{L}{})_{AB}{}^{A'B'}={}&-12 \Lambda \underset{2,2}{L}{}_{AB}{}^{A'B'}
 + \underset{2,2}{L}{}^{CDA'B'} \Psi_{ABCD}
 + 2 \underset{2,2}{L}{}_{AB}{}^{C'D'} \bar\Psi^{A'B'}{}_{C'D'}\nonumber\\
& -  \tfrac{3}{2} (\sDiv_{3,3} \sTwist_{2,2} \underset{2,2}{L}{})_{AB}{}^{A'B'}
 + 6 \underset{2,2}{L}{}^{C}{}_{(A}{}^{C'(A'}\Phi_{B)C}{}^{B')}{}_{C'}\nonumber\\
& + \tfrac{4}{3} (\sTwist_{1,1} \sDiv_{2,2} \underset{2,2}{L}{})_{AB}{}^{A'B'},\label{eq:CommCurlDaggerCurlL22}\\
(\sTwist_{3,1} \sCurl_{2,2} \underset{2,2}{L}{})_{ABCD}{}^{A'B'}={}&\tfrac{3}{2} (\sCurl_{3,3} \sTwist_{2,2} \underset{2,2}{L}{})_{ABCD}{}^{A'B'}
 + 3 \underset{2,2}{L}{}_{(AB}{}^{C'(A'}\Phi_{CD)}{}^{B')}{}_{C'}\nonumber\\
& + 3 \underset{2,2}{L}{}_{(A}{}^{FA'B'}\Psi_{BCD)F}.\label{eq:CommTwistCurlL22}
\end{align}
\end{subequations}
It is now clear that \eqref{eq:TwistCurlFirstOrder} is a consequence of \eqref{eq:CommTwistCurlL22} and \eqref{eq:TwistL22}. The commutators \eqref{eq:CommDivCurlL22} and \eqref{eq:CommCurlDaggerCurlL22} together with \eqref{eq:TwistL22} can be used to reduce \eqref{eq:N20eq1} and \eqref{eq:TwistM11Eq1} to
\begin{subequations}
\begin{align}
\underset{2,0}{N}{}_{BC}={}&\tfrac{1}{2} (\sCurl_{1,1} \underset{1,1}{M}{})_{BC}
 -  \tfrac{1}{3} (\sCurl_{1,1} \sDiv_{2,2} \underset{2,2}{L}{})_{BC}
 -  \tfrac{1}{2} \underset{2,2}{L}{}_{(B}{}^{DB'C'}\Phi_{C)DB'C'},\label{eq:N20eq2}\\
(\sTwist_{1,1} \underset{1,1}{M}{})_{BCA'B'}={}&- \tfrac{4}{5} \underset{2,2}{L}{}^{DF}{}_{A'B'} \Psi_{BCDF}
 + \tfrac{4}{5} \underset{2,2}{L}{}_{BC}{}^{C'D'} \bar\Psi_{A'B'C'D'}
 + \tfrac{4}{3} (\sTwist_{1,1} \sDiv_{2,2} \underset{2,2}{L}{})_{BCA'B'}.\label{eq:TwistM11Eq2}
\end{align}
\end{subequations}
Now, in view of the form of \eqref{eq:TwistM11Eq2} we make the ansatz
\begin{align}
\underset{1,1}{M}{}_{AA'}={}&2 P_{AA'}
 + \tfrac{4}{3} (\sDiv_{2,2} \underset{2,2}{L}{})_{AA'},\label{eq:M11Ansatz}
\end{align}
where $P_{AA'}$ is a new spinor field. With this choice \eqref{eq:N20eq2} and \eqref{eq:TwistM11Eq2}  reduce to
\begin{subequations}
\begin{align}
\underset{2,0}{N}{}_{BC}={}&(\sCurl_{1,1} P)_{BC}
 + \tfrac{1}{3} (\sCurl_{1,1} \sDiv_{2,2} \underset{2,2}{L}{})_{BC}
 -  \tfrac{1}{2} \underset{2,2}{L}{}_{(B}{}^{DB'C'}\Phi_{C)DB'C'},\label{eq:N20eq3}\\
(\sTwist_{1,1} P)_{BCA'B'}={}&- \tfrac{2}{5} \underset{2,2}{L}{}^{DF}{}_{A'B'} \Psi_{BCDF}
 + \tfrac{2}{5} \underset{2,2}{L}{}_{BC}{}^{C'D'} \bar\Psi_{A'B'C'D'}.\label{eq:TwistPeq1}
\end{align}
\end{subequations}
In conclusion, the third, second and first order parts of \eqref{eq:chiMaxwell} vanishes if and only if \eqref{eq:ThirdOrderSystem}, \eqref{eq:SecondOrderSystem}, \eqref{eq:N40eq1}, \eqref{eq:M11Ansatz}, \eqref{eq:N20eq3} and \eqref{eq:TwistPeq1} are satisfied.

\subsubsection{Zeroth order part}
After making irreducible decompositions of the derivatives, using \eqref{eq:TwistL22} and contracting the remaining part of \eqref{eq:chiMaxwell} with an arbitrary tensor $T^{AA'}$, splitting $T_{AA'} \phi_{CD}$ into irreducible parts, and using the argument in Section~\ref{sec:indepspinors}, we find that the order zero part of \eqref{eq:chiMaxwell} is equivalent to the system

\begin{subequations}
\begin{align}
0={}&4 \Lambda P_{BA'}
 -  \tfrac{4}{3} \Phi_{BCA'B'} P^{CB'}
 + \tfrac{2}{9} \Phi^{CD}{}_{A'}{}^{B'} (\sCurl_{2,2} \underset{2,2}{L}{})_{BCDB'}
 -  \tfrac{8}{15} \Psi_{BCDF} (\sCurl_{2,2} \underset{2,2}{L}{})^{CDF}{}_{A'}\nonumber\\*
& + \tfrac{26}{45} \underset{2,2}{L}{}^{CD}{}_{A'}{}^{B'} (\sCurl_{2,2} \Phi)_{BCDB'}
 + \tfrac{2}{9} \Phi_{B}{}^{CB'C'} (\sCurlDagger_{2,2} \underset{2,2}{L}{})_{CA'B'C'}
 + \tfrac{2}{45} \underset{2,2}{L}{}_{B}{}^{CB'C'} (\sCurlDagger_{2,2} \Phi)_{CA'B'C'}\nonumber\\*
& + \tfrac{2}{3} (\sCurlDagger_{2,0} \sCurl_{1,1} P)_{BA'}
 + \tfrac{2}{9} (\sCurlDagger_{2,0} \sCurl_{1,1} \sDiv_{2,2} \underset{2,2}{L}{})_{BA'}
 + \tfrac{8}{3} \Lambda (\sDiv_{2,2} \underset{2,2}{L}{})_{BA'}
 -  \tfrac{20}{27} \Phi_{BCA'B'} (\sDiv_{2,2} \underset{2,2}{L}{})^{CB'}\nonumber\\*
& + \tfrac{28}{9} \underset{2,2}{L}{}_{BCA'B'} (\sTwist_{0,0} \Lambda)^{CB'}
 -  \tfrac{1}{3} (\sTwist_{0,0} \underset{0,0}{N}{})_{BA'}
 -  \tfrac{1}{3} \underset{2,2}{L}{}^{CDB'C'} (\sTwist_{2,2} \Phi)_{BCDA'B'C'},\label{eq:OrderZeroEq1}\\
0={}&P^{D}{}_{A'} \Psi_{ABCD}
 + 2 \Lambda (\sCurl_{2,2} \underset{2,2}{L}{})_{ABCA'}
 + \tfrac{1}{5} (\sCurlDagger_{4,0} \sCurl_{3,1} \sCurl_{2,2} \underset{2,2}{L}{})_{ABCA'}
 + \tfrac{2}{3} \Psi_{ABCD} (\sDiv_{2,2} \underset{2,2}{L}{})^{D}{}_{A'}\nonumber\\*
& -  \tfrac{37}{75} \underset{2,2}{L}{}_{(A}{}^{D}{}_{|A'|}{}^{B'}(\sCurl_{2,2} \Phi)_{BC)DB'}
 + \tfrac{7}{30} \underset{2,2}{L}{}_{(AB}{}^{B'C'}(\sCurlDagger_{2,2} \Phi)_{C)A'B'C'}
 -  \tfrac{5}{3} \underset{2,2}{L}{}_{(AB|A'|}{}^{B'}(\sTwist_{0,0} \Lambda)_{C)B'}\nonumber\\*
& + \tfrac{7}{10} \underset{2,2}{L}{}_{(A}{}^{DB'C'}(\sTwist_{2,2} \Phi)_{BC)DA'B'C'}
 -  \Phi_{(AB|A'|}{}^{B'}P_{C)B'}
 -  \tfrac{19}{15} \Phi_{(A}{}^{D}{}_{|A'|}{}^{B'}(\sCurl_{2,2} \underset{2,2}{L}{})_{BC)DB'}\nonumber\\*
& + \tfrac{1}{6} \Phi_{(AB}{}^{B'C'}(\sCurlDagger_{2,2} \underset{2,2}{L}{})_{C)A'B'C'}
 -  \tfrac{5}{9} \Phi_{(AB|A'|}{}^{B'}(\sDiv_{2,2} \underset{2,2}{L}{})_{C)B'}
 -  \tfrac{1}{5} \Psi_{(AB}{}^{DF}(\sCurl_{2,2} \underset{2,2}{L}{})_{C)DFA'}\nonumber\\*
& + \tfrac{3}{5} \underset{2,2}{L}{}^{DF}{}_{A'}{}^{B'} (\sTwist_{4,0} \Psi)_{ABCDFB'}
 -  \tfrac{1}{2} (\sTwist_{2,0} \sCurl_{1,1} P)_{ABCA'}
 -  \tfrac{1}{6} (\sTwist_{2,0} \sCurl_{1,1} \sDiv_{2,2} \underset{2,2}{L}{})_{ABCA'}.\label{eq:OrderZeroEq2}
\end{align}
\end{subequations}
Applying the commutators repeatedly we have in general that
\begin{align}
(\sCurlDagger_{4,0} &\sCurl_{3,1} \sCurl_{2,2} \underset{2,2}{L}{})_{ABC}{}^{A'}\nonumber\\
={}&\tfrac{5}{3} \square^{A'}{}_{B'}(\sCurl_{2,2} \underset{2,2}{L}{})_{ABC}{}^{B'}
 -  \tfrac{4}{3} (\sDiv_{4,2} \sTwist_{3,1} \sCurl_{2,2} \underset{2,2}{L}{})_{ABC}{}^{A'}
 -  \square_{(A}{}^{D}(\sCurl_{2,2} \underset{2,2}{L}{})_{BC)D}{}^{A'}\nonumber\\
& + \tfrac{5}{4} (\sTwist_{2,0} \sDiv_{3,1} \sCurl_{2,2} \underset{2,2}{L}{})_{ABC}{}^{A'}\nonumber\\
={}&\tfrac{5}{3} \square^{A'}{}_{B'}(\sCurl_{2,2} \underset{2,2}{L}{})_{ABC}{}^{B'}
 - 2 (\sDiv_{4,2} \sCurl_{3,3} \sTwist_{2,2} \underset{2,2}{L}{})_{ABC}{}^{A'}
 -  \square_{(A}{}^{D}(\sCurl_{2,2} \underset{2,2}{L}{})_{BC)D}{}^{A'}\nonumber\\
& -  \nabla^{DB'}\square_{(AB}\underset{2,2}{L}{}_{C)D}{}^{A'}{}_{B'}
 -  \nabla^{DB'}\square_{(A|D|}\underset{2,2}{L}{}_{BC)}{}^{A'}{}_{B'}
 -  \tfrac{5}{4} \nabla_{(A}{}^{A'}\square^{B'C'}\underset{2,2}{L}{}_{BC)B'C'}\nonumber\\
& + \tfrac{5}{6} (\sTwist_{2,0} \sCurl_{1,1} \sDiv_{2,2} \underset{2,2}{L}{})_{ABC}{}^{A'}\nonumber\\
={}&-10 \Lambda (\sCurl_{2,2} \underset{2,2}{L}{})_{ABC}{}^{A'}
 - 2 \Psi_{ABCD} (\sDiv_{2,2} \underset{2,2}{L}{})^{DA'}
 - 2 (\sDiv_{4,2} \sCurl_{3,3} \sTwist_{2,2} \underset{2,2}{L}{})_{ABC}{}^{A'}\nonumber\\
& + \tfrac{25}{3} \underset{2,2}{L}{}_{(AB}{}^{A'B'}(\sTwist_{0,0}\Lambda)_{C)B'}
 + \tfrac{49}{15} \underset{2,2}{L}{}_{(A}{}^{DA'B'}(\sCurl_{2,2} \Phi)_{BC)DB'}
 + \tfrac{5}{6} \underset{2,2}{L}{}_{(AB}{}^{B'C'}(\sCurlDagger_{2,2} \Phi)_{C)}{}^{A'}{}_{B'C'}\nonumber\\
& -  \tfrac{7}{2} \underset{2,2}{L}{}_{(A}{}^{DB'C'}(\sTwist_{2,2} \Phi)_{BC)D}{}^{A'}{}_{B'C'}
 + \tfrac{19}{3} \Phi_{(A}{}^{DA'B'}(\sCurl_{2,2} \underset{2,2}{L}{})_{BC)DB'}\nonumber\\
& -  \tfrac{5}{6} \Phi_{(AB}{}^{B'C'}(\sCurlDagger_{2,2} \underset{2,2}{L}{})_{C)}{}^{A'}{}_{B'C'}
 + \tfrac{25}{9} \Phi_{(AB}{}^{A'B'}(\sDiv_{2,2} \underset{2,2}{L}{})_{C)B'}\nonumber\\
& + \tfrac{7}{2} \Phi_{(A}{}^{DB'C'}(\sTwist_{2,2} \underset{2,2}{L}{})_{BC)D}{}^{A'}{}_{B'C'}
 -  \Psi_{(AB}{}^{DF}(\sCurl_{2,2} \underset{2,2}{L}{})_{C)DF}{}^{A'}\nonumber\\
& -  \underset{2,2}{L}{}^{DFA'B'} (\sTwist_{4,0} \Psi)_{ABCDFB'}
 + \tfrac{5}{6} (\sTwist_{2,0} \sCurl_{1,1} \sDiv_{2,2} \underset{2,2}{L}{})_{ABC}{}^{A'}.
\end{align}
With this and \eqref{eq:TwistL22}, the equation \eqref{eq:OrderZeroEq2} reduces to
\begin{align}
0={}&P^{D}{}_{A'} \Psi_{ABCD}
 + \tfrac{4}{15} \Psi_{ABCD} (\sDiv_{2,2} \underset{2,2}{L}{})^{D}{}_{A'}
 + \tfrac{4}{25} \underset{2,2}{L}{}_{(A}{}^{D}{}_{|A'|}{}^{B'}(\sCurl_{2,2} \Phi)_{BC)DB'}\nonumber\\*
& + \tfrac{2}{5} \underset{2,2}{L}{}_{(AB}{}^{B'C'}(\sCurlDagger_{2,2} \Phi)_{C)A'B'C'}
 -  \Phi_{(AB|A'|}{}^{B'}P_{C)B'}
 -  \tfrac{2}{5} \Psi_{(AB}{}^{DF}(\sCurl_{2,2} \underset{2,2}{L}{})_{C)DFA'}\nonumber\\*
& + \tfrac{2}{5} \underset{2,2}{L}{}^{DF}{}_{A'}{}^{B'} (\sTwist_{4,0} \Psi)_{ABCDFB'}
 -  \tfrac{1}{2} (\sTwist_{2,0} \sCurl_{1,1} P)_{ABCA'}.
\end{align}
Using the commutator
\begin{align}
(\sTwist_{2,0} \sCurl_{1,1} P)_{ABCA'}={}&2 P^{D}{}_{A'} \Psi_{ABCD}
 + 2 (\sCurl_{2,2} \sTwist_{1,1} P)_{ABCA'}
 - 2 \Phi_{(AB|A'|}{}^{B'}P_{C)B'}.
\end{align}
this becomes
\begin{align}
0={}&- (\sCurl_{2,2} \sTwist_{1,1} P)_{ABCA'}
 + \tfrac{4}{15} \Psi_{ABCD} (\sDiv_{2,2} \underset{2,2}{L}{})^{D}{}_{A'}
 + \tfrac{4}{25} \underset{2,2}{L}{}_{(A}{}^{D}{}_{|A'|}{}^{B'}(\sCurl_{2,2} \Phi)_{BC)DB'}\nonumber\\
& + \tfrac{2}{5} \underset{2,2}{L}{}_{(AB}{}^{B'C'}(\sCurlDagger_{2,2} \Phi)_{C)A'B'C'}
 -  \tfrac{2}{5} \Psi_{(AB}{}^{DF}(\sCurl_{2,2} \underset{2,2}{L}{})_{C)DFA'}
 + \tfrac{2}{5} \underset{2,2}{L}{}^{DF}{}_{A'}{}^{B'} (\sTwist_{4,0} \Psi)_{ABCDFB'}.
\end{align}
However, after substituting \eqref{eq:TwistPeq1} in this equation, decomposing the derivatives into irreducible parts and using \eqref{eq:TwistL22}, this equation actually becomes trivial.

Doing the same calculations as for the Dirac-Weyl case we see that \eqref{eq:CurlDaggerCurlDivLL22Dirac} also holds for the Maxwell case. 
Directly from the commutators we find
\begin{align}
(\sCurlDagger_{2,0} \sCurl_{1,1} P)_{AA'}={}&-6 \Lambda P_{AA'}
 + 2 \Phi_{ABA'B'} P^{BB'}
 -  (\sDiv_{2,2} \sTwist_{1,1} P)_{AA'}
 + \tfrac{3}{4} (\sTwist_{0,0} \sDiv_{1,1} P)_{AA'}.
\end{align}
With this, \eqref{eq:CurlDaggerCurlDivLL22Dirac} and \eqref{eq:TwistL22} we can reduce \eqref{eq:OrderZeroEq1} to
\begin{align}
0={}&- \tfrac{2}{15} \Phi^{CD}{}_{A'}{}^{B'} (\sCurl_{2,2} \underset{2,2}{L}{})_{BCDB'}
 -  \tfrac{4}{15} \Psi_{BCDF} (\sCurl_{2,2} \underset{2,2}{L}{})^{CDF}{}_{A'}
 + \tfrac{2}{15} \underset{2,2}{L}{}^{CD}{}_{A'}{}^{B'} (\sCurl_{2,2} \Phi)_{BCDB'}\nonumber\\
& + \tfrac{4}{15} \bar\Psi_{A'B'C'D'} (\sCurlDagger_{2,2} \underset{2,2}{L}{})_{B}{}^{B'C'D'}
 -  \tfrac{2}{15} \Phi_{B}{}^{CB'C'} (\sCurlDagger_{2,2} \underset{2,2}{L}{})_{CA'B'C'}
 -  \tfrac{2}{5} \underset{2,2}{L}{}_{B}{}^{CB'C'} (\sCurlDagger_{2,2} \Phi)_{CA'B'C'}\nonumber\\
& -  \tfrac{4}{45} \Phi_{BCA'B'} (\sDiv_{2,2} \underset{2,2}{L}{})^{CB'}
 -  \tfrac{2}{3} (\sDiv_{2,2} \sTwist_{1,1} P)_{BA'}
 + \tfrac{4}{15} \underset{2,2}{L}{}_{BCA'B'} (\sTwist_{0,0} \Lambda)^{CB'}
 -  \tfrac{1}{3} (\sTwist_{0,0} \underset{0,0}{N}{})_{BA'}\nonumber\\
& -  \tfrac{1}{5} \underset{2,2}{L}{}^{CDB'C'} (\sTwist_{2,2} \Phi)_{BCDA'B'C'}
 + \tfrac{1}{2} (\sTwist_{0,0} \sDiv_{1,1} P)_{BA'}
 + \tfrac{2}{15} (\sTwist_{0,0} \sDiv_{1,1} \sDiv_{2,2} \underset{2,2}{L}{})_{BA'}.
\end{align}
Using \eqref{eq:TwistPeq1} and the irreducible decompositions, we find
\begin{align}
(\sDiv_{2,2} \sTwist_{1,1} P)_{BA'}={}&- \tfrac{2}{5} \Psi_{BCDF} (\sCurl_{2,2} \underset{2,2}{L}{})^{CDF}{}_{A'}
 + \tfrac{2}{5} \underset{2,2}{L}{}^{CD}{}_{A'}{}^{B'} (\sCurl_{2,2} \Phi)_{BCDB'}\nonumber\\
& + \tfrac{2}{5} \bar\Psi_{A'B'C'D'} (\sCurlDagger_{2,2} \underset{2,2}{L}{})_{B}{}^{B'C'D'}
 -  \tfrac{2}{5} \underset{2,2}{L}{}_{B}{}^{CB'C'} (\sCurlDagger_{2,2} \Phi)_{CA'B'C'}.\label{eq:DivTwistQ2}
\end{align}
To simplify the remaining terms, we use the same trick as for the Dirac-Weyl case. The definition \eqref{eq:UpsilonDef} and the equation \eqref{eq:GradUpsilon} can be used together with \eqref{eq:DivTwistQ2}, to reduce the equation \eqref{eq:OrderZeroEq2} to
\begin{align}
0={}&- \tfrac{1}{3} (\sTwist_{0,0} \underset{0,0}{N}{})_{BA'}
 -  \tfrac{1}{5} (\sTwist_{0,0} \Upsilon)_{BA'}
 + \tfrac{1}{2} (\sTwist_{0,0} \sDiv_{1,1} P)_{BA'}
 + \tfrac{2}{15} (\sTwist_{0,0} \sDiv_{1,1} \sDiv_{2,2} \underset{2,2}{L}{})_{BA'}.\label{eq:OrderZeroEq3}
\end{align}
We therefore make the ansatz
\begin{align}
\underset{0,0}{N}{}={}&3 Q
 -  \tfrac{3}{5} \Upsilon
 + \tfrac{3}{2} (\sDiv_{1,1} P)
 + \tfrac{2}{5} (\sDiv_{1,1} \sDiv_{2,2} \underset{2,2}{L}{}).\label{eq:N00Ansatz}
\end{align}
Now, \eqref{eq:OrderZeroEq3} becomes
\begin{align}
0={}&(\sTwist_{0,0} Q)_{AA'}.\label{eq:OrderZeroEq4}
\end{align}
\subsubsection{Potential representation}
From all this we can conclude that the only equations that restrict the geometry are \eqref{eq:TwistL22} and \eqref{eq:TwistPeq1}. Now, the operator takes the form
\begin{align}
\chi_{AB}={}&\tfrac{1}{3} \underset{0,0}{N}{} \phi_{AB}
 + \underset{4,0}{N}{}_{ABCD} \phi^{CD}
 -  \underset{2,0}{N}{}_{(A}{}^{C}\phi_{B)C}
 -  \tfrac{4}{5} (\sCurl_{2,2} \underset{2,2}{L}{})_{(A}{}^{CDA'}(\sTwist_{2,0} \phi)_{B)CDA'}\nonumber\\
& + \tfrac{1}{2} \underset{1,1}{M}{}_{CA'} (\sTwist_{2,0} \phi)_{AB}{}^{CA'}
 + \tfrac{3}{5} \underset{2,2}{L}{}^{CDA'B'} (\sTwist_{3,1} \sTwist_{2,0} \phi)_{ABCDA'B'}.
\label{eq:chiform1}
\end{align}
where $\underset{0,0}{N}$, $\underset{2,0}{N}{}_{AB}$, $\underset{4,0}{N}{}_{ABCD}$, $\underset{1,1}{M}{}_{AA'}$ are given by \eqref{eq:N00Ansatz}, \eqref{eq:N20eq3}, \eqref{eq:N40eq1} and \eqref{eq:M11Ansatz} respectively.

We can in fact simplify this expression by defining the following spinor
\begin{align}
A_{AA'}\equiv{}&P_{BA'} \phi_{A}{}^{B}
 + \tfrac{1}{5} \phi^{BC} (\sCurl_{2,2} \underset{2,2}{L}{})_{ABCA'}
 + \tfrac{4}{15} \phi_{A}{}^{B} (\sDiv_{2,2} \underset{2,2}{L}{})_{BA'}
 + \tfrac{3}{5} \underset{2,2}{L}{}_{BCA'B'} (\sTwist_{2,0} \phi)_{A}{}^{BCB'}.
\end{align}
Substituting this onto the following, and comparing with \eqref{eq:chiform1}, we find
\begin{align}
(\sCurl_{1,1} A)_{AB}={}&- Q \phi_{AB}
 + \chi_{AB}
 -  \tfrac{1}{15} \phi_{(A}{}^{C}(\sCurl_{1,1} \sDiv_{2,2} \underset{2,2}{L}{})_{B)C}
 + \tfrac{1}{10} \phi_{(A}{}^{C}(\sDiv_{3,1} \sCurl_{2,2} \underset{2,2}{L}{})_{B)C}\nonumber\\
& -  \tfrac{1}{10} \underset{2,2}{L}{}_{(A}{}^{CA'B'}\Phi_{|C}{}^{D}{}_{A'B'|}\phi_{B)D}
 -  \tfrac{1}{10} \underset{2,2}{L}{}^{CDA'B'}\Phi_{(A|CA'B'|}\phi_{B)D}\nonumber\\
={}&- Q \phi_{AB}
 + \chi_{AB},\label{eq:CurlAeq1}
\end{align}
where the last equality follows from a commutator relation. In fact the coefficients in $A_{AA'}$ were initially left free, and then chosen so all first and second order derivatives of $\phi_{AB}$ where eliminated in \eqref{eq:CurlAeq1}.

We also get
\begin{align}
(\sCurlDagger_{1,1} A)_{A'B'}={}&\tfrac{12}{5} \Lambda \underset{2,2}{L}{}_{ABA'B'} \phi^{AB}
 -  \tfrac{3}{5} \underset{2,2}{L}{}^{CD}{}_{A'B'} \Psi_{ABCD} \phi^{AB}
 + \tfrac{1}{5} \phi^{AB} (\sCurlDagger_{3,1} \sCurl_{2,2} \underset{2,2}{L}{})_{ABA'B'}\nonumber\\
& -  \tfrac{6}{5} \underset{2,2}{L}{}^{AB}{}_{(A'}{}^{C'}\Phi_{|A|}{}^{C}{}_{B')C'}\phi_{BC}
 -  \phi^{AB} (\sTwist_{1,1} P)_{ABA'B'}\nonumber\\
& -  \tfrac{3}{5} (\sTwist_{2,2} \underset{2,2}{L}{})_{ABCA'B'C'} (\sTwist_{2,0} \phi)^{ABCC'}
 -  \tfrac{4}{15} \phi^{AB} (\sTwist_{1,1} \sDiv_{2,2} \underset{2,2}{L}{})_{ABA'B'}\nonumber\\
={}&0.
\end{align}
where we in the last step used \eqref{eq:TwistPeq1}, a commutator and \eqref{eq:TwistL22}

To get the highest order coefficient equal to 1 in $A_{AA'}$ and in $\chi_{AB}$, we define a new symmetric spinor, which is just a rescaling of $\underset{2,2}{L}{}_{ABA'B'}$
\begin{align}
L_{ABA'B'}\equiv{}&\tfrac{3}{5} \underset{2,2}{L}{}_{ABA'B'}.
\end{align}
Now, the only equations we have left are
\begin{subequations}
\begin{align}
(\sTwist_{2,2} L)_{ABC}{}^{A'B'C'}={}&0,\\
(\sTwist_{1,1} P)_{AB}{}^{A'B'}={}&- \tfrac{2}{3} (\ObstrOne L)_{AB}{}^{A'B'},\\
(\sTwist_{0,0} Q)_{BA'}={}&0,\\
A_{AA'}={}&P_{BA'} \phi_{A}{}^{B}
 + \tfrac{1}{3} \phi^{BC} (\sCurl_{2,2} L)_{ABCA'}
 + \tfrac{4}{9} \phi_{A}{}^{B} (\sDiv_{2,2} L)_{BA'}\nonumber\\
& -  L^{BC}{}_{A'}{}^{B'} (\sTwist_{2,0} \phi)_{ABCB'}.
\end{align}
\end{subequations}
\end{proof}

\subsection{Second kind of symmetry operator for the Maxwell equation}\label{subsection:second_kind}
For the symmetry operators of the second kind, one can follow the same procedure as above. However, this case was completely handled in \cite{KalMcLWil92a}. In that paper it was shown that a symmetry operator of the second kind always has the form $\phi_{AB}\rightarrow \omega_{A'B'}$,
\begin{align}
\omega_{A'B'}={}&\tfrac{3}{5} \phi^{CD} (\sCurlDagger_{3,1} \sCurlDagger_{4,0} L)_{CDA'B'}
 -  \tfrac{8}{5} (\sCurlDagger_{4,0} L)^{CDF}{}_{(A'}(\sTwist_{2,0} \phi)_{|CDF|B')}\nonumber\\
& + L^{CDFH} (\sTwist_{3,1} \sTwist_{2,0} \phi)_{CDFHA'B'},
\end{align}
where $L_{ABCD}=L_{(ABCD)}$ satisfies
\begin{align}
(\sTwist_{4,0} L)_{ABCDF}{}^{A'}={}&0.\label{eq:LSecKillingSpinor}
\end{align}
Hence, the treatment in \cite{KalMcLWil92a} is satisfactory. However, it is interesting to see if the operator can be written in terms of a potential. Let
\begin{align}
B_{AA'}\equiv{}&\tfrac{3}{5} \phi^{BC} (\sCurlDagger_{4,0} L)_{ABCA'}
 + L_{ABCD} (\sTwist_{2,0} \phi)^{BCD}{}_{A'}.\label{eq:BDef}
\end{align}
Then, from the definition of $\sCurlDagger$, the irreducible decompositions and \eqref{eq:LSecKillingSpinor} we get
\begin{align}
(\sCurlDagger_{1,1} B)_{A'B'}={}&\tfrac{3}{5} \phi^{AB} (\sCurlDagger_{3,1} \sCurlDagger_{4,0} L)_{ABA'B'}
 -  \tfrac{8}{5} (\sCurlDagger_{4,0} L)^{ABC}{}_{(A'}(\sTwist_{2,0} \phi)_{|ABC|B')}\nonumber\\
& + L^{ABCD} (\sTwist_{3,1} \sTwist_{2,0} \phi)_{ABCDA'B'}\nonumber\\
={}&\omega_{A'B'}.\label{eq:CurlDaggerB}
\end{align}
The coefficients in \eqref{eq:BDef} where initially left free, and then chosen to get \eqref{eq:CurlDaggerB}.

We also get
\begin{align}
(\sCurl_{1,1} B)_{AB}={}&6 \Lambda L_{ABCD} \phi^{CD}
 -  \tfrac{3}{2} L_{AB}{}^{FH} \Psi_{CDFH} \phi^{CD}
 + \tfrac{3}{5} \phi^{CD} (\sCurl_{3,1} \sCurlDagger_{4,0} L)_{ABCD}\nonumber\\
& + \tfrac{3}{10} \phi_{(A}{}^{C}(\sDiv_{3,1} \sCurlDagger_{4,0} L)_{B)C}
 -  \tfrac{3}{2} L_{(A}{}^{CDF}\Psi_{B)CD}{}^{H}\phi_{FH}
 -  \tfrac{1}{2} L_{(A}{}^{CDF}\Psi_{|CDF|}{}^{H}\phi_{B)H}\nonumber\\
& -  (\sTwist_{4,0} L)_{ABCDFA'} (\sTwist_{2,0} \phi)^{CDFA'}\nonumber\\
={}&- \tfrac{1}{2} \phi^{CD} (\sDiv_{5,1} \sTwist_{4,0} L)_{ABCD}
 -  \tfrac{4}{5} \phi_{B}{}^{C} L_{(A}{}^{DFH}\Psi_{C)DFH}
 -  \tfrac{4}{5} \phi_{A}{}^{C} L_{(B}{}^{DFH}\Psi_{C)DFH}\nonumber\\
& -  (\sTwist_{4,0} L)_{ABCDFA'} (\sTwist_{2,0} \phi)^{CDFA'}\nonumber\\
={}&0.
\end{align}
Here, we have used \eqref{eq:LSecKillingSpinor} together with the irreducible decomposition of $L_{AB}{}^{FH} \Psi_{CDFH}$ and the relations
\begin{subequations}
\begin{align}
(\sDiv_{3,1} \sCurlDagger_{4,0} L)_{AB}={}&-2 L_{(A}{}^{CDF}\Psi_{B)CDF},\\
(\sCurl_{3,1} \sCurlDagger_{4,0} L)_{ABCD}={}&-10 \Lambda L_{ABCD}
 -  \tfrac{5}{6} (\sDiv_{5,1} \sTwist_{4,0} L)_{ABCD}
 + 5 L_{(AB}{}^{FH}\Psi_{CD)FH},\\
L_{(B}{}^{DFH}\Psi_{C)DFH}={}&0.
\end{align}
\end{subequations}
The last equation follows from the integrability condition (cf. Section~\ref{sec:integrabilitycond})
\begin{align}
L_{(ABC}{}^{L}\Psi_{DFH)L} ={}& - \tfrac{1}{4} (\sCurl_{5,1} \sTwist_{4,0} L)_{ABCDFH}=0,\label{eq:intcondKS4}
\end{align}
as explained in \cite{KalMcLWil92a}.

\section{Factorizations} \label{sec:factorizations} 
In this section we will consider special cases for which the auxiliary conditions will always have a solution. 
We will now prove Proposition~\ref{prop:factorize}, considering each case in turn.

\subsection{The case when $L_{ABA'B'}$ factors in terms of conformal Killing vectors}
\begin{proof}[Proof of Proposition~\ref{prop:factorize} part (i)]
If $\xi_{AA'}$ and $\zeta_{AA'}$ are conformal Killing vectors, i.e. 
\begin{align}
(\sTwist_{1,1}\xi)_{AB}{}^{A'B'}&=0,&
(\sTwist_{1,1}\zeta)_{AB}{}^{A'B'}&=0,\label{eq:xizetaCKV}
\end{align} 
then we have a solution 
\begin{align}
\mathfrak{L}_{\xi\zeta AB}{}^{A'B'}\equiv{}&\zeta_{(A}{}^{(A'}\xi_{B)}{}^{B')}\label{eq:Lxizetadef}
\end{align}
to the equation
\begin{align}
(\sTwist_{2,2}\mathfrak{L}_{\xi\zeta})_{ABC}{}^{A'B'C'}&=0.
\end{align}
Let
\begin{subequations}
\begin{align}
\mathfrak{Q}_{\xi\zeta}\equiv{}&\Lambda \zeta^{AA'} \xi_{AA'}
 + \tfrac{1}{3} \Phi_{ABA'B'} \zeta^{AA'} \xi^{BB'}
 + \tfrac{1}{8} (\sCurl_{1,1} \zeta)^{AB} (\sCurl_{1,1} \xi)_{AB}\nonumber\\
& + \tfrac{1}{6} \xi^{AA'} (\sCurl_{0,2} \sCurlDagger_{1,1} \zeta)_{AA'}
 + \tfrac{1}{6} \zeta^{AA'} (\sCurl_{0,2} \sCurlDagger_{1,1} \xi)_{AA'}
 + \tfrac{1}{8} (\sCurlDagger_{1,1} \zeta)^{A'B'} (\sCurlDagger_{1,1} \xi)_{A'B'}\nonumber\\
& -  \tfrac{1}{32} (\sDiv_{1,1} \zeta) (\sDiv_{1,1} \xi),\label{eq:Qxizetadef}\\
\mathfrak{P}_{\xi\zeta AA'}\equiv{}&\tfrac{1}{4} \xi^{B}{}_{A'} (\sCurl_{1,1} \zeta)_{AB}
 + \tfrac{1}{4} \zeta^{B}{}_{A'} (\sCurl_{1,1} \xi)_{AB}
 -  \tfrac{1}{4} \xi_{A}{}^{B'} (\sCurlDagger_{1,1} \zeta)_{A'B'}
 -  \tfrac{1}{4} \zeta_{A}{}^{B'} (\sCurlDagger_{1,1} \xi)_{A'B'}.\label{eq:Pxizetadef}
\end{align}
\end{subequations}

Applying the $\sTwist$ operator to the equation \eqref{eq:Qxizetadef}, decomposing the derivatives into irreducible parts and using \eqref{eq:xizetaCKV} gives a long expression with the operators $\sDiv$, $\sCurl$, $\sCurlDagger$, $\sCurl\sCurlDagger$, $\sCurlDagger\sCurl$, $\sTwist\sDiv$, $\sTwist\sCurl$, $\sTwist\sCurlDagger$, $\sDiv\sCurl\sCurlDagger$, $\sCurl\sCurl\sCurlDagger$ and $\sCurlDagger\sCurl\sCurlDagger$ operating on $\xi^{AA'}$ and $\zeta^{AA'}$. Using the commutators \eqref{eq:DivCurl}, \eqref{eq:CurlTwist}, \eqref{eq:CurlDaggerTwist}, \eqref{eq:DivTwistCurlCurlDagger} and \eqref{eq:DivTwistCurlDaggerCurl} on the outermost operators and using \eqref{eq:xizetaCKV}, the list of operators appearing can be reduced to the set $\sDiv$, $\sCurl$, $\sCurlDagger$, $\sDiv\sTwist\sCurlDagger$, $\sCurl\sTwist\sCurlDagger$ and $\sCurl\sCurl\sCurlDagger$. Then using the relations \eqref{eq:CurlDaggerTwist} and \eqref{eq:DivTwistCurlCurlDagger} on the innermost operators the list of operators appearing is reduced to $\sCurl$, $\sCurlDagger$, $\sCurl\sTwist\sDiv$, where the latter can be eliminated with \eqref{eq:CurlTwist} on the outer operators. After making an irreducible decomposition of $\xi^{AA'}\zeta^{BB'}$ and identifying the symmetric part though \eqref{eq:Lxizetadef}, one is left with
\begin{align}
(\sTwist_{0,0} \mathfrak{Q}_{\xi\zeta})_{A}{}^{A'}={}&\mathfrak{L}_{\xi\zeta}{}^{BCA'B'} (\sCurl_{2,2} \Phi)_{ABCB'}
 + \tfrac{1}{4} \Psi_{ABCD} \xi^{BA'} (\sCurl_{1,1} \zeta)^{CD}
 + \tfrac{1}{4} \Psi_{ABCD} \zeta^{BA'} (\sCurl_{1,1} \xi)^{CD}\nonumber\\
& + \mathfrak{L}_{\xi\zeta}{}_{A}{}^{BB'C'} (\sCurlDagger_{2,2} \Phi)_{B}{}^{A'}{}_{B'C'}
 + \tfrac{1}{4} \bar\Psi^{A'}{}_{B'C'D'} \xi_{A}{}^{B'} (\sCurlDagger_{1,1} \zeta)^{C'D'}\nonumber\\
& + \tfrac{1}{4} \bar\Psi^{A'}{}_{B'C'D'} \zeta_{A}{}^{B'} (\sCurlDagger_{1,1} \xi)^{C'D'}.\label{eq:TwistQxizetaEq1}
\end{align}

Applying the $\sTwist$ operator to the equation \eqref{eq:Pxizetadef}, decomposing the derivatives into irreducible parts and using \eqref{eq:xizetaCKV} gives a expression with the operators  $\sCurl\sCurlDagger$, $\sCurlDagger\sCurl$, $\sTwist\sCurl$, $\sTwist\sCurlDagger$ operating on $\xi^{AA'}$ and $\zeta^{AA'}$. Using the commutators \eqref{eq:CurlTwist}, \eqref{eq:CurlDaggerTwist} and \eqref{eq:CurlCurlDagger} and using \eqref{eq:xizetaCKV}, the entire expression can be reduced to only contain curvature terms. After making an irreducible decomposition of $\xi^{AA'}\zeta^{BB'}$ and identifying the symmetric part though \eqref{eq:Lxizetadef}, one is left with
\begin{align}
(\sTwist_{1,1} \mathfrak{P}_{\xi\zeta}{})_{AB}{}^{A'B'}={}&\mathfrak{L}_{\xi\zeta}{}^{CDA'B'} \Psi_{ABCD}
 -  \mathfrak{L}_{\xi\zeta}{}_{AB}{}^{C'D'} \bar\Psi^{A'B'}{}_{C'D'}.\label{eq:TwistPxizetaEq1}
\end{align}
Substituting \eqref{eq:Lxizetadef} into the definition of $\ObstrZero$, allows us to see that  \eqref{eq:TwistQxizetaEq1} and \eqref{eq:TwistPxizetaEq1} reduces to
\begin{subequations}
\begin{align}
(\sTwist_{0,0} \mathfrak{Q}_{\xi\zeta})_{A}{}^{A'}={}&(\ObstrZero \mathfrak{L}_{\xi\zeta})_{A}{}^{A'},\label{eq:Qxizetaaux}\\
(\sTwist_{1,1} \mathfrak{P}_{\xi\zeta})_{AB}{}^{A'B'}={}&(\ObstrOne \mathfrak{L}_{\xi\zeta})_{AB}{}^{A'B'}.\label{eq:Pxizetaaux}
\end{align}
\end{subequations}
The actual form of \eqref{eq:Qxizetadef} and \eqref{eq:Pxizetadef} was obtained by making sufficiently general symmetric second order bi-linear ans\"atze. The coefficients where then chosen to eliminate as many extra terms as possible in \eqref{eq:Qxizetaaux} and \eqref{eq:Pxizetaaux}.
\end{proof}

\subsection{The case when $L_{ABA'B'}$ factors in terms of Killing spinors}
Another way of constructing conformal Killing tensors is to make a product of valence $(2,0)$ and valence $(0,2)$ Killing spinors. It turns out that also this case admits solutions to the auxiliary conditions.

In principle we could construct $L_{ABA'B'}$ from two different Killing spinors, but if the dimension of the space of Killing spinors is greater than one, the spacetime has to be locally isometric to Minkowski space. In these spacetimes the picture is much simpler and has been studied before. The auxiliary conditions will be trivial in these cases. We will therefore only consider one Killing spinor.
\begin{proof}[Proof of Proposition~\ref{prop:factorize} part (ii)]
Let $\kappa_{AB}$ be a Killing spinor, i.e. a solution to 
\begin{align}
(\sTwist_{2,0} \kappa)_{ABCA'}={}&0.\label{eq:kappaKS}
\end{align}
We have a solution 
\begin{align}\label{eq:Lkappadef}
\mathfrak{L}_{\kappa AB}{}^{A'B'}\equiv{}&\kappa_{AB}\bar\kappa^{A'B'},
\end{align}
to the equation
\begin{align}
(\sTwist_{2,2}\mathfrak{L}_{\kappa})_{ABC}{}^{A'B'C'}&=0.
\end{align}
Now, let
\begin{subequations}
\begin{align}
\mathfrak{Q}_{\kappa}\equiv{}&\tfrac{2}{3} \Phi_{ABA'B'} \kappa^{AB} \bar{\kappa}^{A'B'}
 + \tfrac{1}{9} \kappa^{AB} (\sCurl_{1,1} \sCurl_{0,2} \bar{\kappa})_{AB}
 + \tfrac{4}{27} (\sCurl_{0,2} \bar{\kappa})^{AA'} (\sCurlDagger_{2,0} \kappa)_{AA'}\nonumber\\
& + \tfrac{1}{9} \bar{\kappa}^{A'B'} (\sCurlDagger_{1,1} \sCurlDagger_{2,0} \kappa)_{A'B'},\label{eq:Qkappadef}\\
\mathfrak{P}_{\kappa AA'}\equiv{}&\tfrac{4}{3} \kappa_{AB} (\sCurl_{0,2} \bar{\kappa})^{B}{}_{A'}
 -  \tfrac{4}{3} \bar{\kappa}_{A'B'} (\sCurlDagger_{2,0} \kappa)_{A}{}^{B'}.\label{eq:Pkappadef}
\end{align}
\end{subequations}

Applying the $\sTwist$ operator to the equation \eqref{eq:Qkappadef}, decomposing the derivatives into irreducible parts and using \eqref{eq:kappaKS} gives a long expression with the operators $\sCurl$, $\sCurlDagger$, $\sDiv\sCurl$, $\sDiv\sCurlDagger$, $\sCurl\sCurlDagger$, $\sCurlDagger\sCurl$, $\sTwist\sCurl$, $\sTwist\sCurlDagger$, $\sCurl\sCurlDagger\sCurlDagger$, $\sCurlDagger\sCurl\sCurl$, $\sTwist\sCurl\sCurl$ and $\sTwist\sCurlDagger\sCurlDagger$ operating on $\kappa_{AB}$ and $\bar\kappa_{A'B'}$. Using the commutators \eqref{eq:DivCurl}, \eqref{eq:DivCurlDagger}, \eqref{eq:CurlTwist}, \eqref{eq:CurlDaggerTwist}, \eqref{eq:DivTwistCurlCurlDagger} and \eqref{eq:DivTwistCurlDaggerCurl} on the outermost operators and using \eqref{eq:kappaKS}, the list of operators appearing can be reduced to the set $\sCurl$, $\sCurlDagger$, $\sCurl\sTwist\sCurl$, $\sCurlDagger\sTwist\sCurlDagger$, $\sDiv\sTwist\sCurl$ and $\sDiv\sTwist\sCurlDagger$. Then using the relations \eqref{eq:DivCurl}, \eqref{eq:DivCurlDagger}, \eqref{eq:CurlTwist} and \eqref{eq:CurlDaggerTwist} on the innermost operators the expression will only contain the operators $\sCurl$, $\sCurlDagger$.
\begin{align}
(\sTwist_{0,0} \mathfrak{Q}_{\kappa})_{A}{}^{A'}={}&\kappa^{BC} \bar{\kappa}^{A'B'} (\sCurl_{2,2} \Phi)_{ABCB'}
 -  \tfrac{2}{9} \Phi_{BC}{}^{A'}{}_{B'} \kappa^{BC} (\sCurl_{0,2} \bar{\kappa})_{A}{}^{B'}
 + \tfrac{1}{3} \Psi_{ABCD} \kappa^{CD} (\sCurl_{0,2} \bar{\kappa})^{BA'}\nonumber\\
& + \tfrac{2}{9} \Phi_{BC}{}^{A'}{}_{B'} \kappa_{A}{}^{C} (\sCurl_{0,2} \bar{\kappa})^{BB'}
 -  \tfrac{2}{9} \Phi_{AC}{}^{A'}{}_{B'} \kappa_{B}{}^{C} (\sCurl_{0,2} \bar{\kappa})^{BB'}\nonumber\\
& + \kappa_{A}{}^{B} \bar{\kappa}^{B'C'} (\sCurlDagger_{2,2} \Phi)_{B}{}^{A'}{}_{B'C'}
 + \tfrac{1}{3} \bar\Psi^{A'}{}_{B'C'D'} \bar{\kappa}^{C'D'} (\sCurlDagger_{2,0} \kappa)_{A}{}^{B'}\nonumber\\
& -  \tfrac{2}{9} \Phi_{ABB'C'} \bar{\kappa}^{B'C'} (\sCurlDagger_{2,0} \kappa)^{BA'}
 + \tfrac{2}{9} \Phi_{ABB'C'} \bar{\kappa}^{A'C'} (\sCurlDagger_{2,0} \kappa)^{BB'}\nonumber\\
& -  \tfrac{2}{9} \Phi_{AB}{}^{A'}{}_{C'} \bar{\kappa}_{B'}{}^{C'} (\sCurlDagger_{2,0} \kappa)^{BB'}.\label{eq:TwistQkappaEq1}
\end{align}

Applying the $\sTwist$ operator to the equation \eqref{eq:Pkappadef}, decomposing the derivatives into irreducible parts and using \eqref{eq:kappaKS} gives an expression with the operators $\sCurl\sCurlDagger$, $\sCurlDagger\sCurl$, $\sTwist\sCurl$ and $\sTwist\sCurlDagger$ operating on $\kappa_{AB}$ and $\bar\kappa_{A'B'}$. Using the commutators \eqref{eq:CurlTwist}, \eqref{eq:CurlDaggerTwist}, \eqref{eq:DivTwistCurlCurlDagger} and \eqref{eq:DivTwistCurlDaggerCurl} and using \eqref{eq:kappaKS}, the expression reduces to
\begin{align}
(\sTwist_{1,1} \mathfrak{P}_{\kappa}{})_{AB}{}^{A'B'}={}&\Psi_{ABCD} \kappa^{CD} \bar{\kappa}^{A'B'}
 -  \bar\Psi^{A'B'}{}_{C'D'} \kappa_{AB} \bar{\kappa}^{C'D'}.\label{eq:TwistPkappaEq1}
\end{align}

Substituting \eqref{eq:Lkappadef} into the definition of $\ObstrZero$, and making an irreducible decomposition of $\kappa_{AB} (\sCurl_{0,2} \bar{\kappa})_{C}{}^{B'}$ and $\bar{\kappa}_{A'B'} (\sCurlDagger_{2,0} \kappa)_{AC'}$, allows us to see that \eqref{eq:TwistQkappaEq1} and \eqref{eq:TwistPkappaEq1} reduces to
\begin{subequations}
\begin{align}
(\sTwist_{0,0} \mathfrak{Q}_{\kappa})_{A}{}^{A'}={}&(\ObstrZero \mathfrak{L}_{\kappa})_{A}{}^{A'},\\
(\sTwist_{1,1} \mathfrak{P}_{\kappa})_{AB}{}^{A'B'}={}&(\ObstrOne \mathfrak{L}_{\kappa})_{AB}{}^{A'B'}.
\end{align}
\end{subequations}
\end{proof}

\subsection{Example of a conformal Killing tensor that does not factor}
The following shows that the condition \ref{point:A0} is non-trivial. We also see that \ref{point:A1} does not imply \ref{point:A0}. Unfortunately, we have not found any example of a valence $(1,1)$ Killing spinor which does not satisfy \ref{point:A1}.

Consider the following St\"ackel metric (see \cite{michel:radoux:silhan:2013arXiv1308.1046M} and \cite{kalnins:mclenaghan:BelgAcad1984} for more general examples.)
\begin{equation}
g_{ab} = dt^2-  dz^2 -  (x + y)(dx^2 + dy^2) 
\end{equation}
with the tetrad
\begin{align*}
l^{a}={}&\tfrac{1}{\sqrt{2}}(\partial_t)^{a}
 + \tfrac{1}{\sqrt{2}}(\partial_z)^{a},&
n^{a}={}&\tfrac{1}{\sqrt{2}}(\partial_t)^{a}
 -  \tfrac{1}{\sqrt{2}}(\partial_z)^{a},&
m^{a}={}&\frac{(\partial_x)^{a}}{\sqrt{2} (x + y)^{1/2}}
 + \frac{i (\partial_y)^{a}}{\sqrt{2} (x + y)^{1/2}}.
\end{align*}
Expressed in the corresponding dyad $(o_A,\iota_A)$, the curvature takes the form
\begin{align}
\Psi_{ABCD}={}&-12 \Lambda o_{(A}o_{B}\iota_{C}\iota_{D)},&
\Phi_{ABA'B'}={}&12 \Lambda o_{(A}\iota_{B)} \bar o_{(A'}\bar\iota_{B')},&
\Lambda={}& \frac{1}{12 (x + y)^3}.
\end{align}
We can see that the spinor
\begin{align}
L_{AB}{}^{A'B'}={}&\tfrac{1}{2} (x + y) 
(\bar o^{A'} \bar o^{B'} \iota_{A} \iota_{B}
 + o_{A} o_{B} \bar\iota^{A'} \bar\iota^{B'})
 - (x - y) o_{(A}\iota_{B)} \bar o^{(A'}\bar\iota^{B')}
\end{align}
is a trace-free conformal Killing tensor. We trivially have solutions to the auxiliary condition \ref{point:A1} because
\begin{align} 
(\ObstrOne L)_{AB}{}^{A'B'} = {}&L{}^{CDA'B'} \Psi_{ABCD}
 - L{}_{AB}{}^{C'D'} \bar\Psi^{A'B'}{}_{C'D'} = 0.
\end{align}
If there is a solution to \eqref{eq:A0} we will automatically have $(\sCurl_{1,1}\ObstrZero L)_{AB}=0$ because 
$\sCurl_{1,1}\sTwist_{0,0}=0$. 
However, with the current $L_{AB}{}^{A'B'}$ we get
\begin{align} 
(\sCurl_{1,1} \ObstrZero L)_{AB}={}&\frac{5i o_{(A}\iota_{B)}}{(x + y)^5}.
\end{align}
This is non vanishing, which means that the auxiliary condition \ref{point:A0} does not admit a solution. This example shows that the conditions \ref{point:A0} and \ref{point:A1} are not equivalent.
From the previous two sections, we can also conclude that this $L_{AB}{}^{A'B'}$ can not be written as a linear combination of conformal Killing tensors of the form $\zeta_{(A}{}^{(A'}\xi_{B)}{}^{B')}$ or $\kappa_{AB} \bar{\kappa}_{A'B'}$. For the more general metric in \cite{kalnins:mclenaghan:BelgAcad1984} we can in fact also construct a valence $(2,2)$ Killing spinor which trivially satisfies condition \ref{point:A1}, but which in general will not satisfy condition \ref{point:A0}. It is interesting to note that in general this metric does not admit Killing vectors, but we can still construct symmetry operators for the Maxwell equation.

\subsection{Auxiliary condition for a symmetry operator of the second kind for the Dirac-Weyl equation}
\begin{proof}[Proof of Proposition~\ref{prop:factorize} part (iii)]
Let $\kappa_{AB}$ be a Killing spinor, and $\xi^{AA'}$ a conformal Killing vector, i.e. 
\begin{align}
(\sTwist_{2,0} \kappa)_{ABCA'}={}&0, &(\sTwist_{1,1}\xi)_{AB}{}^{A'B'}&=0.\label{eq:TwistkappaTwistxi}
\end{align}
then we have a solution 
\begin{align}
\mathfrak{L}_{\kappa\xi ABC}{}^{A'}\equiv{}&\kappa_{(AB}\xi_{C)}{}^{A'}
\end{align}
to the equation
\begin{align}
(\sTwist_{3,1}\mathfrak{L}_{\kappa\xi})_{ABCD}{}^{A'B'}&=0.
\end{align}
The auxiliary equation \eqref{eq:A1/2*} now takes the form
\begin{align}
0={}&\tfrac{3}{4} \Psi_{ABDF} \kappa^{CD} (\sCurl_{1,1} \xi)_{C}{}^{F}
 + \Psi_{ABCD} \xi^{CA'} (\sCurlDagger_{2,0} \kappa)^{D}{}_{A'}
 -  \tfrac{3}{4} \Psi_{ABCD} \kappa^{CD} (\sDiv_{1,1} \xi)\nonumber\\
& -  \tfrac{5}{4} \Psi_{(A}{}^{CDF}\kappa_{B)C}(\sCurl_{1,1} \xi)_{DF}
 -  \tfrac{5}{4} \Psi_{(A}{}^{CDF}\kappa_{|CD|}(\sCurl_{1,1} \xi)_{B)F}
 + \tfrac{6}{5} \kappa_{(A}{}^{C}\xi^{DA'}(\sCurl_{2,2} \Phi)_{B)CDA'}\nonumber\\
& + \tfrac{3}{5} \kappa^{CD}\xi_{(A}{}^{A'}(\sCurl_{2,2} \Phi)_{B)CDA'}
 - 2 \kappa^{DF} \xi^{CA'} (\sTwist_{4,0} \Psi)_{ABCDFA'}.\label{eq:auxcondDirac2factored}
\end{align}
Using the technique from Section~\ref{sec:integrabilitycond} we get that the integrability conditions for \eqref{eq:TwistkappaTwistxi} are
\begin{subequations}
\begin{align}
0={}&\Psi_{(ABC}{}^{F}\kappa_{D)F},\label{eq:intcondkappaDirac2}\\
0={}&\tfrac{1}{2} \Psi_{ABCD} (\sDiv_{1,1} \xi)
 + 2 \Psi_{(ABC}{}^{F}(\sCurl_{1,1} \xi)_{D)F}
 -  \tfrac{4}{5} \xi_{(A}{}^{A'}(\sCurl_{2,2} \Phi)_{BCD)A'}
 + \xi^{FA'} (\sTwist_{4,0} \Psi)_{ABCDFA'}.\label{eq:intcondxiDirac2}
\end{align}
\end{subequations}
Applying the operator $\sCurlDagger$ on the condition \eqref{eq:intcondkappaDirac2} gives
\begin{align}
0={}&- \tfrac{1}{2} \Psi_{ABCD} (\sCurlDagger_{2,0} \kappa)^{D}{}_{A'}
 -  \tfrac{9}{10} \kappa_{(A}{}^{D}(\sCurl_{2,2} \Phi)_{BC)DA'}
 + \tfrac{1}{4} \kappa^{DF} (\sTwist_{4,0} \Psi)_{ABCDFA'}.\label{eq:curlintcondkappaDirac2}
\end{align}
Using \eqref{eq:intcondxiDirac2} to elliminate $\Psi_{ABCD} (\sDiv_{1,1} \xi)$ and \eqref{eq:curlintcondkappaDirac2} to elliminate $\kappa^{DF} (\sTwist_{4,0} \Psi)_{ABCDFA'}$, and doing an irreducible decomposition of $\Psi_{ABCF} \kappa_{D}{}^{F}$ we see that \eqref{eq:auxcondDirac2factored} reduces to
\begin{align}
0={}&-2 (\sCurl_{1,1} \xi)^{CD} \Psi_{(ABC}{}^{F}\kappa_{D)F},
\end{align}
which is trivially satisfied due to \eqref{eq:intcondkappaDirac2}.
\end{proof}

\subsection{Factorization of valence $(4,0)$ Killing spinors with aligned matter}\label{sec:factorValence4}
\begin{proof}[Proof of Theorem~\ref{thm:Valence4Factorization}]
Assume that the matter field and the curvature are aligned, that is
\begin{align}
0={}&\Psi_{(ABC}{}^{F}\Phi_{D)FA'B'}.\label{eq:AlignmentPsiPhi}
\end{align}
Furthermore, assume that $\Psi_{ABCD}$ does not vanish, and assume that there is a solution $L_{ABCD}$ to
\begin{equation}
(\sTwist_{4,0}L)_{ABCDEA'} = 0.\label{eq:KS4}
\end{equation}
The integrability condition \eqref{eq:intcondKS4} for this equation together with the non-vanishing of the Weyl spinor, gives that $L_{ABCD}$ and $\Psi_{ABCD}$ are proportional  (c.f. \cite{KalMcLWil92a}). This means that
\begin{subequations}
\begin{align}
0={}&L_{(ABC}{}^{F}\Phi_{D)FA'B'},\label{eq:LPhiAlign}\\
0={}&- L_{(ABCD}(\sTwist_{0,0}\Lambda)_{F)A'} + L_{(ABC}{}^{H}(\sCurl_{2,2} \Phi)_{DF)HA'}
 + \tfrac{1}{5} \Phi_{(AB|A'|}{}^{B'}(\sCurlDagger_{4,0} L)_{CDF)B'},\label{eq:DerLPhiAlign}
\end{align}
\end{subequations}
where the second equation is obtained by taking a derivative of the first, decomposing the derivatives into irreducible parts, using the Killing spinor equation, and symmetrizing over all unprimed indices.

Split
$L_{ABCD}$ into principal spinors
$L_{ABCD}=\alpha_{(A}\beta_B\gamma_C\delta_{D)}$.
Now, the Killing spinor equation \eqref{eq:KS4}, and the alignment equation \eqref{eq:LPhiAlign} gives
\begin{subequations}
\begin{align}
0={}&\alpha^{A} \alpha^{B} \alpha^{C} \alpha^{D} \alpha^{F} (\sTwist_{4,0} L)_{ABCDFA'}=
\alpha^{A} \beta_{A} \alpha^{B} \gamma_{B} \alpha^{C} \delta_{C} \alpha^{D} \alpha^{F} \nabla_{FA'}\alpha_{D},\\
0={}&\alpha^{A} \alpha^{B} \alpha^{C} \alpha^{D} L_{(ABC}{}^{F}\Phi_{D)FA'B'}=
\tfrac{1}{4}  \alpha^{A} \beta_{A} \alpha^{B} \gamma_{B} \alpha^{C} \delta_{C} \alpha^{D} \alpha^{F}\Phi_{DFA'B'}.
\end{align}
\end{subequations}
We will first assume that $\alpha^A$ is not a repeated principal spinor of $L_{ABCD}$. This means that	$\alpha^A\beta_A\alpha^B\gamma_B\alpha^C\delta_{C} \neq 0$ and hence $\alpha^A\alpha^B\nabla_{A'A}\alpha_{B}=0$, that is $\alpha_A$ is a
shear-free geodesic null congruence. We also get $\alpha^{D} \alpha^{F}\Phi_{DFA'B'}=0$.
Contracting \eqref{eq:DerLPhiAlign} with $\alpha^A\alpha^B\alpha^C\alpha^D\alpha^F$ we get
\begin{align}
0 ={}& \tfrac{1}{4} \alpha^{A} \beta_{A}\alpha^{B} \gamma_{B}\alpha^{C} \delta_{C} \alpha^{D} \alpha^{F} \alpha^{H} (\sCurl_{2,2} \Phi)_{DFHA'} + \tfrac{1}{5} \Phi_{ABA'}{}^{B'} \alpha^{A} \alpha^{B} \alpha^{C} \alpha^{D} \alpha^{F} (\sCurlDagger_{4,0} L)_{CDFB'}\nonumber\\
={}&\tfrac{1}{4} \alpha^{A} \beta_{A}\alpha^{B} \gamma_{B}\alpha^{C} \delta_{C} \alpha^{D} \alpha^{F} \alpha^{H} (\sCurl_{2,2} \Phi)_{DFHA'}.
\end{align}
Hence, $\alpha^{A} \alpha^{B} \alpha^{C} (\sCurl_{2,2} \Phi)_{ABCA'}=0$.
But the Bianchi equations give
\begin{align}
\alpha^{A} \alpha^{B} \alpha^{C} \nabla^{DD'}\Psi_{ABCD}={}&\alpha^{A} \alpha^{B} \alpha^{C} (\sCurl_{2,2} \Phi)_{ABCA'}=0.
\end{align}
It follows from the generalized Goldberg-Sachs theorem that $\alpha^A$ is a repeated principal spinor of $\Psi_{ABCD}$, see for instance \cite[Proposition 7.3.35]{PenRin86}. But $L_{ABCD}$ and $\Psi_{ABCD}$ are proportional, so $\alpha^A$ is a repeated principal spinor of $L_{ABCD}$ after all. Without loss of generality, we can assume that $\gamma^A=\alpha^A$, a relabelling  and rescaling of $\beta^{A}$, $\gamma^{A}$ and $\delta^{A}$ can achieve this. Repeating the argument with $\beta^A$, we find that also $\beta^A$ is a repeated principal spinor of $L_{ABCD}$. If $\beta^A\alpha_A=0$, we can repeat the argument again with $\delta^A$ and see that all principal spinors are repeated, i.e. Petrov type N. Otherwise, we have Petrov type D.  In conclusion, we have after rescaling $L_{ABCD}=\alpha_{(A}\alpha_{B}\beta_{C}\beta_{D)}$. Now, let $\kappa_{AB}=\alpha_{(A}\beta_{B)}$.

First assume that $\alpha^A\beta_A\neq 0$. Contracting \eqref{eq:KS4} with $\alpha^{A} \alpha^{B} \alpha^{C} \alpha^{D} \beta^{F}$, $\alpha^{A} \alpha^{B} \alpha^{C} \beta^{D} \beta^{F}$, $\alpha^{A} \alpha^{B} \beta^{C} \beta^{D} \beta^{F}$, $\alpha^{A} \beta^{B} \beta^{C} \beta^{D} \beta^{F}$ we find
\begin{align*}
0={}&\alpha^{A} \alpha^{B} \alpha^{C} (\sTwist_{2,0} \kappa)_{ABCA'},&
0={}&\alpha^{A} \alpha^{B} \beta^{C} (\sTwist_{2,0} \kappa)_{ABCA'},\\
0={}&\alpha^{A} \beta^{B} \beta^{C} (\sTwist_{2,0} \kappa)_{ABCA'},&
0={}&\beta^{A} \beta^{B} \beta^{C} (\sTwist_{2,0} \kappa)_{ABCA'}.
\end{align*}
Hence, $(\sTwist_{2,0} \kappa)_{ABCA'}=0$.

If $\alpha^A\beta_A= 0$, we can find a dyad $(o^A, \iota^A)$ so that $\alpha^A=o^A$. Then we have $L_{ABCD}=\upsilon^2 o_{A}o_{B}o_{C}o_{D}$ and $\kappa_{AB}=\upsilon o_{A}o_{B}$.
Contracting \eqref{eq:KS4} with $o^{A} o^{B} o^{C} \iota^{D} \iota^{F}\upsilon^{-1}$, $o^{A} o^{B} \iota^{C} \iota^{D} \iota^{F}\upsilon^{-1}$, $o^{A} \iota^{B} \iota^{C} \iota^{D} \iota^{F}\upsilon^{-1}$, $\iota^{A} \iota^{B} \iota^{C} \iota^{D} \iota^{F}\upsilon^{-1}$ we find
\begin{align*}
0={}&o^{A} o^{B} o^{C} (\sTwist_{2,0} \kappa)_{ABCA'},&
0={}&o^{A} o^{B} \iota^{C} (\sTwist_{2,0} \kappa)_{ABCA'},\\
0={}&o^{A} \iota^{B} \iota^{C} (\sTwist_{2,0} \kappa)_{ABCA'},&
0={}&\iota^{A} \iota^{B} \iota^{C} (\sTwist_{2,0} \kappa)_{ABCA'}.
\end{align*}
Hence, $(\sTwist_{2,0} \kappa)_{ABCA'}=0$.

We can therefore conclude that if the curvature satisfies \eqref{eq:AlignmentPsiPhi}, $\Psi_{ABCD}$ does not vanish, and we have a valence $(4,0)$ Killing spinor $L_{ABCD}$, then we have a valence $(2,0)$ Killing spinor $\kappa_{AB}$ such that $L_{ABCD}=\kappa_{(AB}\kappa_{CD)}$.
\end{proof}

\section{The symmetry operators with factorized Killing spinor}\label{sec:symopfactored}

\subsection{Symmetry operators for the conformal wave equation}
Let us now consider special cases of symmetry operators for the conformal wave equation. 
If we choose
\begin{align}
L_{ABA'B'}={}&\mathfrak{L}_{\xi\zeta ABA'B'},&
P_{AA'}={}&0,&
Q ={}& \tfrac{2}{5} \mathfrak{Q}_{\xi\zeta}.
\end{align}
Then the operator takes the form
\begin{align}
\chi={}&\tfrac{1}{2} \hat{\mathcal{L}}_{\zeta}\hat{\mathcal{L}}_{\xi}\phi
 + \tfrac{1}{2} \hat{\mathcal{L}}_{\xi}\hat{\mathcal{L}}_{\zeta}\phi.
\end{align}
One can also add an arbitrary first order symmetry operator to this.

We can also choose
\begin{align}
L_{ABA'B'}={}&\mathfrak{L}_{\kappa ABA'B'},&
P_{AA'}={}&0,&
Q ={}& \tfrac{2}{5} \mathfrak{Q}_{\kappa}.
\end{align}
Substituting these expressions into \eqref{eq:wavesymop1} gives a symmetry operator, but we have not found any simpler form than the one given by \eqref{eq:wavesymop1}.

\begin{remark}
Apart from factorizations, one can in special cases get symmetry operators from Killing tensors. 
If $K_{AB}{}^{A'B'}$ is a Killing tensor, then we have 
\begin{align*}
(\sTwist_{2,2} L)_{ABCA'B'C'}={}&0, &
(\sDiv_{2,2} L)_{AA'} ={}& - \tfrac{3}{4} (\sTwist_{0,0} S)_{AA'},&
K^{ABA'B'}={}&L^{ABA'B'} + \tfrac{1}{4} S \epsilon^{AB} \bar\epsilon^{A'B'}
\end{align*}
where $L_{AB}{}^{A'B'}=K_{(AB)}{}^{(A'B')}$ and $S=K_{A}{}^{A}{}_{A'}{}^{A'}$. The commutator \eqref{eq:CurlTwist} gives $(\sCurl_{1,1} \sDiv_{2,2} L)_{AB} = 0$.
If we also assume vacuum, then the equation \eqref{eq:CurlDaggerCurlDivLL22Dirac} gives
\begin{align}
(\sTwist_{0,0} \sDiv_{1,1} \sDiv_{2,2} L)_{AA'}={}&-2 \Psi_{ABCD} (\sCurl_{2,2} L)^{BCD}{}_{A'}
 - 2 \bar\Psi_{A'B'C'D'} (\sCurlDagger_{2,2} L)_{A}{}^{B'C'D'}.
\end{align}
Hence, we can choose 
\begin{align}
Q ={}& - \tfrac{1}{15} (\sDiv_{1,1} \sDiv_{2,2} L),
\end{align}
to satisfy condition \ref{point:A0}, and get the well known symmetry operator
\begin{align}
\chi={}&-\tfrac{1}{2} (\sTwist_{0,0} S)^{AA'} (\sTwist_{0,0} \phi)_{AA'}
 + L^{ABA'B'} (\sTwist_{1,1} \sTwist_{0,0} \phi)_{ABA'B'}&
={}&\nabla_{AA'}(K^{ABA'B'} \nabla_{BB'}\phi),
\end{align}
which is valid for vacuum spacetimes.
\end{remark}

\subsection{Symmetry operator of the first kind for the Dirac-Weyl equation}
Let us now consider special cases of symmetry operators of the first kind for the Dirac-Weyl equation. 
We can choose
\begin{align}
L_{ABA'B'}={}&\mathfrak{L}_{\xi\zeta ABA'B'},&
P^{AA'}={}&- \tfrac{1}{3} \mathfrak{P}_{\xi\zeta}{}^{AA'},&
Q ={}& \tfrac{3}{10} \mathfrak{Q}_{\xi\zeta},
\end{align}
to get a symmetry operator for the Dirac-Weyl equation.
The operator then becomes
\begin{align}
\chi_{A}={}&\tfrac{1}{2}\hat{\mathcal{L}}_{\xi}\hat{\mathcal{L}}_{\zeta}\phi_{A}+\tfrac{1}{2}\hat{\mathcal{L}}_{\zeta}\hat{\mathcal{L}}_{\xi}\phi_{A}.
\end{align}
We can add any conformal Killing vector to $P^{AA'}$ and any constant to $Q$.
Note that if we add the conformal Killing vector
$\tfrac{1}{2} (\xi^{BB'} \nabla_{BB'}\zeta^{AA'} -  \zeta^{BB'} \nabla_{BB'}\xi^{AA'})$
to $P^{AA'}$, the operator gets the factored form
\begin{align}
\chi_{A}={}&\hat{\mathcal{L}}_{\xi}\hat{\mathcal{L}}_{\zeta}\phi_{A}.
\end{align}

We can also choose
\begin{align}
L_{ABA'B'}={}&\mathfrak{L}_{\kappa ABA'B'},&
P^{AA'}={}&- \tfrac{1}{3} \mathfrak{P}_{\kappa}{}^{AA'},&
Q ={}& \tfrac{3}{10} \mathfrak{Q}_{\kappa}.
\end{align}
Substituting these expressions into \eqref{eq:diracsymop1} gives a symmetry operator, but we have not found any simpler form than the one given by \eqref{eq:diracsymop1}.

\subsection{Symmetry operator of the first kind for the Maxwell equation}\label{sec:symopfirstmaxwellfact}
Let us now consider the symmetry operators of the first kind for the Maxwell equation. 
Let
\begin{align}
L_{ABA'B'}={}&\mathfrak{L}_{\xi\zeta ABA'B'},&
P^{AA'}={}&- \tfrac{2}{3} \mathfrak{P}_{\xi\zeta}{}^{AA'},&
Q ={}&0,
\end{align}
to get a symmetry operator.
With this choice the symmetry operator and the potential reduce to
\begin{subequations}
\begin{align}
\chi_{AB}={}&\tfrac{1}{2} \hat{\mathcal{L}}_{\zeta}\hat{\mathcal{L}}_{\xi}\phi_{AB}
 + \tfrac{1}{2} \hat{\mathcal{L}}_{\xi}\hat{\mathcal{L}}_{\zeta}\phi_{AB},\\
A_{AA'}={}&- \tfrac{1}{2} \zeta^{B}{}_{A'}\hat{\mathcal{L}}_{\xi}\phi_{AB} 
 -  \tfrac{1}{2} \xi^{B}{}_{A'}\hat{\mathcal{L}}_{\zeta}\phi_{AB} .
\end{align}
\end{subequations}
A general first order operator can be added to this. If we add an the same commutator as above with an appropriate coefficient to $P^{AA'}$, we get the same kind of factorization of the operator as above.

We can also get a solution by setting
\begin{align}
L_{ABA'B'}={}&\mathfrak{L}_{\kappa ABA'B'},&
P^{AA'}={}&- \tfrac{2}{3} \mathfrak{P}_{\kappa}{}^{AA'},&
Q ={}&0,
\end{align}
With this choice the symmetry operator and the potential reduce to
\begin{subequations}
\begin{align}
\chi_{AB}={}&(\sCurl_{1,1} A)_{AB},\\
A_{AA'}={}&- \tfrac{1}{3} \Theta_{AB} (\sCurl_{0,2} \bar{\kappa})^{B}{}_{A'}
 + \bar{\kappa}_{A'B'} (\sCurlDagger_{2,0} \Theta)_{A}{}^{B'},\\
\Theta_{AB}\equiv{}&-2 \kappa_{(A}{}^{C}\phi_{B)C}.
\end{align}
\end{subequations}
This proves the first part of Theorem~\ref{Thm:SymopMaxwellSimple}.

\subsection{Symmetry operator of the second kind for the Dirac-Weyl equation}\label{sec:symopsecondDiracfactored}
Let
\begin{subequations}
\begin{align}
L_{ABC}{}^{A'}={}&\mathfrak{L}_{\kappa\xi ABC}{}^{A'},\\
P_{AB}
={}&- \tfrac{1}{2} \hat{\mathcal{L}}_{\xi}\kappa_{AB}
 + \tfrac{3}{8} \kappa_{AB} (\sDiv_{1,1} \xi)\nonumber\\
={}&\tfrac{1}{8} \kappa_{AB} (\sDiv_{1,1} \xi)
 -  \tfrac{1}{2} \kappa_{(A}{}^{C}(\sCurl_{1,1} \xi)_{B)C}
 + \tfrac{1}{3} \xi_{(A}{}^{A'}(\sCurlDagger_{2,0} \kappa)_{B)A'}.
\end{align}
\end{subequations}
Using the equations \eqref{eq:TwistkappaTwistxi}, the commutators \eqref{eq:DivTwistCurlCurlDagger}, \eqref{eq:DivTwistCurlDaggerCurl}, \eqref{eq:CurlTwist}, \eqref{eq:CurlDaggerTwist} and the irreducible decompositions of $\Psi_{ABCF} \kappa_{D}{}^{F}$ and $\Phi_{ABA'B'} \xi_{C}{}^{B'}$ we get
\begin{align}
(\sTwist_{2,0} P)_{ABCA'}={}&- \tfrac{1}{6} \kappa_{(AB}(\sCurlDagger_{2,0} \sCurl_{1,1} \xi)_{C)A'}
 -  \tfrac{1}{2} \kappa_{(A}{}^{D}(\sTwist_{2,0} \sCurl_{1,1} \xi)_{BC)DA'}
 + \tfrac{1}{8} \kappa_{(AB}(\sTwist_{0,0} \sDiv_{1,1} \xi)_{C)A'}\nonumber\\
& + \tfrac{1}{6} \xi_{(A|A'|}(\sCurl_{1,1} \sCurlDagger_{2,0} \kappa)_{BC)}
 + \tfrac{1}{3} \xi_{(A}{}^{B'}(\sTwist_{1,1} \sCurlDagger_{2,0} \kappa)_{BC)A'B'}\nonumber\\
={}&\xi^{D}{}_{A'} \Psi_{(ABC}{}^{F}\kappa_{D)F}\nonumber\\
={}&0,
\end{align}
where we in the last step used the integrability condition \eqref{eq:intcondxiDirac2}.
Observe that $P_{AB}$ is given by a conformally weighted Lie derivative, but now with a different weight. The operator $\hat{\mathcal{L}}_{\xi}$ has a conformal weight adapted to the weight of the conformally invariant operator $\sCurlDagger$. The operator $\sTwist$ is also conformally invariant, but with a different weight. This explains the extra term in $P_{AB}$.

The symmetry operator of the second kind for the Dirac-Weyl equation now takes the form
\begin{align}
\omega_{A'}={}&\kappa^{BC} (\sTwist_{1,0} \hat{\mathcal{L}}_{\xi}\phi)_{BCA'}-\tfrac{2}{3} \hat{\mathcal{L}}_{\xi}\phi_{B} (\sCurlDagger_{2,0} \kappa)^{B}{}_{A'}.
\end{align}
Hence, we can conclude that if $L_{ABCA'}$ factors, then one can choose a corresponding $P_{AB}$ so that the operator factors as a first order symmetry operator of the first kind followed by a first order symmetry operator of the second kind.

\subsection{Symmetry operator of the second kind for the Maxwell equation}\label{sec:symopsecondmaxwellfact}
If we let $L_{ABCD}=\kappa_{(AB}\kappa_{CD)}$ with
\begin{align}
(\sTwist_{2,0} \kappa)_{ABCA'} ={}& 0,
\end{align}
Then the operator if the second kind now takes the form 
\begin{subequations}
\begin{align}
\omega_{A'B'}={}&(\sCurlDagger_{1,1}B)_{A'B'},\\
B_{AA'}={}&\kappa_{AB} (\sCurlDagger_{2,0} \Theta)^{B}{}_{A'}
 + \tfrac{1}{3} \Theta_{AB} (\sCurlDagger_{2,0} \kappa)^{B}{}_{A'},\\
\Theta_{AB}\equiv{}&-2 \kappa_{(A}{}^{C}\phi_{B)C}.
\end{align}
\end{subequations}
This proves the second part of Theorem~\ref{Thm:SymopMaxwellSimple}.

\section*{Acknowledgements}
The authors would like to thank Steffen Aksteiner and Lionel Mason for helpful discussions. We are particularly grateful to Lionel Mason for his ideas concerning Theorem~\ref{thm:Valence4Factorization}. Furthermore we would like to thank Niky Kamran and J.~P. Michel for helpful comments. LA thanks Shing-Tung Yau for generous hospitality and many interesting discussions on symmetry operators and related matters, during a visit to Harvard University, where some initial work on the topic of this paper was done. 
This material is based upon work supported by the National Science Foundation under Grant No. 0932078 000, while the authors were in residence at the Mathematical Sciences Research Institute in Berkeley, California, during the semester of 2013.

\section*{References}

\newcommand{\prd}{Phys. Rev. D} 

\providecommand{\newblock}{}

\end{document}